%% file: main.tex
\newif\ifanon
\anonfalse

\documentclass{article}
\usepackage[T1]{fontenc}

\usepackage[normalem]{ulem}
\usepackage{stackengine}
\usepackage{float}
\usepackage[utf8]{inputenc}
\usepackage{fullpage}
\usepackage{amsmath}
\usepackage{amssymb,amsfonts,dsfont}
\usepackage{amsthm}
\usepackage{authblk}
\usepackage{thm-restate}
\usepackage{dsfont}
\usepackage{physics}
\usepackage{tikz}
\usetikzlibrary{calc,decorations.pathreplacing,fit,shapes.misc}
\usepackage{array}
\usepackage{comment}
\usepackage{enumitem}
\usepackage[backend=bibtex,style=alphabetic,maxbibnames=99, minalphanames=4,maxalphanames=5]{biblatex}
\usepackage{authblk}

\usepackage{todonotes}
\usepackage{lmodern}
\usetikzlibrary{decorations.pathmorphing}
\usepackage[colorlinks=true,linkcolor=blue,allcolors=blue]{hyperref}
\usepackage[capitalize,nameinlink]{cleveref}
\usepackage{algorithm2e,setspace}
\usepackage{caption}

\tikzset{snake it/.style={decorate, decoration=snake}}
\makeatletter
\tikzset{
  fitting node/.style={
    inner sep=0pt,
    fill=none,
    draw=none,
    reset transform,
    fit={(\pgf@pathminx,\pgf@pathminy) (\pgf@pathmaxx,\pgf@pathmaxy)}
  },
  reset transform/.code={\pgftransformreset}
}
\makeatother

\renewcommand{\P}{\mathsf{P}}
\DeclareMathOperator{\BPP}{\mathsf{BPP}}
\DeclareMathOperator{\BQP}{\mathsf{BQP}}

\DeclareMathOperator{\MA}{\mathsf{MA}}
\DeclareMathOperator{\NP}{\mathsf{NP}}
\DeclareMathOperator{\RP}{\mathsf{RP}}

\DeclareMathOperator{\PP}{\mathsf{PP}}

\DeclareMathOperator{\PSPACE}{\mathsf{PSPACE}}
\DeclareMathOperator{\EXP}{\mathsf{EXP}}

\DeclareMathOperator{\Ptime}{\mathsf{P}}

\newcommand{\QMA}{\mathsf{QMA}}
\DeclareMathOperator{\QCMA}{\mathsf{QCMA}}

\DeclareMathOperator{\coQMA}{\mathsf{coQMA}}
\DeclareMathOperator{\LH}{\mathsf{LH}}

\DeclareMathOperator{\QXC}{\mathsf{QXC}}

\newcommand{\poly}{\operatorname{poly}}

\newcommand{\id}{\mathrm{I}}

\newcommand{\eps}{\varepsilon}
\newcommand{\expec}[1]{\mathbb{#1}}

\newtheorem{theorem}{Theorem}[section]

\newtheorem{lemma}[theorem]{Lemma}

\newtheorem{prop}[theorem]{Proposition}

\newtheorem{fact}[theorem]{Fact}
\newtheorem{claim}[theorem]{Claim}
\newtheorem{conjecture}[theorem]{Conjecture}
\newtheorem{corollary}[theorem]{Corollary}
\newtheorem{task}[theorem]{Task}
\newtheorem{hypothesis}[theorem]{Hypothesis}

\newtheorem{assumption}{Assumption}
\crefname{assumption}{assumption}{assumptions}

\theoremstyle{definition}

\newtheorem{definition}[theorem]{Definition}

\theoremstyle{remark}
\newtheorem{remark}[theorem]{Remark}

\newtheorem{property}{Property}
\crefname{property}{Property}{Properties}

\crefname{algocf}{Algorithm}{Algorithms}

\newcommand{\bs}[1]{\boldsymbol{#1}}
\newcommand{\bg}{\boldsymbol{g}}
\newcommand{\mc}{\mathcal}
\newcommand{\cpoly}{\textsf{poly}}
\renewcommand{\tilde}{\widetilde}
\renewcommand{\bar}{\overline}

\newcommand{\Qquad}{\qquad\qquad}

\newcommand{\Paren}[1]{\left( #1 \right)}
\newcommand{\Brac}[1]{\left[ #1 \right]}

\newcommand{\round}[1]{\lfloor #1 \rceil}

\DeclareMathOperator*{\E}{\mathbb{E}}
\newcommand{\st}{\,:\,}
\newcommand{\dfn}{\triangleq}

\newcommand{\Norm}[1]{\left\| #1 \right\|}

\newcommand{\A}{\mathcal A}
\newcommand{\C}{\mathbb{C}}
\newcommand{\R}{\mathbb{R}}

\newcommand{\Id}{\mathbb I}
\renewcommand{\O}{\mathcal O}

\renewcommand{\i}{\iota}

\newcommand{\UNP}{\mathsf{UniqueNP}}
\newcommand{\UQCMA}{\mathsf{UniqueQCMA}}
\newcommand{\UQMA}{\mathsf{UniqueQMA}}

\newcommand{\ULH}{\mathsf{UniqueLH}}

\newcommand{\sket}[1]{|#1\rangle}

\newcommand{\ketbratwo}[2]{\left| #1 \middle\rangle \middle\langle #2 \right| }

\newcommand{\tout}{\textup{out}}
\newcommand{\tin}{\textup{in}}
\newcommand{\taux}{\textup{aux}}
\newcommand{\tsys}{\textup{sys}}

\newcommand{\tcon}{\textup{control}}
\newcommand{\tacc}{\textup{accept}}

\newcommand{\QApxDim}{\textsf{QApxDim}}
\newcommand{\oprob}{\QApxDim(k_1, k_2,N)}
\newcommand{\ApxDim}{\textsf{ApxDim}}
\newcommand{\reg}[1]{\textsf{#1}}
\newcommand{\regt}[1]{\reg R_{#1}}
\newcommand{\ttC}{\texttt{C}}
\newcommand{\haar}{\textup{Haar}}

\renewcommand{\S}{S}
\newcommand{\qS}{\mathcal{S}}
\newcommand{\T}{T}
\newcommand{\qT}{\mathcal{T}}

\newcounter{relctr} %% <- counter for relations
\newcommand\labelrel[2]{%
  \begingroup
    \refstepcounter{relctr}%
    \stackrel{\textnormal{(\alph{relctr})}}{\mathstrut{#1}}%
    \originallabel{#2}%
  \endgroup
}
\AtBeginDocument{\let\originallabel\label}

\newcommand\numberthis{\addtocounter{equation}{1}\tag{\theequation}}

\usepackage[most,breakable]{tcolorbox}
\usepackage{changepage}
\usepackage[final,nopatch=footnote]{microtype}
\newtcolorbox{myframe}[1][]{
  enhanced,
  arc=0pt,
  outer arc=0pt,
  colback=white,
  boxrule=0.8pt,
  #1
}
\newtcolorbox{algbox}{breakable,colback=gray,colframe=gray,standard jigsaw,opacityback=0.15,opacityframe=0.5,boxrule=0.75pt,before upper={\parindent15pt}}

\newcommand{\trdd}{\Tr(\Pi_{\Delta})}

\newcommand{\yw}[1]{\todo[inline, color=blue!30]{Yeongwoo: #1}}
\newcommand{\ywside}[1]{\todo[color=blue!30,size=\small]{Yeongwoo: #1}}

\newcommand{\qn}[1]{\todo[inline, color=brown!30]{Quynh: #1}}

\title{On the complexity of unique quantum witnesses \\ and quantum approximate counting}

\ifanon
\author{}
\date{}
\else
\author[1]{Anurag Anshu}
\author[2]{Jonas Haferkamp}
\author[1]{Yeongwoo Hwang}
\author[1]{Quynh T. Nguyen}

\affil[1]{Harvard University}
\affil[2]{Saarland University}
\fi

\date{\today}

\begin{document}
\captionsetup{width=.9\linewidth}

\maketitle

\begin{abstract}
We study the long-standing open question on the power of unique witnesses in quantum protocols, which asks if $\UQMA$, a variant of $\QMA$ whose accepting witness space is 1-dimensional, contains $\QMA$ under quantum reductions.

This work rules out any black-box reduction from $\QMA$ to $\UQMA$ by showing a quantum oracle separation between $\BQP^\UQMA$ and $\QMA$. This provides a contrast to the classical case, where the Valiant-Vazirani theorem shows a black-box randomized reduction from $\UNP$ to $\NP$, and suggests the need for studying the structure of the ground space of local Hamiltonians in distilling a potential unique witness. Via similar techniques, we show, relative to a quantum oracle, that $\QMA^\QMA$ cannot decide quantum approximate counting, ruling out a quantum analogue of Stockmeyer's algorithm in the black-box setting. Our results employ a subspace reflection oracle, previously considered in \cite{Aaronson_AK2007_QuantumClassicalProofs, Aaronson_AKKT2020_QuantumLowerBounds, she2023unitary}, but we introduce new tools which allow us to exploit the unique witness constraint. We also show a strong ``polarization'' behavior of $\QMA$ circuits, which could be of independent interest in studying quantum polynomial hierarchies.

We then ask a natural question; what structural properties of the local Hamiltonian problem can we exploit? We introduce a physically motivated candidate by showing that the ground energy of local Hamiltonians that satisfy a computational variant of the eigenstate thermalization hypothesis (ETH) can be estimated through a $\UQMA$ protocol. Our protocol can be viewed as a quantum expander test in a low energy subspace of the Hamiltonian and verifies a unique entangled state across two copies of the subspace. This allows us to conclude that if $\UQMA$ is not equivalent to $\QMA$, then $\QMA$-hard Hamiltonians must violate ETH under adversarial perturbations (more accurately, further assuming the quantum PCP conjecture if ETH only applies to extensive energy subspaces). Under the same assumption, this also serves as evidence that chaotic local Hamiltonians, such as the SYK model may be computationally simpler than general local Hamiltonians.

\end{abstract}

\newpage
\tableofcontents

\newpage
\section{Introduction}

One of the central aims of quantum complexity theory is to explore the structure of ground, and more generally low-energy, subspaces in quantum many-body systems. A key question in this field is whether these low-energy subspaces of local Hamiltonians are computationally simpler than generic quantum subspaces.\footnote{From an information-theoretic point of view, the ground space of a local Hamiltonian is known to be simpler as it has a polynomial-sized tensor network approximation \cite{PhysRevLett.98.140506}. However this tensor network is not known to be computationally tractable.} A long-standing problem which formalizes this inquiry is the relationship between $\UQMA$ and $\QMA$~\cite{Aharonov_ABBS2022_PursuitUniquenessExtending}. The complexity class $\UQMA$ captures verification tasks in which positive instances have a 1-dimensional accepting subspace (unique quantum witness) while $\QMA$ places no restriction on the dimension of the accepting subspace. The open question is to determine the relationship between these classes: are they equivalent under quantum reductions, or does the uniqueness promise strictly reduce computational power? On one hand, an equivalence would imply that a polynomial-time quantum algorithm can \emph{uniquely} verify some fixed state in a low-energy subspace of local Hamiltonians, a task that intuitively seems unattainable for generic quantum subspaces. On the other, a separation would formalize the physics intuition that ``natural'' Hamiltonians with a non-degenerate ground space and inverse polynomial spectral gap are computationally simpler than $\QMA$-hard Hamiltonians and thus amenable to analysis.

This question is also motivated by the insights produced by studying the class $\UNP$. Classically, the celebrated Valiant-Vazirani isolation lemma~\cite{Valiant_VV1986_NPEasyDetecting} was used to show that $\UNP$ contains $\NP$ under \emph{randomized} reductions, i.e., $\NP \subseteq \RP^{\UNP}$ (where $\RP$ is a one-sided error version of $\BPP$), suggesting that the uniqueness promise does not significantly reduce the computational power of $\NP$, at least under randomized reductions. The isolation lemma has underlaid various developments in theoretical computer science, including approximate counting~\cite{stockmeyer1983complexity}, Toda's theorem~\cite{toda1991pp}, parallel computation~\cite{Mulmuley1987}, one-way functions~\cite{grollmann1988complexity}, and average-case complexity \cite{ben1989theory}. One could hope that studying the possibility of a ``quantum isolation lemma'' would sharpen our understanding in a similarly wide range of quantum problems.

\paragraph{Limitations of unique quantum witnesses.} There has been a series of works attempting to resolve this question. Aharonov et al.~\cite{Aharonov_ABBS2022_PursuitUniquenessExtending} (first appearing on arxiv in 2008) initiated the study of $\UQMA$ and possible quantum analogues of the isolation lemma. There, the authors extended the hashing-based arguments in the Valiant-Vazirani isolation lemma to the classes $\MA$ and $\QCMA$. However, they noted the failure of a similar strategy for $\QMA$, citing the fact that two random quantum states have exponentially small overlap. A later work~\cite{Jain_JKK+2011_PowerUniqueQuantum} sought a completely different approach to obtain a \emph{deterministic} isolation procedure in a special case, but also fell short of achieving a unique witness protocol for general $\QMA$ problems as their strategy worked only for accepting subspaces of polynomial dimension. This led~\cite{Jain_JKK+2011_PowerUniqueQuantum} to conjecture that a \emph{quantum} reduction is required to obtain a unique verifier, suggesting the question of whether the inclusion $\QMA {\subseteq} \BQP^{\UQMA}$ holds.\footnote{In this work, is implicit that the oracularized complexity classes refers to the promise version (i.e. $\UQMA$ actually refers to $\mathsf{PromiseUniqueQMA}$). The precise definition of these oracularized classes requires some care, see~\Cref{sec:prelim}.} This discussion brings us to our first result.

\begin{theorem}[\Cref{cor:UQMA-quantum-oracle-sep} restated]
\label{thm:informaloraclesep}
There exists a \emph{quantum} oracle relative to which $\QMA \not\subseteq \BQP^{\UQMA}$.
\end{theorem}

This oracle separation implies there is no black-box quantum isolation procedure for quantum witnesses, unlike its classical counterpart~\cite{Valiant_VV1986_NPEasyDetecting}, and explains the failures of all approaches in~\cite{Aharonov_ABBS2022_PursuitUniquenessExtending, Jain_JKK+2011_PowerUniqueQuantum}. It may also be viewed as some evidence for the inequality between $\UQMA$ and $\QMA$. Furthermore, it suggests the need to exploit the physical structure of quantum low-energy spaces to construct a $\UQMA$ verifier for local Hamiltonians (see our result below).

Our~\Cref{thm:informaloraclesep} is based on a quantum oracle reflecting about an unknown subspace, and the task is to determine whether the subspace is non-empty (of some large dimensionality).  
 Although this result in a sense is a continuation of a series of works on quantum oracle separations and unitary property testing~\cite{Aaronson_AK2007_QuantumClassicalProofs, irani2021quantum, she2023unitary, agarwal2024oracle},
the techniques used in our work differ in that they make careful use of the structure of the quantum witness. Previous separations used techniques which were not sensitive to the presence of a witness (e.g. replacing the witness with the maximally mixed state), which would undermine the  uniqueness promise in our case. We hope that the techniques presented in this work will help develop similar query lower bounds on other witness-restricted quantum classes, such as $\QMA(2)$ where constructing oracle separations have been challenging.

Along the way to \Cref{thm:informaloraclesep} we also obtain a \emph{classical} oracle separation between $\UQMA$ and $\QMA$ (\Cref{cor:oracle_separation}). In fact, we can even separate $\NP$ from $\P^{\UQMA}$ under the same classical oracle. However, this separation is perhaps not as interesting as~\Cref{thm:informaloraclesep}, for one would a priori expect randomized or quantum reductions were needed to reduce $\QMA$ to $\UQMA$, in analogy to the classical case. Regardless, we note that \emph{derandomizing} the Valiant-Vazirani lemma is a long-standing open question by itself~\cite{beigel1998np, klivans1999graph}, and this intermediate result rules out such a derandomization in the black-box setting, see~\Cref{remark:derandomVV}. Thus, we believe the relationship between $\BQP^\UQMA$ and $\QMA$ is the proper question through which to study the power of unique quantum witnesses. 

\paragraph{On classical vs. quantum oracles.} Both of our separations ~\Cref{thm:informaloraclesep} and ~\Cref{thm:informalqmaseparation} are by quantum oracles. In our approach in~\Cref{thm:informaloraclesep}, the quantumness is necessary as otherwise the classical version of our oracle problem is easily shown to be contained in $\NP \subseteq \BQP^{\UQMA}$. A similar observation can be made for~\Cref{thm:informalqmaseparation}. Historically, when the class or resource is classical, a classical separation can be viewed as the ``gold-standard,'' as evidenced by the long series of works aiming to separate $\QCMA$ and $\QMA$ with a classical oracle (see the exposition in \cite{zhandry_toward_2024}). However, this decision should be made on a case-by-case basis. In our context, the subspace reflection oracle captures a way in which an algorithm might probe the ground space of a Hamiltonian. Naively, the best algorithmic way to ‘access’ this subspace is to use quantum phase
estimation to, say, implement reflections about it. There is a priori no known reason to suspect that this subspace will be well structured for worst-case Hamiltonians and thus we instantiate the ground space with a subspace in a randomly chosen basis. Another consideration when using and interpreting quantum oracle separations is that they only rule out quantum relativizing proof strategies.
Indeed,~\cite{aaronson2008perfect, agarwal2025cautionary} already point out that the proof technique of explicitly representing quantum amplitudes is classically relativizing but quantumly not (and, to our knowledge, this is the only known proof technique with such property). However, this technique has only been used used to show containments in very powerful classical classes like $\EXP$ and $\PSPACE$. Thus, quantum oracle separations are still useful in ruling out black-box reductions at lower complexity regimes like $\BQP^\UQMA$ and $\QMA$.

\paragraph{Improved quantum approximate counting lowerbounds.} Classically, the question of uniqueness is closely related to approximate counting. In \cite{Valiant_VV1986_NPEasyDetecting}, a random hash is used to isolate a single solution to an $\NP$ circuit. The celebrated Stockmeyer algorithm~\cite{stockmeyer1983complexity} uses similar tools to give a $\BPP^{\NP}$ algorithm for estimating the number of accepting witnesses to multiplicative precision. A quantum analogue of this is to multiplicatively estimate the dimensionality of the accepting subspace of a $\QMA$ circuit~\cite{brown2011computational, shi2009note, bravyi2024quantum}. The complexity class capturing this quantum approximate counting problem is $\QXC$. As usual, one should be careful when defining quantum analogues of classical classes. However, this problem is well motivated and shown in \cite{Bravyi_BCGW2022_ComplexityQuantumPartition} to be
connected to a number of physical relevant computational problems, such as computing quantum partition functions and estimating local observables of Gibbs states. $\QXC$ was further
studied in \cite{bravyi2024quantum} and shown to contain the problem of estimating the Kronecker coefficients of the
symmetric group. Stockmeyer's approximate counting algorithm has led to many important applications, including the sampling-counting equivalence~\cite{jerrum1986random} and even the theoretical foundations for quantum advantage experiments~\cite{hangleiter2023computational}. Thus, it is natural to ask whether there exists a quantum analogue: Could quantum approximate counting be solved in (quantum) polynomial time with access to an oracle solving $\QMA$ decision problems? In essence, do we have $\QXC {\subseteq} \BQP^{\QMA}$?

The phrase ``quantum approximate counting'' has appeared elsewhere in the literature. In \cite{kuperberg2009hard, aaronson2022acrobatics}, the authors study $\textsf{SBQP}$, the quantum analogue of $\textsf{SBP}$ (where $\textsf{SBP}$ is the multiplicatively precise analogue of $\textsf{BPP}$), which captures the task of estimating the
acceptance probability of a $\BQP$ machine to multiplicative precision. This is not known to be equivalent to the definition of $\textsf{QXC}$ (unlike $\textsf{SBP}$ which is equivalent to classical approximate counting). We do not consider $\mathsf{SBQP}$ in this work. Additionally, in \cite{Aaronson_AKKT2020_QuantumLowerBounds}, ``quantum approximate counting'' is used to refer to the task of designing quantum ($\BQP$ or $\QMA$) algorithms for classical approximate counting, given quantum queries to a Boolean oracle.~\cite{she2023unitary} consider a quantum generalization of the problem in~\cite{Aaronson_AKKT2020_QuantumLowerBounds}, called the quantum approximate \emph{dimension} problem; the task here is to decide whether an unknown subspace has dimension $\leq w$ or $\geq 2w$, given access to the corresponding reflection oracle. Using the quantum approximate dimension problem, we show

\begin{theorem}[\Cref{cor:qxc-vsqma2}, restated]
\label{thm:informalqmaseparation}
There exists a \emph{quantum} oracle relative to which $ \QXC \not\subseteq \BQP^{\QMA}$, and in fact, $ \QXC \not\subseteq \QMA^{\QMA}$.
\end{theorem}

Our~\Cref{thm:informalqmaseparation} rules out the existence of a quantum analogue of Stockmeyer's theorem in the black-box setting. This is a strengthening of a result of \cite{she2023unitary}, where the authors show a $\QMA$ query lower bound against this task via a reduction to the classical problem in~\cite{Aaronson_AKKT2020_QuantumLowerBounds}. Again, the quantum oracle is necessary in our approach, because the classical version of our oracle is the classical approximate counting problem, which is contained in $\BPP^{\NP} \subseteq \BQP^{\QMA}$. The polynomial methods used in \cite{Aaronson_AKKT2020_QuantumLowerBounds, she2023unitary} do not transfer to the relativized setting in the presence of $\QMA$-oracle calls. Moreover, prior works~\cite{aaronson2022acrobatics, agarwal2024oracle} dealing with quantum versions of the polynomial hierarchy are all based on variants of the switching lemma, which are not known to be compatible with quantum oracles. Therefore, we need to develop new techniques to handle quantum oracle calls in ``stacked'' complexity classes like $\QMA^\QMA$. 

Along the way we also prove a \emph{classical} oracle separation $\QXC \not\subseteq \QMA$, in fact, $\mathsf{SBP} \not\subseteq \Ptime^{\QMA}$ (\Cref{cor:pqmavsqxc}), thus giving a much simpler and elementary proof of the $\QMA$ lower-bound against classical approximate counting  of~\cite{Aaronson_AKKT2020_QuantumLowerBounds} (who showed $\mathsf{SBP} \not \subseteq \QMA$). Our separation is perhaps stronger as $\P^\QMA$ likely strictly contains $\QMA$ as it contains $\coQMA$.
 
\paragraph{Unique witness from a well-connected ground space.}
\label{par:unique_witness}
\Cref{thm:informaloraclesep} underscores the importance of gaining a deeper understanding of the structure of quantum ground spaces to extract a unique quantum witness. For
this, it is crucial to identify classes of Hamiltonians for which we know how to verify a unique low-energy witness.
However, the only such classes of Hamiltonian known so far are classical Hamiltonians
\cite{Valiant_VV1986_NPEasyDetecting}, local Hamiltonians with ground states that can be prepared by a polynomial-sized
quantum circuit \cite{Aharonov_ABBS2022_PursuitUniquenessExtending}, and local Hamiltonians with a polynomial sized
low-energy subspace \cite{Jain_JKK+2011_PowerUniqueQuantum}. Our final result identifies a new class of Hamiltonians for
which we can verify a unique low energy state. Our starting point is a simple observation already hinted at in the
previous paragraphs: An algorithm that only uses subroutines treating the \emph{entire} Hamiltonian $H$ as a single
operator cannot distinguish between states of the same energy as they are indistinguishable with respect to $H$. For example, quantum
phase estimation, which only uses time evolutions of $H$, falls into this category. Thus, any quantum algorithm that
aims to verify a unique low-energy state must use operators which are more fine-grained than $H$ itself, namely the
\emph{individual} local terms in $H$.

The action of local operators on Hamiltonian eigenstates is in itself a challenging problem, with few general results known \cite{Arad_AKL2016_ConnectingGlobalLocal}. Despite this, there are families of local Hamiltonians for which useful statements can be made. A particularly interesting such family is the set of local Hamiltonians that satisfy the Eigenstate Thermalization Hypothesis (ETH) from quantum statistical mechanics \cite{srednicki1999approach}. This hypothesis roughly says that within an energy window, eigenstates are well-connected in the sense that local operators induce transitions from one eigenstate to many others.  We adopt the following informal definition from~\cite{Chen_CB2023_FastThermalizationEigenstate}. 
\begin{quote}     \textit{ETH (informal).} There is a collection of polynomially many local observables $A_1, \ldots A_m$ such that for every eigenstate $\ket{\psi}$ with energy roughly $E$, the vector $A_i\ket{\psi}$ roughly stays within a small energy window around $E$, and has a fairly good overlap with all the eigenstates in this window.
\end{quote}
\noindent See~\Cref{sec:ethoverview} for the formal definition of ETH (\Cref{hypo:cETH}) and a discussion on its physical origins. From the computer science perspective, a well-connected low-energy eigenspace should have features of an expanding graph and hence allow one to single out a unique `fixed point'. We solidify this intuition by generalizing the quantum expander test for verifying the maximally entangled state \cite{Aharonov_AHL+2014_LocalTestsGlobal} and show the following theorem:
\begin{theorem}[Informal restatement of \Cref{thm:alg_correctness}]
\label{thm:informalthmETH}
There is a $\UQMA$ protocol for estimating the ground energy of a local Hamiltonian satisfying ETH.  
\end{theorem}
The protocol verifies a unique quantum state in two copies of the subspace with energy around $E$, if ETH holds in that subspace. Assuming $E$ is the ground energy itself, we can extract a unique ground state witness. If $E$ is larger but sub-extensive (that is, $n^c$ with $c<1$, as a function of the number of qubits $n$), we can still obtain a unique witness for the purposes of $\UQMA$ protocols due to known reductions concerning the local-Hamiltonian problem.\footnote{If ETH only holds for extensive energies (that is, for energies $\geq cn$ for some small constant $c$,) then  \Cref{thm:informalthmETH} does not apply; however it gives some implications assuming quantum PCP conjecture, see~\Cref{subsec:connections}.} 

A likely implication of \Cref{thm:informalthmETH} is that the complexity of local Hamiltonians satisfying ETH, such as chaotic quantum Hamiltonians including the SYK model \cite{Sonner2017}, is significantly lower than that of general local Hamiltonians. It is unclear if the theorem would help establish that $\UQMA$ and $\QMA$ are equivalent under quantum reductions, since we do not know if $\QMA$-hard local Hamiltonians would satisfy the ETH assumption.

\subsection{Techniques}

\paragraph{The oracles.}
Our oracles will be of the form $\O_\qS = \Id - 2 \Pi_\qS$, corresponding to a reflection about an unknown subspace $\qS$. Variations of this oracle have been used in other works. The case when $\qS$ corresponds to a $1$-dimensional subspace is exactly the oracle used by \cite{Aaronson_AK2007_QuantumClassicalProofs} to separate $\QCMA$ from $\QMA$. When increasing the dimension and restricting the subspace $\qS$  to be classical, this becomes the classical oracle used to separate $\QMA$ and classical approximate counting \cite{Aaronson_AKKT2020_QuantumLowerBounds}; when the subspace is quantum, this separates $\QMA$ and quantum approximate counting \cite{she2023unitary}.
Having said that, our results and proof techniques do not follow from these works in any straightforward way. For example, the~\cite{Aaronson_AK2007_QuantumClassicalProofs} oracle would not work in our setting as this problem is easily seen to be in $\UQMA$. A better candidate is the oracle $\mathcal{O}_S$ corresponding to a large subspace $S$, as there does not seem to be a natural ``unique'' witness that the prover could provide. However, when constructing our query lower bounds (which are, as usual, the heart of nearly all oracle separations), new techniques are necessary to account for the uniqueness restriction on the witness. Indeed, prior oracle separations involving $\QMA$ avoid directly dealing with the witness either by conditioning on a particular classical witness \cite{Aaronson_AK2007_QuantumClassicalProofs}, replacing the witness with the maximally mixed state~\cite{Aaronson_AKKT2020_QuantumLowerBounds, she2023unitary}, or reducing to a classical circuit lower bound \cite{aaronson2022acrobatics, agarwal2024oracle}. Instead, our work carefully uses the requirement that the witness be unique.

\paragraph{Separating $\QMA$ from $\BQP^\UQMA$} 
We consider circuits with access to the unknown subspace reflection oracle $\O_S$, as well as an oracle $\mc A$ for $\textsf{UniqueQCircuitSat}$ (more precisely, their ``oracularized'' versions), the natural complete problems for $\UQMA$. The oracle distinguishing task (simplified) is the following.
\begin{task}[Non-empty vs empty]
\label{task:uqma_classical}
Given an oracle $\O_\S$ encoding a reflection about an unknown subspace $\S$, distinguish between the following cases:
\begin{itemize}
    \item (YES case) $\S$ is a non-trivial subspace so that $\dim(S) \geq 1$, and
    \item (NO case) $\S$ is an empty subspace.
\end{itemize}
\end{task}

The overall idea for both separations is to use a hybrid argument and separately consider steps of the $\BQP^{\UQMA}$ verifier where a query is made to the circuit sat oracle $\mc A$ and steps where the verifier directly queries the oracle $\O_S$. Therefore, our separation requires two ingredients.
\begin{enumerate}
    \item First, we argue that the verifier's calls to $\mc A$ are either on circuits $\ttC$ which are themselves good verifiers for the oracle distinguishing task \emph{or} are nearly useless. We show this by establishing a \emph{polarization} property of oracular quantum circuits. \label{item:one}
    \item Second, we establish our oracle separation between $\UQMA$ and $\QMA$. This allows us to rule out the first case in the item above, as well as to argue that the verifier's direct calls to $\O_S$ are not too helpful. \label{item:two}
\end{enumerate}
We next give some details on the above items.

In a sense \Cref{item:two} is where the uniqueness property is crucially used. Imagine we have two classical subspaces $S$ and $T$, where $T$ is a subspace \emph{extending} $S$, so that $T = S \oplus \Delta$. We will pick the dimension of $S$ and $T$ to be on the order $\mathrm{superpoly}(n)$; we will show that as $\dim(S), \dim(T) \ll 2^n$, a $\UQMA$ verifier using $\poly(n)$ oracle queries cannot distinguish $S$ and \emph{random} extensions $T$. When $S$ is taken from as a YES instance of an oracle distinguishing task with a unique witness $\phi_S$, this observation essentially says that the ``same'' witness works for $T$. Flipping the order of quantification, we can equivalently say that for a random $T$, \emph{nearly all} of its ``roots'' $S$ have the same witness. But since $\dim(T) \gg \dim(S)$, the $S$'s can be picked so that they are orthogonal. Finally, applying a variation of the hybrid method shows that the same witness must overlap too many orthogonal subspaces, yielding a contradiction. As mentioned previously, the oracle here can even be restricted to be classical (see~\Cref{cor:oracle_separation}).

However, when moving to \Cref{item:one}, the use of quantum oracle is seemingly necessary. When considering the $\BQP^{\UQMA}$ verifier's direct calls to $\O_S$, the correlation between the verifier's state and the oracle is naturally captured by the overlap of some intermediate state and the projector onto the oracle subspace; we can get a lot of mileage by studying quantities of the form ``$\tr[\Pi_S \psi]$'' and averaging over subspaces $S$. The calls to the circuit sat oracles $\mc A$ require more careful handling. The challenge is, for each oracular circuit $\ttC^{\O_S}$ on which the outer verifier is called, we would
like to show that the answer $\mc A(\ttC^{\O_S})$ cannot help tell apart the yes and no case oracles (with high probability). We establish this via a concentration argument, enabled by our use of quantum oracles. In our case, we use a version of Levy's lemma for the Grassmannian (the space of $k$-dimensional subspaces of $\C^N$), due to \cite{gotze2023higher}.
\begin{lemma}[Levy's lemma on the Grassmannian (informal)]
    \label{lem:concentration_simple}
    Let $\Delta$ be uniform over $G_{N,k}$ and $\Pi_\Delta$ be corresponding projector. Given any quantum state $\ket \psi \in \C^{N'}$, with $N' \geq N$, we have that
    \[
        \Pr_{\Delta \sim G_{N,k}}[\tr[\Pi_\Delta \ketbra \psi] \geq 2N^{-1/3}] \leq \exp\Paren{-\textsf{\textup{poly}}(N)}\,.
    \]
\end{lemma}
Via~\Cref{lem:concentration_simple} we show that the acceptance probabilities of oracular quantum polynomial time circuits must \emph{concentrate} when averaged over the random quantum subspace oracles. A similar observation is made in \cite{irani2021quantum} to rule out the possibility of a search-to-decision reduction for $\QMA$. We use this observation somewhat differently. Whereas \cite{irani2021quantum} used concentration to argue that on across \emph{yes} instances the circuit behaves similarly, we instead want to argue the same thing across yes\emph{ and no} instances. This requires a different and more intricate argument considering the average acceptance probabilities in the yes and no cases individually. In the specific form used in this work, we call this property ``polarization'' (see \Cref{lem:polarization}).

\begin{lemma}[Circuit polarization (informal)]
\label{lem:polarization_informal}
Let $\texttt{C}^{\O}$ be any $\cpoly(n)$-query $\UQMA$ verifier with access to a subspace oracle $\O$. Let $\gamma \in \exp(-\cpoly(2^n))$. Then, either
\[
    \Pr_{\text{random oracles}}[\mathcal{A}(\texttt{C}^{\O_\mathrm{Yes}}) \neq \mathcal{A}(\texttt{C}^{\O_\mathrm{No}})] \geq 1 - \gamma \qquad \text{or} \qquad \Pr_{\text{random oracles}}[\mathcal{A}(\texttt{C}^{\O_\mathrm{Yes}}) \neq \mathcal{A}(\texttt{C}^{\O_\mathrm{No}})] \leq \gamma\,.
\]
\end{lemma}
The final obstacle when designing oracle separations for relativized (or ``stacked'') classes is the fact that the natural complete problem for $\UQMA$, quantum circuit sat, is defined as a \emph{promise problem}. This subtlety makes defining the precise behavior of an oracular verifier non-trivial, but also gives us a surprising amount of flexibility. Indeed, this freedom in defining a quantum circuit sat oracle $\mc A$'s behavior on instances outside the promise range is what enabled \cite{irani2021quantum} to demonstrate the impossibility of a generic search-to-decision reduction for $\QMA$. In this work, we use the definition of an oracular verifier proposed by \cite{aaronson2022acrobatics}, where we are free to ``extend'' the oracle to behave arbitrarily on instances outside the promise gap, see~\Cref{sec:prelim}.

\paragraph{Separating $\QXC$ from $\QMA^\QMA$}
The overall argument here is similar to above, except we consider a slightly different task.
\begin{task}[Small vs large]
\label{task:qma_classical}
Fix constants $0 < k_1 < k_2 < 1$. Given an oracle $\O^\S$ encoding a reflection about an unknown subspace $\S$, distinguish between the following cases:
\begin{itemize}
    \item (YES case) $\dim(\S) \leq 2^{k_1 n}$, and
    \item (NO case) $\dim(\S) \geq 2^{k_2 n}$.
\end{itemize}
\end{task}

Intuitively, the different parameters can be attributed to the fact that in $\UQMA$,  the unique witness condition affords structure even within the YES case instances (thus, both subspaces of dimension $\approx 2^{k_1 n}$ and $\approx 2^{k_2 n}$ can be treated as YES cases). For $\QMA$, we need to treat the dimension $\geq 2^{k_2 n}$ as a NO case in order to carry through similar arguments. As in the case of $\UQMA$, we will use an oracle separation between $\QMA$ and $\QXC$ as well a version of the polarization lemma for $\QMA$ verifiers. In both separations, the polarization property puts strong limitations on the way the outer verifiers can interact with the circuit sat oracles. As such, we hope these techniques will generalize beyond our setting to yield lower bounds against verifiers querying quantum oracles, such as those capturing constant stacks of $\QMA$ (the `$\QMA$ hierarchy').

Finally, we make a remark on our (classical) oracle separation of $\QMA$ and $\QXC$. One might be tempted to leverage an argument similar to \cite{Aaronson_AK2007_QuantumClassicalProofs} to remove the outer witness. This does not work, as the witness in our case is quantum and would require conditioning on a event with probability $\leq \exp(-2^{\cpoly(n)})$. Instead, we use a hybridization argument similar to the $\UQMA$ proof. When extending a (yes case) small subspace $S$ to a (no cases) larger subspace $S \oplus \Delta$, the only way the $\QMA$ algorithm changes its behavior is that the witness overlaps with $\Delta$. Thus, choosing many orthogonal extensions $\Delta$ leads to a contradiction. This recovers the $\QMA$ lower bound for classical approximate counting (with quantum queries) of \cite{Aaronson_AKKT2020_QuantumLowerBounds}, albeit by a much more straightforward proof without involving the machinery of Laurent polynomials used in~\cite{Aaronson_AKKT2020_QuantumLowerBounds}. As an easy consequence of our proof, we improve the latter result to show the same lower bound holds for $\P^\QMA$.

\paragraph{Protocol for unique verification.} Finally, we switch gears and show that under an appropriate structural assumption, unique verification \emph{is} possible for the local Hamiltonian problem. As mentioned earlier, we prove \Cref{thm:informalthmETH} by designing a quantum expander-like test \cite{Aharonov_AHL+2014_LocalTestsGlobal} for a given low-energy subspace that satisfies the eigenstate thermalization hypothesis. Whereas the protocol of \cite{Aharonov_AHL+2014_LocalTestsGlobal} yields a test for states which are maximally entangled over an explicitly known subspace, we use the ETH to extend their test to a subspace implicitly described by a local Hamiltonian $H$. The unique witness that we verify is a highly entangled state on two copies of the local Hamiltonian system, localized in a low-energy window where the ETH is assumed to hold. To accomplish ``mixing,'' the verification protocol uses the set of local operators $A_1, \ldots, A_m$ from ETH, which will behave as approximate unitaries in the subspace.

More concretely, given a purported low-energy state, we first use quantum phase estimation to perform a projector $\{\Pi, I-\Pi\}$ on the eigensubspace corresponding the energy window satisfying ETH. We then would like to apply a set of ``expanding unitaries'' within this subspace. An idea is to simply use the observables $A_i$ to generate unitaries of the form $e^{\mathrm{i}A_i}$ on each copy of the system. However, doing so will uncontrollably bring the system outside of the relevant subspace. We take ideas from the quantum Zeno effect and instead apply unitaries of the form $e^{\mathrm{i}\varepsilon A_i}\otimes e^{-\mathrm{i}\varepsilon A_i}$ for sufficiently small $\varepsilon>0$. Finally, we apply the projector $\{\Pi, I-\Pi\}$ again. To analyze this protocol, we formulate a computational interpretation of ETH in~\Cref{hypo:cETH} that is inspired by the pseudorandomness literature, which could be of independent interests.
We then use the computational ETH to show that the resulting approximate unitaries in the energy window form an approximate quantum expander. One difference from standard quantum expanders is that the unique state we can verify is not necessarily the maximally entangled state; all we know is that it is a gapped stationary state guaranteed by the Perron-Frobenius theorem.
We point out here that we also need a mild assumption on the Hamiltonian density of states (see \Cref{as:smooth}). Additionally, we highlight that \cite{Chen_CB2023_FastThermalizationEigenstate} also used quantum expander behaviors and other assumptions in a quantum Gibbs sampler, but our protocol is simpler since its goal is low-energy state verification instead of Gibbs state preparation.

It may appear surprising that a (variant of) the EPR testing protocol appears in the above computational protocol, while it was originally introduced in the communication setting \cite{Aharonov_AHL+2014_LocalTestsGlobal}. This is an example of the communication vs computation correspondence: creating the state to be tested is computationally hard while testing it can be easy, in analogy with the fact that creating a maximally entangled state requires large communication but testing it does not.

\subsection{Physical relevance}
On the surface, our first result can be viewed as evidence that $\UQMA$ is not equal to $\QMA$. An alternative interpretation is via a characterization due to \cite{Aharonov_ABBS2022_PursuitUniquenessExtending}; in that paper, the authors show that the class of local Hamiltonians with an inverse polynomial spectral gap is $\UQMA$ complete. Thus, our oracle separations give evidence that the local Hamiltonian problem is easier for inverse polynomially gapped instances\footnote{A more recent work~\cite{deshpande2022importance} also studied the role of the spectral gap in small promise gap regimes.}. This is interesting from a physics perspective as many classes of physically relevant Hamiltonians are known to have an inverse polynomial spectral gap (e.g., ferromagnetic Heisenberg~\cite{koma1997spectral}, 1D AKLT~\cite{affleck2004rigorous}, Movassagh-Shor~\cite{movassagh2014power}). In particular, the presence of a spectral gap in these models immediately makes studying various properties more tractable \cite{orus2014practical,landau2015polynomial}. Thus our work provides some evidence for the physical intuition that inverse polynomially gapped Hamiltonians are ``easy'' by separating such Hamiltonians from $\QMA$, under an oracle.

\vspace{1em}
\noindent %Our second result shows that a large class of Hamiltonians, satisfying a computational variant of the ETH, can be solved by a $\UQMA$ protocol.Hence, just as easily as local Hamiltonians with a $\poly(n)^{-1}$ gap.
Our third result has implications for the Eigenstate Thermalization Hypothesis. If $\UQMA$ and $\QMA$ are not equivalent (under quantum reductions), then hard instances of the local Hamiltonian problem cannot satisfy ETH in its low energy window (near-extensive). This violation of ETH is robust against adversarial perturbations of near-extensive strength.
Under the quantum PCP conjecture~\cite{AAV13}, which states that there are local Hamiltonians such that ground energy estimation is $\QMA$-hard even up to extensively scaling accuracy, the latter statement can be extended to extensive energy scales. In other words, if $\UQMA$ and $\QMA$ are not equivalent, then quantum PCP Hamiltonians need to robustly violate ETH. Note that the same assumption implies $\QMA$-hard Hamiltonians are not many-body localized~\cite{nandkishore2015many} either (see~\Cref{subsec:connections}). Thus, our results suggest that these Hamiltonians may correspond to an interesting phase of matter that is robust against adversarial perturbations.

We expect our unique verification protocol to work for families of random quantum Hamiltonians such as random spin systems~\cite{erdHos2014phase} or the SYK model~\cite{sachdev1993gapless,kitaev2015simpleI,kitaev2015simpleII,hastings2022optimizing}. If random quantum Hamiltonians satisfy a sufficiently strong variant of ETH, our result implies that the average-case is considerably easier than the worst-case, under the assumption $\UQMA$ and $\QMA$ are not equivalent.
Indeed, randomized versions of $\NP$-hard problems such as optimizing the Sherrington-Kirkpatrick model~\cite{sherrington1975solvable} have been proven to be average-case easy~\cite{montanari2021optimization}.

\subsection{Open problems}
There are multiple avenues to continue this work. Here, we list a couple:
\begin{itemize}
    \item Mulmuley, Vazirani, and Vazirani \cite{Mulmuley1987} introduced a variant of the isolation lemma in \cite{Valiant_VV1986_NPEasyDetecting}, by introducing random weights to the edges of a graph to isolate a unique maximally weighted clique. The quantum analogue of this approach is to consider the quantum variant of the Clique problem, which is $\QMA$ hard \cite{BeigiShor08}. It seems to suffer with issues analogous to those raised in \cite{Aharonov_ABBS2022_PursuitUniquenessExtending}: introducing $\poly(n)$ amount of  randomness does not suffice to single out a unique quantum witness. Can the ideas in~\cite{Mulmuley1987} still be useful in some quantum problems?  
    \item \Cref{thm:informalqmaseparation} demonstrates a quantum oracle to which $\QXC^\O \not \subseteq \QMA^{\QMA^\O}$. We expect our tools can be extended to showing separations for constants stacks of $\QMA$ (i.e. for the $\QMA$ hierarchy~\cite{aaronson2022acrobatics})? Extending our argument to, e.g., $\QMA^{\QMA^\QMA}$ would require 1) showing polarization for $\QMA^\QMA$ circuits, and 2) a two-sided query lower bound, along the lines of \cite{she2023unitary}, but for relativized classes.
    \item  \Cref{thm:informalthmETH} constructs a quantum algorithm that verifies some low energy state under ETH. Can
        we verify a more natural low energy state assuming ETH, such as the purification of Gibbs state (called the
        thermofield double state)? Note that, due to the Feynman-Kitaev circuit-to-Hamiltonian mapping~\cite{Kitaev_KSV2002_ClassicalQuantumComputation}, uniquely verifying a state is equivalent to the state being a near ground state of a local Hamiltonian with inverse polynomial spectral gap. Such local Hamiltonians for the thermofield double state have been constructed under physics-motivated assumptions in \cite{CFHL19}.
        \item Our results hint at a separation between worst-case and average-case local Hamiltonians.
    Random spin Hamiltonians~\cite{erdHos2014phase} or the SYK model~\cite{sachdev1993gapless,kitaev2015simpleI,kitaev2015simpleII} are expected to satisfy the ETH and are therefore solvable by a $\UQMA$ machine. It is interesting to find more evidence for the easiness of average-case Hamiltonians.
    Recent work studies the complexity of approximating the maximal energy of the SYK model. 
       Ref.~\cite{hastings2022optimizing} provides a quantum algorithm to prepare a state with energy $c\sqrt{n}$ -- a constant fraction of the expected optimum.  
       Ref.~\cite{anshuetz2024strongly} finds further evidence that estimating the energy of the SYK might even be quantumly easy.
       More concretely, they show that the annealed and quenched free energies agree at inverse polynomial temperatures. 
       This is in contrast to classical random spin systems where these quantities disagree when the system is in a ``glassy'' phase that is algorithmically hard.
        \item It is interesting to find more connections between quantum many-body physics assumptions such as ETH and quantum computer science. The works \cite{Chen_CB2023_FastThermalizationEigenstate, SM23} show that ETH implies fast preparation of Gibbs quantum states, if there is no bottleneck throughout the eigenvalue distribution and Gibbs distribution. However, the latter assumption cannot be justified on general grounds. On the other hand, our results work perfectly fine without this assumption.
\end{itemize}

\subsection{Organization of paper}
In~\Cref{sec:prelim}, we define basic notions such as complexity classes, local Hamiltonians, quantum oracles, norms, and introduce and prove some useful concentration lemmas. In~\Cref{sec:qma_vs_uqma}, we give background on the $\UQMA$ vs $\QMA$ questions and its physical relevance before proving the oracle separation in~\Cref{thm:informaloraclesep}. We prove classical and quantum oracle separations regarding quantum approximate counting in~\Cref{sec:qxc}. In~\Cref{sec:eth}, we overview the physical motivations and introduce the eigenstate thermalization hypothesis. Then we formulate a computational variant of ETH suitable for applications to complexity theory and describe a unique verification algorithm based on this variant. The Appendix sections include random matrix theory calculations relevant for the verification algorithm as well as a proof of norm properties.

\ifanon
\else
\subsection{Acknowledgments}
In an earlier version of this paper, the oracle separation between $\UQMA$ and $\QMA$ was given via a quantum oracle. We thank Chinmay Nirkhe and Anand Natarajan for pointing out that our proof did not in fact need quantumness in this case. This version has been updated to provide a classical oracle separation.
 
We thank Soonwon Choi, Sam Garratt, Yingfei Gu, Rahul Jain and Shivaji Sondhi for insightful discussions, and especially thank Soonwon Choi for sharing an example which shows that the computational ETH assumption is incorrect if the Gaussian prescription is applied at high energies (we only applying it at low energies). This work was done in part while the authors were visiting the Simons Institute for the Theory of Computing, supported by DOE QSA grant number FP00010905. 
AA and QTN acknowledge support through the NSF Award No. 2238836. AA acknowledges support through the NSF award QCIS-FF: Quantum Computing \& Information Science Faculty Fellow at Harvard University (NSF 2013303). QTN acknowledges support through the Harvard Quantum Initiative PhD fellowship. 
JH acknowledges support through the Harvard Quantum Initiative postdoctoral fellowship.
YH is supported by the National Science Foundation Graduate Research Fellowship under Grant No. 2140743. Any opinion, findings, and conclusions or recommendations expressed in this material are those of the authors(s) and do not necessarily reflect the views of the National Science Foundation.
\fi

\section{Preliminaries}\label{sec:prelim}

\input{preliminaries}

\section{\texorpdfstring{$\UQMA$ versus $\QMA$}{UniqueQMA versus QMA}}
\label{sec:qma_vs_uqma}

\input{plausibility}

\input{separation}

\input{vv_separation.tex}

\input{apx_counting}

\begin{comment}
\section{Quantum approximate counting and $\BQP^{\QMA}$}
\label{sec:counting}

\input{counting}
\end{comment}

\section{UniqueQMA protocol for chaotic Hamiltonians}
\label{sec:eth}

\input{eth}

\newpage

\printbibliography
\appendix

\section{Miscellaneous proofs}
\label{sec:misc_proofs}
\subsection{Probabilistic tools}
\label{sec:prob_tools}
In this section, we use \textbf{bold face} to distinguish random variables from constants.
\begin{fact}[Bernstein's inequality]
\label{fact:bernstein}
Let $\bs x_1, \dots, \bs x_t$ be independent random variables with $\E \bs x_i = 0$, satisfying
\begin{enumerate}
    \item $\sum_i \E \bs x_i^2 = v$
    \item For all $k > 2$, $\E \bs x_i^k \leq \tfrac {\E \bs x_i^2}{2} k! L^{k-2}$
\end{enumerate}
Then,
\[
\Pr[\sum X_i \geq \eps] \leq \text{exp}\left[\frac{-\eps^2/2}{v + L \eps /3}\right]
\]
\end{fact}
\begin{fact}[Distribution of squares of complex Gaussians]
\label{fact:square_gauss}
Let $g \sim \mathcal{CN}(0,1)$. Then, $|\bs g|^2 \equiv \tfrac 1 2 \bs h$, where $\bs h$ is drawn from a Chi-squared distribution with $2$ degrees of freedom, i.e. $\bs h \sim \chi^2(2)$.
\end{fact}
\begin{proof}
    Writing $\bs g = \bs a + \i \bs b$ with $\bs a, \bs b \sim \mathcal{N}(0,1/2)$, we have that $|\bs g|^2 = \bs a^2 + \bs b^2$, i.e. the sum of the squares of two Gaussians with variance $\tfrac 1 2$. Renormalizing such that each of their variances is $1$ yields the claim. 
\end{proof}
\begin{fact}[Linear combination of complex Gaussians]
\label{fact:Gaussian_comb}
Let $\bs g_1,\dots, \bs g_n$ be a collection of complex Gaussians, $\bs g_i \sim \mathcal{CN}(0,1)$. Then, for any $u \in \C^n$,
\[
\tilde \bg \dfn \sum_i u_i \bg_i \sim \mathcal{CN}(0, \|u\|_2^2)
\]
\end{fact}
\begin{proof}
    Each $\bg_i \equiv \bs a_i + \i \bs b_i$ with $\bs a_i, \bs b_i \sim \mathcal{N}(0,1/2)$. The claim will follow by showing that,
    \begin{enumerate}
        \item If $\bs g \sim \mathcal{CN}(0,1)$ then $u \bs g \sim \mathcal{CN}(0,|u|^2)$
        \item If $\bs g, \bs h \sim \mathcal{CN}(0,1)$ then $\bs g + \bs h \sim \mathcal {CN}(0,2)$.
    \end{enumerate}
    The claim follows easily then, since each $u_i \bs g_i \sim \mathcal{CN}(0, |u_i|^2)$ and $\sum_i u_i \bs g_i \sim \mathcal{CN}(0, \sum_i |u_i|^2) = \mathcal{CN}(0,\|u_i\|_2^2)$.
    
    First we show the first property. Let $u = v + \i w$, we have that
    \begin{align*}
        u\bs g&= (v\bs a- w\bs b) + \i (w\bs a+ v\bs b)\\
        &= {\bs a}' + \i {\bs b}'
    \end{align*}
    Clearly $\bs a'$ and $\bs b'$ are each Gaussians; moreover, they can be verified to be independent. Additionally, $\E \bs a' = 0$ and $\E (\bs a')^2 = v^2 \E \bs a^2 + w^2 \E \bs b^2 = \tfrac 1 2|u|$, so $\bs a' \equiv |u|\bs a''$ with $\bs a'' \sim \mathcal {CN}(0,1/2)$. Similarly $\bs b' \equiv |u|\bs b''$. So,
    \begin{equation}
        \label{eq:scaling}
        u \bs g = (\bs a' + \i \bs b') \equiv |u|(\bs a'' + \i \bs b'') \sim \mathcal{CN}(0, |u|)
    \end{equation}
    Now consider the second property. If $\bs g = \bs a + \i \bs b$ and $\bs h = \bs c + \i \bs d$, then clearly $\bs g + \bs h = (\bs a + \bs c) + \i (\bs b + \bs d) =: \bs a' + \i \bs b'$. Clearly $\bs a', \bs b'$ are independent and each distributed as $\mathcal{N}(0,1)$. Pulling out the scaling factor as in \Cref{eq:scaling} yields that $\bs g + \bs h \sim \mathcal{CN}(0,2)$.
\end{proof}
Next, we prove \Cref{clm:iidgauss}
\iidgauss
\begin{proof}
This proof closely follows the proof of Fact D.1 in \cite{Chen_CB2023_FastThermalizationEigenstate}, which in turn is based on argument from Section 2.3 of \cite{tao2023topics} (with the slightly modification that in \cite{tao2023topics}, they consider bounded random variables). Write $P = \sum_{\alpha,\beta} \bs g_{\alpha,\beta}\ketbra{v_\alpha}{v_\beta}$. Recalling that $\|P\| = \max_{\ket \psi} \|P \ket \psi\|_2$, we first analyze the concentration we get for a single fixed $\ket \psi = \sum_\alpha \psi_\alpha \ket {v_\alpha}$. We have,
\[
    P \ket \psi = \sum_\alpha \langle (\bs g_{\alpha,\beta})_{\beta}, (\psi_\beta)_{\beta}\rangle \ket{v_\alpha} =: \sum_\alpha \bs P_\alpha \ket{v_\alpha}
\]
By \Cref{fact:Gaussian_comb}, $\bs P_\alpha = \langle (\bs g_{\alpha,\beta})_{\beta}, (\psi_\beta)_{\beta}\rangle \sim \mathcal{CN}(0, \|(\psi_\beta)_\beta\|_2^2)$. Then,
\[
\Pr[\|P \ket \psi \|_2 \geq c \sqrt D] = \Pr[\sum_\alpha |\bs  P_\alpha|^2 \geq c^2 D]
\]
At this point, we'd like to apply \Cref{fact:bernstein} with the random variables $\bs x_\alpha' = |\bs P_\alpha|^2$. By \Cref{fact:square_gauss} and using that $\|\ket \psi\|_2 = 1$, each $\bs x_\alpha'$ is distributed as $\tfrac 1 2\|(\psi_\beta)_\beta\|_2^2 \cdot \bs y_\alpha = \tfrac 1 2 \bs y_\alpha$ where $\bs y_\alpha \sim \chi^2(2)$. A mild issue is that $\E \bs x_\alpha' = 1$, not $0$ as required in Bernstein's inequality. Thus, we pick $\bs x_\alpha = \bs x_\alpha' - 1$. We also need to bound $\E \bs x_\alpha^k$, compared to $\E \bs x_\alpha^2$. We have that,
\begin{align*}
    \E \bs x_\alpha^k &= \E \Paren{\frac 1 2 \bs y_\alpha - 1}^k\tag{$\bs y_\alpha \sim \chi^2(2)$}\\
    &= \sum_{\ell = 0}^k \binom k \ell \frac{1}{2^\ell} (-1)^{k-\ell}\E \bs y_\alpha^\ell\\
    &= \sum_{\ell = 0}^k \binom k \ell \frac{1}{2^\ell} (-1)^{k-\ell} (2^\ell \Gamma(\ell + 1))\tag*{Moment of Chi-square Distribution}\\
    &\leq \sum_{\ell=0}^k \binom k \ell (-1)^{k-\ell} \frac{(2\ell)!!}{2^\ell}\tag{$\Gamma(\ell+1) \leq \frac{(2k!!)}{2^{k-1}}$}\\
    &\leq (2k)!!\sum_{\ell=0}^k \binom k \ell \frac{1}{2^\ell}\\
    &\leq (2k)!! 2^k\\
    &= (2^k k!)2^k = 4^{k} k!
\end{align*}
Noticing that $\E \bs x_\alpha^2 = 1$, it suffices to take $L = 8$ in the statement of \Cref{fact:bernstein}. Finally, we have that $\E \sum_\alpha \bs x_\alpha^2 = D$. And thus, we have
\begin{align*}
   \Pr[\|P \ket \psi \|_2 \geq c \sqrt D] &= \Pr[\sum_\alpha |\bs  P_\alpha|^2 - 1 \geq (c^2-1) D]\\
   &= \Pr[\sum_\alpha \bs x_\alpha \geq (c^2-1) D]\\
   &\leq \text{exp}\left[\frac{-\frac {(c^2-1)^2} 2 D^2}{D + \frac 8 3 (c^2-1) D}\right]\tag{By \Cref{fact:bernstein}}\\
   \label{eq:single_psi_bound}
   &\leq e^{-\tfrac{(c^2-1)}{7} D}\numberthis
\end{align*}
This is a bound for a single $\ket \psi$. In order to lift this to a bound on the operator norm, we use Lemma 2.3.2 of \cite{tao2023topics} as well as a standard bound on the size of epsilon-nets.

\begin{lemma}[Lemma 2.3.2 of \cite{tao2023topics}]
\label{lem:union_bound}
Let $\Sigma$ be a maximal $1/2$-net of the sphere $S$. Then, for any $D \times D$ matrix with complex coefficients and any $\lambda > 0$, we have
\[
\Pr[\|M\| \geq \lambda] \leq \Pr[\bigcup_{y \in \Sigma} \|M y \|_2 \geq \frac \lambda 2]
\]
\end{lemma}
\begin{lemma}[Lemma 5.2 of \cite{vershynin2010introduction}]
Let $ 0 < \eps < 1$, and let $\Sigma$ be a $1/2$-net of the sphere $S$. Then $\Sigma$ has cardinality at most $5^D$.
\end{lemma}
Now, setting $\lambda = 10 \sqrt D$ in \Cref{lem:union_bound} and $c = 5$ in \Cref{eq:single_psi_bound},  we find that $\Pr[\|P\| \geq 10 \sqrt D] \leq e^{-3D}$. Taking the union bound over all states in the epsilon net (of which there are at most $5^D \leq e^{2 D}$, we obtain a bound of
\[
\Pr[\|P\| \geq 10 \sqrt D] \leq e^{-D}
\]
\end{proof}

\subsection{\texorpdfstring{Properties of the $Q$ operator}{Properties of the Q operator}}
\label{sec:q_proofs}
Let $e_0 \in (0,1)$ be some which energy value and $\mc I_\omega = [e_0 - \omega/2, e_0 + \omega/2]$ be an interval of
width $\omega$ around $e_0$. Let $\textsf{disc}_L(\mc I_\omega)$ be the discretization of $\mc I_\omega$ consisting of
multiples of $1/L$ contained in $\mc I_\omega$. Then $L \cdot \textsf{disc}_L(\mc I_\omega)$ is the element-wise
multiplication of $m \in \textsf{disc}_L(\mc I_\omega)$ by $L$ such that set is now a set of consecutive integers. Then,
given a Hamiltonian $H$, $Q(H)$ is defined as
\begin{equation}
    \label{eq:q_def_appendix}
    Q(H) \dfn \sum_{m \in L\cdot \textsf{disc}_L(\mc I_\omega)}\left(\text{sinc}_L(H-m/L)\right)^2.
\end{equation}
In this section we'll work often with sums of the form $\sum_x \tfrac 1 {x^2}$. These can be computed by the following
lemma,
\begin{lemma}
    \label{lem:polygamma} It holds that
    \[
        \sum_{x = a}^b \frac 1 {x^2} = \Gamma^{(1)}(a) - \Gamma^{(1)}(b+1),
    \]
    where $\Gamma^{(n)}$ is known as the \emph{polygamma} function and is the $n^\text{th}$ derivative of the gamma
    function.
\end{lemma}
This identity follows from the series expansion of $\Gamma^{(n)}(x)$ which can be found in
\cite[\href{http://dlmf.nist.gov/5.15.E1}{(5.15.1)}]{DLMF}. Now, we establish basic properties of the $Q$ operator.
\begin{lemma}
    \label{lem:q_properties}
    Let $Q$ be defined as in \Cref{eq:q_def_appendix}. Let $\{\ket{v_\alpha}\}_{\alpha \in [N]}$ be the eigenbasis of
    $H$. Then,
    \begin{enumerate}
        \item $Q$ is diagonal in the basis $\{\ket{v_\alpha}\}_\alpha$,
        \item (Almost Identity) When $\lambda_\alpha \in [e_0 - \omega/2 + \tfrac c L, e_0 + \omega/2 - \tfrac c L]
            \subseteq \mc I_\omega$, then $\|Q(H) \ket{v_\alpha}\|_2^2 \geq 1 - \frac 1 c$,
            \label{item:almost_identity}
        \item (Almost Zero) When $|\lambda_\alpha - e| \geq \tfrac c L$ for all $e \in \mc I_\omega$, then $\|Q(H)
            \ket{v_\alpha}\|_2^2 \leq \frac 1 c$.
            \label{item:almost_zero}
    \end{enumerate}
\end{lemma}
\begin{proof}
    The first property is clear from the definition of $Q(H)$. To show \Cref{item:almost_identity}, take $\lambda_\alpha
    \in [e_0 - \omega/2 + \tfrac c L, e_0 + \omega/2 - \tfrac c L]$. Since $Q(H)$ is diagonal in $H$'s eigenbasis,
    \begin{align}
        &Q(H) \ket{v_\alpha} = \sum_{m \in L \cdot \textsf{disc}_L(\mc I_\omega)} \text{sinc}_L(\lambda_\alpha - m/L)^2
        \ket{v_\alpha}\nonumber\\
        \label{eq:q_individual_sum}
        &\implies \|Q(H) \ket{v_\alpha}\|_2 = \sum_{m \in L \cdot \textsf{disc}_L(\mc I_\omega)} \text{sinc}_L(\lambda_\alpha - m/L)^2
    \end{align}
    Write $\lambda_\alpha = \tfrac 1 L \round{L \cdot \lambda_\alpha}_L + \delta$, where $\round{L \cdot
    \lambda_\alpha}_L$ is $L \cdot \lambda_\alpha$ rounded to the \emph{nearest} integer and $\delta \in [0, \tfrac 1
    {2L})$ is the remainder\footnote{Rounding to the nearest integer rather than the smaller one is important to ensure
    the worst case is $\text{sinc}_L(1/2)$.}. Each summand in \Cref{eq:q_individual_sum} is then $\text{sinc}_L(\tfrac 1
    L (\round{L \cdot \lambda_\alpha} - m) + \delta)$. Since we've picked $\lambda_\alpha$ to be in the interval $[e_0 -
    \tfrac \omega 2 + \tfrac c L,e_0 + \tfrac \omega 2 - \tfrac c L]$, there exists $m^*$ in the sum such that $m^* =
    \round{L \cdot \lambda_\alpha}$. Furthermore, there are at least $c-1$-values of $m$ both above and below $m^*$.

    Thus, we have
    \begin{align}
        \sum_{m \in L \cdot \textsf{disc}_L(\mc I_\omega)} \text{sinc}_L(\lambda_\alpha - m/L)^2 &\geq \sum_{k = -c + 1}^{c
        - 1}\text{sinc}_L(\tfrac k L + \delta)\nonumber\\
            &= \sum_{k = -c+1}^{c-1} \Paren{\frac{\sin(\pi(k + L\delta))}{L \sin(\pi(\tfrac k L + \delta))}}^2\nonumber\\
            &\geq \sum_{k = -c +1}^{c-1} \Paren{\frac{\sin(\pi(k + L\delta))}{\pi(k + L \delta)}}^2\tag{Since $\sin x \leq x$}\\
            &= \frac{\sin(\pi L\delta)^2}{\pi^2}\sum_{k = -c +1}^{c-1} \frac{1}{(k + L \delta)^2}\nonumber\\
            \label{eq:polygamma}
            &= 1 - \frac{\sin(\pi L\delta)^2(\Gamma^{(1)}(c - L\delta) + \Gamma^{(1)}(c +
            L\delta))}{\pi^2}\\
            &\geq 1 - \frac{\sin(\pi L\delta)^2}{\pi^2}\Paren{\frac 1 {c-L\delta} + \frac{1}{c +
            L\delta}}\tag{$\Gamma^{(1)}(x) \leq \frac{x-1}{x^2}$}
    \end{align}
    where \Cref{eq:polygamma} uses \Cref{lem:polygamma} and that $\Gamma^{(1)}(1-x) = \tfrac{\pi^2}{\sin(\pi x)^2}-
    \Gamma^{(1)}(x)$. Recall that $\delta \in (0,\tfrac 1 {2L})$ and we can upper bound $\tfrac 1 {c-L\delta} + \tfrac 1
    {c+L\delta}$ by $\tfrac 2 c$. Additionally, $\sin(\pi L \delta)^2 \leq \tfrac 1 2 \pi^2$ and so the above can be
    lower bounded by $1 - \tfrac 1 c$.

    Next, we show \Cref{item:almost_zero}. We continue from \Cref{eq:q_individual_sum}. We will use the following upper
    bound on $\text{sinc}_L(x)$
    \[
        \text{sinc}_L(x) = \Paren{\frac{\sin(\pi L x)}{L \sin(\pi x)}}^2 \leq \Paren{\frac 1 {L \sin(\pi x)}}^2 \leq
        \Paren{\frac 1 {L(\pi x - \tfrac 1 3 (\pi x)^3)}}^2 = \Paren{\frac{1}{L\pi x} \cdot \frac{1}{1 - \tfrac 1 3 (\pi
        x)^2}}^2 \leq \frac 4 {(L\pi x)^2}
    \]
    where the last inequality holds as long as $x \leq \frac 1 \pi$. Plugging this back into \Cref{eq:q_individual_sum},
    we see that
    \[
        \|Q(H)\ket{v_\alpha}\|_2 \leq \sum_{m \in L \cdot \textsf{disc}_L(\mc I_\omega)} \frac 4 {(L\pi
        (\lambda_\alpha - m/L))^2} = \frac 4 {\pi^2}\sum_{m \in L \cdot \textsf{disc}_L(\mc I_\omega)} \frac 1 {
        (L \cdot \lambda_\alpha - m)^2}.
    \]
    By assumption, all $m \in L \cdot \textsf{disc}_L(\mc I_\omega)$ satisfies $|\lambda_\alpha - \tfrac
    {m} L| \geq \tfrac c L$, and in particular $|L \cdot \lambda_\alpha - m| \geq c$. Let $x^* = \min |L \cdot
    \lambda_\alpha - m| \geq c$. Observing that $L \cdot \textsf{disc}_L(\mc I_\omega)$ consists of consecutive
    integers, we rewrite the above sum as,
    \begin{align}
        \frac 4 {\pi^2}\sum_{x = x^*}^{x^* + |\textsf{disc}_L(\mc I_\omega)|} \frac 1 {x^2} &= \frac 4 {\pi^2}
        (\Gamma^{(1)}(x^*) - \Gamma^{(1)}(1 + x^* + |\textsf{disc}_L(\mc I_\omega)|)\tag{\Cref{lem:polygamma}}\\
                                                                                            &\leq \frac 4 {\pi^2}
                                                                                            \Gamma^{(1)}(x^*)\tag{$\Gamma^{(1)}(x)
                                                                                            \geq 0$ for $x > 0$}\\
                                                                                            &\leq \frac 4 {\pi^2
                                                                                                x^*}\tag{$\Gamma^{(1)}(x)
                                                                                                    \leq
                                                                                                \tfrac{x-1}{x^2}$}.
    \end{align}
    Recalling that $x^* \geq c$, we may upper bound \Cref{eq:q_individual_sum} as
    \begin{align}
        \|Q(H) \ket{v_\alpha}\|_2 &\leq \frac{4}{\pi^2 c} \leq \frac{1}{c}.
    \end{align}
\end{proof}

With the above lemmas, we can prove the promised claims.
\qmass*
\begin{proof}[Proof of \Cref{clm:Qmass}]
    Since $\Pi_\Delta$ and $Q(H)$ are both diagonal in $H$'s eigenbasis, it suffices to consider each eigenstate of $H$
    in turn. For $\ket{v_\alpha}$ such that $\lambda_\alpha \in \mc I_\Delta$, $\Pi_\Delta$ acts as the identity and
    clearly there is no difference between $Q(H)$ and $\Pi_\Delta Q(H)$. For $\ket{v_\alpha}$ with $\lambda_\alpha
    \not\in \mc I_\Delta$, $\Pi_\Delta \ket{v_\alpha} = 0$ and thus a bound on $\|Q(H)\ket{v_\alpha}\|_2$ yields a bound
    on $\|Q(H) - \Pi_\Delta Q(H)\|$.

    By definition of $Q(H)$, we can write
    \begin{align}
        &Q(H) \ket{v_\alpha} = \sum_{m = (e_0 - \omega/2)L}^{(e_0 + \omega/2)L} \text{sinc}_L(\lambda_\alpha - m/L)^2
        \ket{v_\alpha}\\
        &\implies \|Q(H) \ket{v_\alpha}\|_2 = \sum_{m = (e_0 - \omega/2)L}^{(e_0 + \omega/2)L} \text{sinc}_L(\lambda_\alpha - m/L)^2.
    \end{align}
    We defined $\omega = \Delta - \tfrac 1 {\sqrt L}$ and since $\lambda_\alpha \not\in \mc I_\Delta$, $|\lambda_\alpha
    - \epsilon_0| \geq \Delta/2 \geq \omega/2 + \tfrac 1 {2\sqrt L}$. As a result, $|\lambda_\alpha - \tfrac m L| \geq
    \tfrac 1 {2\sqrt L}$. Using \Cref{item:almost_zero} of \Cref{lem:q_properties} with $c = \frac{\sqrt{L}}{2}$ yields
    that $\|Q(H) \ket{v_\alpha} \|_2 \leq \frac{2}{\sqrt L}$.
\end{proof}

\qdoesntmatter*
\begin{proof}[Proof of \Cref{clm:Qdoesntmatter}]

    Let $\mc Q \dfn Q(H) \otimes Q(H)$. Note that $\ket\Psi = \sum_{\alpha \in \mc I_\Delta} x_\alpha \ket{v_\alpha, v_\alpha}$ is such
    that all coefficients are positive and by \Cref{claim:eq_gap} are within a factor of $\tfrac{f^4} 6$ of each other.
    Define 
    $$S = \{\alpha \st \lambda_\alpha \in [e_0 - \tfrac \omega 2 + \tfrac c L, e_0 + \tfrac \omega 2 - \tfrac c L]\}$$
    and
    $$\overline S = \{\alpha \st \lambda_\alpha \in [e_0 - \tfrac \Delta 2, e_0 - \tfrac \omega 2 + \tfrac c L) \cup
            (e_0 + \tfrac \omega 2 - \tfrac c L, e_0 + \tfrac \Delta 2]$$
    Thus,
    \begin{align*}
        \bra\Psi \mc Q \ket\Psi &= \sum_\alpha x_\alpha^2 \cdot \|Q \otimes Q \ket{v_\alpha, v_\alpha}\|_2 = \sum_\alpha x_\alpha^2 \cdot \|Q\ket{v_\alpha}\|_2^2\tag{$Q$ is diagonal w.r.t. $\{\ket{v_\alpha}\}_\alpha$}\\
        &\geq \Paren{1-\frac{1}{c}}\sum_{\alpha\in S} x_\alpha^2\tag{\Cref{item:almost_identity} of \Cref{lem:q_properties}}\\
        &=\Paren{1-\frac{1}{c}}\Paren{1-\sum_{\alpha\in \overline S} x_\alpha^2}\\
        &\geq \Paren{1-\frac{1}{c}}\Paren{1- |\overline S|x_{\max}^2}\tag{$x_{\max} \dfn \max_\alpha x_\alpha$.}
    \end{align*}
    However, $\sum_\alpha x_\alpha^2=1$, which means $x_{\max}^2(|S|+|\overline S|)\leq \frac{6}{f^4}$. Thus, $\bra\Psi
    Q \ket\Psi \geq (1-1/c)\big(1- \frac{6}{f^4}\frac{|\overline{S}|}{|S|+|\overline S|}\big)$. By \Cref{as:smooth} we
    have that for $c= \sqrt{L}$, $\frac{|\overline{S}|}{(|S|+|\overline S|)}\leq 1 - C'$, which implies 
    \[
        \bra\Psi Q \ket\Psi \geq \Paren{1-\frac 1 {\sqrt{L}}}\Paren{1- \frac{6(1-C')}{f^4}} \geq 1 - \frac 1 {\sqrt L} -
        \frac{6(1-C')}{f^4}.
    \]
This completes the proof.
\end{proof}

\section{Lipschitz property of the Ky-Fan 2 Norm}
\label{sec:lipschitz}
Recall that $V^S$ is a (POVM for a) $L$-query oracular verifier, so that,
    \begin{align}
        \label{eq:apdx_accept_proj}
        V^S &= (\ketbra 0 \otimes \Id) U_0^\dagger (\Pi_\tcon \otimes \O^S) U_1^\dagger \dots
        U_{L-1}^\dagger (\Pi_\tcon \otimes \O^S) U_L^\dagger\cdot \nonumber\\
        &\qquad (\Id \otimes \ketbra 1) U_L \O^S U_{L-1} \dots U_1 (\Pi_\tcon \otimes \O^S) U_0
        (\ketbra 0 \otimes \Id)\,.
    \end{align}
The Ky-Fan $k$ norm is defined as the 1-norm of the top $k$ singular values of an operator and naturally arises when considering operators $V^S$ which correspond to $\UQMA$ verifiers. We show that the Lipschitz constant of these operators is bounded.
\kyfanlipschitz*

\begin{proof}
    To bound $|F(\Pi_S) - F(\Pi_{S'})|$, we will hybridize between $S$ and $S'$. Define ``forward'' and ``backward''
    hybridized circuits,
    \begin{align*}
        C^\dagger[\ell] &\dfn (\ketbra 0 \otimes \Id) U^\dagger_0 (\Pi_\tcon \otimes \O^S) U_1^\dagger \dots U_{L-\ell}^\dagger
        (\Pi_\tcon \otimes \O^S) U_{L-\ell+1}^\dagger (\Pi_\tcon \otimes \O^{S'}) \dots U_L^\dagger (\Id \otimes \ketbra 1)\\
        C[\ell] &\dfn (\Id \otimes \ketbra 1) U_L (\Pi_\tcon \otimes \O^S) U_{L-1} \dots U_{\ell} 
        (\Pi_\tcon \otimes \O^S) U_{\ell-1} (\Pi_\tcon \otimes \O^{S'}) \dots  U_0 (\ketbra 0 \otimes \Id)
    \end{align*}
    where for both $C[\ell]$ and $C^\dagger[\ell]$, $\ell$ denotes the number of oracle calls from the right which are
    $S'$ (the remaining being to $S$). Therefore, to hybridize $V^S_\tacc$ and $V^{S'}_\tacc$, we proceed as,
    \begin{align*}
        V^S_\tacc &= C^\dagger[0] C[0] \rightarrow C^\dagger[0] C[1] \rightarrow \dots \rightarrow C^\dagger[0] C[L]
        \rightarrow C^\dagger[1] C[L] \rightarrow \dots \rightarrow C^\dagger[L] C[L] = V^{S'}_\tacc
    \end{align*}
    Then,
    \begin{align*}
        |F(V^S) - F(V^{S'})| &= |\|V^S_\tacc\|_{\uparrow 2} - \|V^{S'}_\tacc\|_{\uparrow 2}|\\
                                 &\leq \|V^S_\tacc - V^{S'}_\tacc\|_{\uparrow 2}\\
                                 &\leq \Norm{\sum_{\ell=0}^{L-1} (C^\dagger[0] C[\ell] - C^\dagger[0] C[\ell+1]) +
                                 \sum_{\ell=0}^{L-1} (C^\dagger[\ell] C[L] - C^\dagger[\ell+1] C[L])}_{\uparrow 2}\\
                                \label{eq:apdx_norm_bound}
                                 &\leq \sum_{\ell=0}^{L-1} \|C^\dagger[0] C[\ell] - C^\dagger[0] C[\ell+1]\|_{\uparrow
                                 2} + \sum_{\ell=0}^{L-1} \|C^\dagger[\ell] C[L] - C^\dagger[\ell+1] C[L]\|_{\uparrow
                                 2}\numberthis
    \end{align*}
    Consider the first summand. Note that besides the projectors $\ketbra 0$ and $\ketbra 1$ at the beginning and end of
    $C^\dagger[0]$, the remainder of the terms are unitary. Write $C^\dagger[0] = \Pi_0 U \Pi_1$. Using
    sub-multiplicativity we write,
    \[
        \|\Pi_0 U \Pi_1 (C[\ell] - C[\ell+1])\|_{\uparrow 2} \leq \| \Pi_0 \|_{\uparrow 2} \cdot \|U \Pi_1 (C[\ell]
        - C[\ell+1])\|_{\uparrow 2} \leq 2 \|U \Pi_1(C[\ell] - C[\ell+1])\|_{\uparrow 2}\,,
    \]
    where we used that since $\Pi_0$ is a projector, $\|\Pi_0\|_{\uparrow 2} \leq 2$. Next, since the Ky Fan norm is
    unitarily invariant, $U$ does not affect the norm, and we can accrue another factor of $2$ to remove the $\Pi_1$
    term as well. Thus, for any $\ell$
    \[
        \|C^\dagger[0] C[\ell] - C^\dagger[0] C[\ell+1]\|_{\uparrow 2} \leq 4\|C[\ell] - C[\ell+1]\|_{\uparrow 2}\,.
    \]

    For this final term, we again remove the leading and trailing projector yielding a multiplicative penalty of $2$,
    then we are left with
    \[
        \|U_L \O^S U_{L-1} \dots U_\ell (\Pi_\tcon \otimes (\O^S - \O^{S'})) U_{\ell-1} \O^{S'} U_{\ell-2} \dots
        U_0\|_{\uparrow 2}\,.
    \]
    Everything except the inner $\Pi_\tcon \otimes (\O^S - \O^{S'})$ terms can be discarded, as they are unitary and it
    suffices to bound $\|\Pi_\tcon \otimes (\O^S - \O^{S'})\|_{\uparrow 2}$. Note that since $\Pi_\tcon$ is a projector,
    it simply copies the singular values of $\O^S - \O^{S'}$ $\dim(\Pi_\tcon)$-times. But we're only taking the top 2
    eigenvalues, so this yields a multiplicative penalty of $2$. In conclusion, we see that
    \[
        \|C^\dagger[0] C[\ell] - C^\dagger[0] C[\ell+1]\|_{\uparrow 2} \leq 16 \|\O^S - \O^{S'}\|_{\uparrow 2}  = 32 \|
        \Pi^S - \Pi^{S'}\|_{\uparrow 2}
    \]

    The final step is to relate the Ky Fan 2-norm to the Frobenius norm. Recall each of the norms' definitions in terms
    of singular values, we note that,
    \[
        \sigma_1 + \sigma_2 \leq \sqrt 2 \sqrt{\sigma_1^2 + \sigma_2^2} \leq \sqrt 2 \sqrt {\sum_i \sigma_i^2} = \sqrt 2 \|\Pi_S
        - \Pi_{S'}\|_2
    \]
    Applying this argument to each of the $L$ terms, and an analogous argument to the second summand in
    \Cref{eq:apdx_norm_bound} yields a final bound of
    \[
        |F(\Pi_S) - F(\Pi_{S'})| \leq 64\sqrt 2 L \|\Pi_S - \Pi_{S'}\|_2
    \]
    
\end{proof}

\section{On the number of efficiently verifiable states}
\label{append:countingverifiable}
In this appendix we argue that there are comparably few states that can be uniquely verified efficiently. 
More precisely, we show the following statement: 
\begin{prop}
    Let $\delta>0$ and $p(n)$ a polynomial in $n$.
    Further, denote by $S_n$ the set of all states $|\psi\rangle\in (\mathbb{C}^2)^{\otimes n}$ that are verifiable by a quantum circuits with $p(n)$ $2$-local gates.
    Then, $S_n$ can be covered by $O_{\delta}(e^{np(n)})$ many balls of radius $\delta$.
\end{prop}
\begin{proof}
    A state $|\psi\rangle$ is verifiable if and only if a quantum circuit $C$ accepts $|\psi\rangle$ with probability $\geq \frac23$ and rejects all state orthogonal to $|\psi\rangle$ with probability $\geq \frac23$.
    We can amplify this probability with majority vote to $1-e^{-\Omega(k)}$ by increasing the gate count by a factor of $k$.
    
    By standard arguments~\cite{milman1997global} an $\varepsilon$-net $N_{\varepsilon}$ on $SU(4)$ of size $O((5/\varepsilon)^{15})$ exists.
    We pick $\varepsilon:=2^{-n}$.
    Now consider any state $|\psi\rangle\in S_n$ and its verifying circuit $C_{\psi}$ with at most $kp(n)$ gates.
    Then, by a standard telescope argument~\cite{Bennett_BBBV1997_StrengthsWeaknessesQuantum}, there is a circuit $C'_{\psi}$ with gates from $N_{\varepsilon}$, which generates states that are $kp(n) 2^{-n}$ close in infidelity. 
Then, $C'_{\psi}$ accepts $|\psi\rangle$ with probability at least $1-e^{-\Omega(k)}-kp(n)2^{-n}$ and rejects every orthogonal state with probability at least $1-e^{-\Omega(k)}-kp(n) 2^{-n}$.
    For every such verifying circuit with gates in $N_{\varepsilon}$ we pick one state that is verified with high probability. 
    The resulting set is called $T_{\varepsilon,k}$
    By the above argument, any state in $S_n$ needs to have an overlap of $1-O(e^{-\Omega(k)}+kp(n) 2^{-n})$ with a state from $T_{\varepsilon,k}$.
     It then suffices to upper bound 
    \begin{equation}
        |T_{n,k}|\leq |N_{\varepsilon}|^{kp(n)}\leq e^{O(nkp(n))}.
    \end{equation}
    Picking $k$ a sufficiently large constant completes the argument. \qedhere
\end{proof}

\end{document}

%% file: preliminaries.tex
\paragraph{Notation} A register $\reg{R}$ is a named finite-dimensional complex Hilbert space. If not otherwise specified, $N \dfn 2^n$, with $n$ defined as appropriate in context. In this work we work with both \emph{classical} subspaces of the form $S \subseteq \{0,1\}^n$ as well as \emph{quantum} subspaces $\qS \subseteq \C^N$. To make this distinction more clear, we use standard letters for classical subspaces $S$, and $\mathcal{CALLIGRAPHIC}$ letters for quantum subspace $\qS$. Additionally, the quantity $\exp(f(n))$ is understood to be the class of functions $2^{\Omega(f(n))}$ and $\cpoly(n) \dfn \bigcup_{k \in \mathbb N} \O(n^k)$.

\subsection{Complexity classes}

\begin{definition}[$\QMA$]
    A promise problem $L = (L_\text{yes}, L_\text{no})$ is in the complexity class \emph{Quantum Merlin-Arthur} (or $\QMA$) if there exists a polynomial-time classical algorithm yielding a description of a quantum circuit $V$ which implements a quantum operator $U_x : \mc H \rightarrow \mc H$ with $\mc H = \mc H_{\reg R_\tin} \otimes \mc H_{\reg R_\taux}$, such that
    \begin{enumerate}
        \item (Completeness) For all $x \in L_\text{yes}$, there exists a \emph{witness} $\ket \phi \in \mc H_{\reg R_\tin}$ such that
        \[
            \|\Pi_\text{acc} U_x \ket \phi_{\reg R_\tin} \ket 0_{\reg R_\taux}\|_2 \geq \frac 2 3
        \]
        \item (Soundness) For all $x \in L_\text{no}$ and $\ket \phi \in \mc H_{\reg R_\tin}$, 
        \[
            \|\Pi_\text{acc} U_x \ket \phi_{\reg R_\tin} \ket 0_{\reg R_\taux}\|_2 \leq \frac 1 3
        \]
    \end{enumerate}
    where $\Pi_\text{acc}$ can be taken to be the projector into the $\ket 0$ on the first qubit.
\end{definition}

\begin{definition}[Unique $\QMA$]
    A promise problem $L = (L_\text{yes}, L_\text{no})$ is in the complexity class \emph{Unique Quantum Merlin-Arthur} (or $\UQMA$) if in addition to the above properties of $\QMA$, we have the further restriction that in the \emph{completeness} case, there is some accepting state $\ket \phi$ and for all states orthogonal to $\ket \phi$, the verifier accepts with probability equal to the soundness case. In particular,
    \begin{enumerate}
        \item (Completeness) For all $x \in L_\text{yes}$, there exists a \emph{unique witness} $\ket \phi \in \mc H_\text{witness}$ such that
        \[
            \|\Pi_\text{acc} U_x \ket \phi_{\reg R_\tin} \ket 0_{\reg R_\taux}\|_2 \geq \frac 2 3
        \]
        and for all states $\ket \psi \perp \ket \phi$,
        \[
            \|\Pi_\text{acc} U_x \ket \psi_{\reg R_\tin} \ket 0_{\reg R_\taux}\|_2 \leq \frac 1 3
        \]
    \end{enumerate}
\end{definition}

Going forward, we rarely explicitly refer to $\Pi_\text{acc}$ and $U_x$. Instead, for both the above classes, we use the notation $V_x$ to refer to both the verification circuit, as well as the POVM measurement corresponding to the accept outcome. Thus, the completeness condition could be rewritten as the existence of $\ket \phi$ such that $\|V_x \ket \phi \ket 0 \|_2 \geq \tfrac 2 3$. Note we omit the register subscripts if clear from context, and write $\ket 0$ to refer to the $\log\dim(\mc H_{\reg R_\taux})$-qubit all-zeros state. Additionally, define the function $\textup{acc}(V_x, \psi) \dfn \|V_x \ket \psi \ket 0\|_2$ and $\textup{acc}(V_x) \dfn \max_{\text{witnesses } \ket \psi} \textup{acc}(V_x, \psi)$. If $V_x$  takes in no witness, then $\text{acc}(V_x)$ is simply the acceptance probability of $V_x$ on input $x$.

A natural complete problem for $\QMA$ is satisfiability of a polynomial-size quantum circuit.\footnote{Circuit vs. verifier, as a circuit $\ttC$ may not correspond to a language.} Notationally, we denote circuits in \texttt{teletype}. We overload notation and for a circuit taking in an input state $\ket \psi$, we write $\text{acc}(\ttC, \psi)$ to denote the probability that $\ttC$ accepts on $\ket \psi$ and $\textup{acc}(\ttC)$ as the maximum over all inputs $\psi$. The natural complete problem for $\UQMA$ is the unique circuit SAT problem.

\begin{definition}
    \label{def:unique_cirsat}
    Let $\texttt C$ be a $\cpoly(n)$-size quantum circuit, taking in a $n$-qubit witness state. $\textsc{\textup{UniqueQCircuitSat}}(c,s)$ is the decision problem of deciding whether,
    \begin{itemize}
        \item (YES case) There exists a witness $\ket \psi$ such that $\text{acc}(\texttt C, \psi) \geq c$ and for all orthogonal witnesses $\ket \phi$, $\text{acc}(\texttt C, \psi) \leq s$, or
        \item (NO case) $\text{acc}(\texttt C) \leq s$.
    \end{itemize}
\end{definition}

In \cite{Jain_JKK+2011_PowerUniqueQuantum}, it was shown that this problem is indeed complete for $\UQMA$.
\begin{theorem}
    \label{thm:unique_cirsat}
    $\textsc{\textup{UniqueQCircuitSat}}(c,s)$ is $\UQMA$-complete with $c - s \geq n^{-k}$ for some constant $k$.
\end{theorem}

Many of the arguments used in this paper involve showing concentration about the operator norm of a POVM. In the case of $\UQMA$, its acceptance probability depends not only on its highest eigenvalue, but it's second highest as well. Thus, we will consider the Ky Fan 2-norm
of a matrix $M$.

\begin{definition}[Ky Fan 2-norm, $\|M\|_{\uparrow 2}$]
    For a matrix $M \in \C^{N \times N}$, the Ky Fan 2-norm is defined as
    \[
        \|M\|_{\uparrow 2} = \sigma_1(M) + \sigma_2(M)\,,
    \]
    where $\sigma_i(M)$ is the $i$-th singular value of $M$. Alternatively, we have the following variational
    definition:
    \[
        \|M\|_{\uparrow 2} = \max_{\substack{U,V \in \C^{2 \times N}\\U U^\dagger = V V^\dagger = \Id_{2 \times 2}}}
        \Trace[U M V^\dagger]\,.
    \]
\end{definition}

The equivalence of these definitions can be seen by taking the singular value decomposition of $M$. The variational definition of the Ky Fan $2$-norm is for showing the following ``Cauchy-Schwarz''-like property.

\begin{lemma}[Properties of the Ky Fan 2-norm]
    The Ky Fan 2-norm is a norm and therefore satisfies the  triangle inequality and for any scalar $t$, $\|t \cdot
    M\|_{\uparrow 2} = |t| \cdot \|M\|_{\uparrow 2}$. Moreover, the Ky Fan 2-norm is \emph{sub-multiplicative} and
    satisfies
    \[
        \|MN\|_{\uparrow 2} \leq \|M\|_{\uparrow 2} \cdot \|N\|_{\uparrow 2}\,.
    \]
\end{lemma}

Finally, $\QXC$, introduced in \cite{Bravyi_BCGW2022_ComplexityQuantumPartition}, is the complexity class which naturally contains quantum approximate counting.

\begin{restatable}[Quantum approximate counting, $\QXC$]{definition}{qxcclass}
Let $V$ be a $m$-sized $\QMA$ verification circuit which takes a $n$-qubit state $\ket\psi$ as input, with $m \in \cpoly(n)$. Then, given thresholds $0 <
b < a \leq 1$ with $a - b \geq \frac 1 {\cpoly(n)}$, and an error parameter $\eps \geq \frac 1 {\cpoly(n)}$, the approximate counting problem $\QXC$ is to
compute an estimate $D$ such that $(1-\eps)D_a \leq D \leq (1+\eps)D_b$, where $D_{a}$ is defined as the dimension of the witness subspace for which $V$ accepts with probability at least $a$, respectively for $D_b$.
\end{restatable}

\subsection{Quantum oracles}

Our first result is a oracle separation between $\QMA$ and $\UQMA$. There have been various types of oracles considered in previous works. Although the ``gold standard'' for oracle separations is arguably a classical oracle (which can be queried in superposition), quantum separations have required more complex oracles, such as the distributional oracles of \cite{natarajan2024distribution}. In this work, we work with quantum oracles, as studied in \cite{Aaronson_AK2007_QuantumClassicalProofs}. Thus, in our setting, an oracle $\O$ refers to some arbitrary unitary operation. We will also allow controlled queries to $\O$, so that a call to $\O$ is in fact a call to the unitary,
\[
(\Id - \Pi) \otimes \Id + \Pi \otimes \O\,,
\]
where $\Pi$ is a projector representing a control subspace. Given a family of oracles, we can define a ``oracle promise problem,''
\begin{definition}[oracle promise Problem]
    \label{def:oracle_decision_problem}
    A oracle promise problem is defined by two disjoint sets of oracles (unitary operators), $\Omega_\text{YES}$ and $\Omega_\text{NO}$. We use $\Omega = (\Omega_\text{YES},\Omega_\text{NO})$ to refer to the decision problem itself.
\end{definition}
Then, any $L$-query algorithm deciding $\Omega$ should be able to receive an oracle $\O$ either in the YES or NO case and use $L$ calls to $\O$ to decide which is the case. Formally, a oracular verifier for $\UQMA^\O$ consists of alternating unitaries (consisting of local quantum gates) and calls to $\O$. The task is to (uniquely) determine whether the unknown oracle corresponds to a YES case or NO case instance.

\begin{definition}[$\UQMA$ Oracle Verifier]
    \label{def:unique_verifier}
    Given a $n$-qubit oracle $\O$ and $m \in \cpoly(n)$, a $m$-qubit, $L$-query $\UQMA$ oracular verifier $V^\O$ is defined over a $n$-qubit oracle register $\reg R_\O$, in a $n_\taux$ qubit auxiliary register $\reg R_\taux$. Additionally, there is a $n_\tin$-qubit input register $\reg R_\tin$. Thus, $m = n_\taux + n_\tin$. Given a $n_\tin$-qubit witness, the verifier then performs a sequence of unitaries on the full $m$-qubit subspace, interwoven with $L$ queries to an oracle $\O$ on $\reg R_\O$, controlled on the remaining system. Finally, the verifier performs the measurement $\{\ketbra 0, \ketbra  1\}$ on the first qubit in the auxiliary register (which we label as the register $\reg R_\tout$), rejecting if it obtains $\ketbra 0$ the and accepting if obtains $\ketbra 1$. See~\Cref{fig:unique-verifier}. Thus, the measurement corresponding to the ``accept'' case can be written as
    \begin{align}
        V &\dfn (\Id \otimes \ketbra 1_{\reg R_\tout}) U_L \O U_{L-1} \dots U_1 \O U_0 (\Id \otimes \ketbra{0^{n_\taux}}_{\reg R_\text{aux}}),\nonumber\\
        \label{eq:first_def}
        \Pi_\text{accept} &\dfn V^\dagger V.
    \end{align}
    We say that $V$ \emph{uniquely solves} the oracle promise problem $\Omega = (\Omega_\text{YES}, \Omega_\text{NO})$ if
    \begin{itemize}
        \item If $\O \in \Omega_\text{YES}$ then there exists a witness $\ket{\psi}_{\reg R_\tin}$ so that 
        \[
        \textup{acc}(V, \psi) = \|V \ket \psi \ket 0_{\reg R_\taux}\|^2_2 = \tr[\Pi_\text{accept} \ketbra \psi_{\reg R_\tin} \otimes \ketbra 0_{\reg R_\taux}] \geq 1 - \eps\,,
        \]
        and for any witness $\ket{\psi'}$ where $\ket{\psi'}$ is orthogonal to $\ket{\psi}$, $\textup{acc}(V, \psi') \leq \eps$.
        \item If $\O \in \Omega_\text{NO}$ then $\textup{acc}(V) \leq \eps$; for all witnesses the verifier accepts with probability at most $\eps$.
    \end{itemize}
\end{definition}
\input{figs/circuit}

\begin{remark}[Error Reduction]
\label{rem:error_reduction}
In \cite{Jain_JKK+2011_PowerUniqueQuantum}, it was observed that the in-place amplification scheme of \cite{Marriott_MW2005_QuantumArthurMerlinGames} works for $\UQMA$. In particular, we may assume that $\eps \in \O(1/\exp(n))$, at the expense of polynomial blowup in $L$, the number of queries to the oracle.
\end{remark}

\begin{remark}[Controlled Oracles]
\label{rem:controlled_oracle}
Without loss of generality, we may assume that each oracle $\O$ is controlled on $\ketbra 0$, on the last qubit of the auxiliary register. This is because we can write a general projector as $\Pi = U \big(\sum_{x \in [\mathrm{rank}(\Pi)]} \ket{x}\bra{x} \big) U^\dagger$. The $U$ and $U^\dagger$ can be absorbed into the unitaries preceding and following the oracle call. The computational basis projector can be turned into one-qubit control by: compute the indicator function $\mathbf{1}_{x \in [\mathrm{rank}(\Pi)]}$ and write output on an ancilla qubit, apply $\mc O$ controlled on the output qubit, and uncompute. As such, we can choose the control subspace as $\Pi = \ketbra 0 \otimes \Id$.
\end{remark}

The standard oracle that we consider in this work are variations of an oracle considered in prior works: the subspace reflection oracle $\O^\qS= \Id - 2\Pi_\qS$, where $\Pi_\qS$ corresponds to a projection onto a $k$-dimensional subspace $\qS$. In \cite{Aaronson_AK2007_QuantumClassicalProofs}, $\O^\qS$ with $k=1$ was used to separate $\QCMA$ from $\QMA$; in \cite{Aaronson_AKKT2020_QuantumLowerBounds}, $\qS$ was restricted to be a classical subspace to separate $\QMA$ from classical approximate counting; and finally \cite{she2023unitary} quantized the prior work to separate $\QMA$ from quantum approximate counting. We borrow notation from these last two works to define the oracle promise problems we study.

\begin{definition}[Quantum approximate dimension]
    In the oracle promise problem $\oprob$, the problem is to decide whether an $n$-qubit unknown oracle corresponds to a
    reflection about a subspace of dimension $\geq k_2$ (the YES case), or a subspace of dimension $\leq k_1$ (the NO
    case), promised that one of these two is the case. For technical reasons, we restrict $k_2 \leq 2^{2n/3}$.

    If we are further promised that the subspace is in the computational basis, then we have the classical problem
    $\ApxDim(k_1,k_2,N)$.
\end{definition}

\begin{remark}[Classical oracles]
\label{rem:classical_oracles}
We use the term ``classical oracle'' to refer to phase oracles $\O^S = \Id - 2 \Pi_S$, where $S$ is restricted to be in the computational basis. The truly ``classical'' analogue would be the oracle mapping $\ket x \ket b \rightarrow \ket x \ket{ b \oplus \mathbf 1_{x \in S}}$. However, as our separations holds under controlled oracles, these notions are equivalent. 
\end{remark}

In this work, we also consider oracle separations between relativized classes, e.g. $\BQP^\UQMA$. We instantiate verifiers belong to a relativized class ${\mc C_1}^{\mc C_2}$ as a verifier for $\mc C_1$, with access to an black-box oracle $\mc O$ for a $\mc C_2$-complete problem. In our case, $\mc C_2$ is either $\UQMA$ or $\QMA$ and thus $\mc A$ is a classical oracle and the corresponding complete problem will be $\textsc{UniqueQCircuitSat}$ or $\textsc{QCircuitSat}$ respectively. However, one difficulty in defining it in this way is that these languages correspond \emph{promise} problems, but there is not guarantee that an algorithm querying $\mc A$ will do so only on valid instances. To deal with this issue, we take an approach used in prior works \cite{aaronson2022acrobatics, irani2021quantum} and imagine that the outer verifier is given access to an oracle $\mc A$ for a decision ``extension'' of the promise problem.
\begin{definition}[Promise oracle extension]
\label{def:oracle_extension}
For any oracle given by the partial function $\O : \{0,1\}^n \rightarrow \{0,1, \perp\}$, an oracle extension is a total function $\mc A : \{0,1\}^n \rightarrow \{0,1\}$ such that $\O(x) = \A(x)$ for all $x \in \text{Dom}(\O)$.
\end{definition}
Thus, we take the following definition of a verifier accessing a promise oracle $\O$, as suggested in \cite{aaronson2022acrobatics}:
\[
V^\O(x) \dfn \begin{cases}
    1 & \text{if }V^\A(x) = 1\text{ for every $\A$ extending $\O$,}\\
    0 & \text{if }V^\A(x) = 0\text{ for every $\A$ extending $\O$,}\\
    \perp & \text{otherwise.}
\end{cases}
\]
As stated in \cite{aaronson2022acrobatics}, it is not obvious that this is the ``right'' way to define oracle access to a promise problem. Nonetheless, we believe this to be a fairly reasonable definition; when querying an oracle for, e.g., \textsc{QuantumCircuitSat} on an instance falling outside of the promise gap, why should the verifier expect a ``reasonable'' response from oracle?

Lastly, we remark that whenever we have a verifier of the form ${\mc C_1}^{\mc C_2}$, we assume that the oracle for ${\mc C_2}$ acts on instances of size polynomial in the outer problem instance. This is required technically as many of our arguments go through union bounding over possible inputs to the oracle for $\mc C_2$. Moreover, as noted in \cite{Aaronson_AKKT2020_QuantumLowerBounds}, there is a distinguisher between (classical) subspaces of dimension $k$ and $2k$ with a witness of size roughly $\O(k)$. Thus, without a bound on the witness register, this distinguisher would preclude a $\P^\QMA$ query lower bound for approximate counting.

\subsection{Local Hamiltonians}

A local Hamiltonian is a Hermitian matrix of the form $H=\sum_i h_i$, with bounded operator norm $\|h_i\|\leq 1$ and each $h_i$ acts on a constant number of qubits. We will also assume that each qubit is acted upon a constant number of local terms. This definition allows the Hamiltonian to act on a high dimensional geometry in a local way. We say that a Hamiltonian is \emph{geometrically-local} if it is defined on a low-dimensional lattice.

\begin{definition}[$(c,s)$-Local Hamiltonian Problem]
    Consider a family of Hamiltonians $\mc H$ and parameters $a,b$ with $b - a \geq \tfrac 1 {\cpoly(n)}$. For any $H \in \mc H$, let $\lambda_0(H) \leq \lambda_1(H) \leq \dots \leq \lambda_N(H)$ denote the eigenvalues of $H$. Then, given $H \in \mc H$, the $(c,s)$-local Hamiltonian problem is the problem of deciding between the two cases, promised that one of the cases holds,
    \begin{enumerate}
        \item (YES) $\lambda_0(H) \leq c$, or
        \item (NO) $\lambda_0(H) \leq s$.
    \end{enumerate}
\end{definition}
The local Hamiltonian problem was shown to be complete for the complexity class $\QMA$ by Kitaev \cite{Kitaev_KSV2002_ClassicalQuantumComputation}. Much like the Local Hamiltonian problem $\LH$ is complete for $\QMA$, we have the following natural complete problem for $\UQMA$.
\begin{definition}[$(c,s)$-Unique Local Hamiltonian Problem]
    \label{def:uqma}
    Consider a family of local Hamiltonian $\mc H$, as above. Then, the $(c,s)$-\emph{Unique Local Hamiltonian Problem} (or $\ULH(c,s)$) is the problem of deciding between the two cases, promised that one of the cases holds.
    \begin{enumerate}
        \item (YES) $\lambda_0(H) \leq c$ and $\lambda_1(H) \geq s$, or
        \item (NO) $\lambda_0(H) \geq s$.
    \end{enumerate}
\end{definition}

Then by a straightforward adaptation of Kitaev's proof, we have that
\begin{theorem}[$\ULH$ is Complete for $\UQMA$]\label{thm:ULH}
    $\ULH(c,s)$ is complete for $\UQMA$ with $s - c \geq n^{-k}$ for some constant $k$.
\end{theorem}

\subsection{Probability and concentration}
Many of the arguments in this paper involve drawing random $k$-dimensional subspaces and showing that their intersection
with a fixed state is small on average. 
\begin{comment}
To argue this formally, we will need a version of the Geometric Lemma of
\cite{Aaronson_AK2007_QuantumClassicalProofs} for random subspaces.

\begin{restatable}[Subspace geometric lemma]{lemma}{subspacegeometric}
    \label{lem:subspace_geom_lemma}
    Let $\sigma$ be a $p$-uniform probability distribution over $\C^N$ with $p \geq 2^{-N^{1-c}}$ for constant $c \in
    (1/2,1)$, and suppose $k = \cpoly(N)$. Then,
    \[
        \E_{\substack{\psi_1, \dots, \psi_k \sim \sigma\\|\braket{\psi_i}{\psi_j}| \leq \frac 1 {N^{-c}}}} \tr\Brac{\sum_i
        \ketbra{\psi_i} \rho} \leq \O\Paren{\frac{k \cdot (1 + \log 1/p)}{N}}
    \]
\end{restatable}
\end{comment}
To argue this formally, we use a version of Levy's Lemma which holds for the
\emph{Grassmannian} $G_{\C^N,k}$, which is the set of dimension-$k$ subspaces of $\mathbb C^{N}$. In the literature, $G_{\C^N, k}$ is often used to refer to the set  of projectors onto the corresponding subspaces. In particular,
\[
    G_{\C^N,k} \dfn \{\Pi \in \C^{N \times N} \mid \Pi^2 = \Pi, \rank(\Pi) = k\}\,.
\]
In this work, we will work with subspaces and the corresponding projectors interchangeably, using a letter like $\qS$ to refer to the subspace, and $\Pi_\qS$ to refer to the subspace projector. We will abbreviate $G_{\C^N, k}$ to $G_{N,k}$. The unique unitarily-invariant Haar measure over $G_{N,k}$ is denoted $\haar_{N,k}$. \cite{gotze2023higher} gives the following analogue of Levy's Lemma over the Grassmannian:
\begin{lemma}[Levy's lemma on the Grassmannian \cite{gotze2023higher}]
    \label{lem:levyslemma}
    Let $F : G_{N,k} \rightarrow \R$ be a $\kappa$-Lipschitz function defined on the Grassmannian. Then,
    \[
        \Pr_{\qS \sim \haar_{N,k}}[|F(\qS) - \E_{\qS'} F(\qS')| \geq \eta] \leq \exp\Paren{-
        \frac{(N-1)\eta^2}{16 \kappa^2}}\,.
    \]
    $F$ is $\kappa$-Lipschitz if for all $\qS, \qS' \in G_{N,k}$, $|F(\qS) - F(\qS')| \leq \kappa \|\Pi_{\qS} - \Pi_{\qS'}\|_2$,
    where $\| \cdot \|_2$ is the Frobenius norm.
\end{lemma}

Additionally, we introduce the notion of subspace extensions. Given a $k_1$ dimension subspace $\qS \subseteq \C^N$, we define a dimension $k_2 > k_1$ extension $\qT$ of $\qS$ as $\qT = \qS \oplus \Delta$, with $\dim(\Delta) = k_2 - k_1$. Moreover, we define the Haar-random distribution over dimension $k_2$ extensions of $\qS$ as $\mc Q_{\uparrow \qS, k_2}$ with,
\[
   \Pr_{T \sim \mc Q_{\uparrow \qS, k_2}}[\qT = \qT^\star] \dfn \Pr_{\Delta \sim \text{Haar}\Paren{G_{\C^N \setminus \qS, k_2 - k_1}}}[\qS \oplus \Delta = \qT^\star]\,.
\]
Extensions $\T$ for classical subspaces $\S$ are defined similarly, but where $\Delta \subseteq \{0,1\}^n$. Notationally, we write this as $\T \sim \mc C_{\uparrow \S, k_2}$ to differentiate it from the quantum distribution.

Most often, we will use the above lemma to bound the probability that a quantum state $\ket \psi$ overlaps a Haar-random
subspace $\qS$ drawn from a $N$-dimensional ambient space.

\begin{lemma}[Concentration of overlap]
    \label{lem:overlap_concentration}
    Let $\qS \sim \mc G_{N, k}$ be a Haar-random $k$-dimensional subspace of $\C^N$ with $k \leq N^{2/3}$ and $\Pi_\qS$ be the projector onto
    $S$. Given any quantum state $\ket \psi \in \C^{N'}$, with $N' \geq N$, we have that
    \[
        \Pr_\qS[\tr[\Pi_\qS \ketbra \psi] \geq 2N^{-1/3}] \leq \exp(-N^{1/3})\,.
    \]
\end{lemma}

\begin{proof}
    Define $f(\qS) = \tr[\Pi_\qS \ketbra \psi]$. Then, $\E_\qS f(\qS) = \frac k N \leq N^{-1/3}$. We bound the Lipschitz
    constant of $f$ as,
    \[
        |f(\qS) - f(\qS')| \leq |\tr[(\Pi_\qS - \Pi_{S'}) \ketbra \psi]| \leq \|\Pi_\qS - \Pi_{S'}\|_\text{op} \|\ketbra \psi
        \|_1 = \|\Pi_\qS - \Pi_{S'}\|_\text{op}\,.
    \]
    Applying \Cref{lem:levyslemma} with $\eta = N^{-1/3}$ and noting that $\E_S [ \tr[\Pi_S \ketbra \psi]] = \frac k N\leq N^{-1/3}$, we find that $\Pr_{S} [ \tr[\Pi_\qS \ketbra \psi] \geq 2
    N^{-1/3}] \leq \exp(-N^{1/3})$ as desired.
\end{proof}

We apply this lemma to show that quantum circuits cannot distinguish between oracles and their random extensions.

\begin{lemma}[Quantum oracle indistinguishability]\label{lem:ext_indistinguishable} Suppose $k' \leq N^{2/3}$ and $k \in o(k')$. Let $V$ be an $L$-query oracle verifier and $\ket{\phi}$ be a quantum state. Let $\qS$ be a $k$-dimensional subspace and $\qT \sim \mc Q_{\qS \uparrow, k'}$ be a random $k'$-dimensional subspace extending $\qS$. Then with probability at least $1-\epsilon$ over the choice of $\qT$, it holds that $\| (V^{\O^\qS} - V^{\O^\qT}) \ket{\phi}\ket 0 \| \in \O(L N^{-1/3})$, where $\epsilon \in \O(\exp(-N^{1/3}))$.
\end{lemma}

\begin{proof}
    Recalling from \Cref{def:unique_verifier}, we define a \emph{hybridized} verifier $V[\ell]$.
    \[
        V[\ell] \dfn \Pi_\tout U_L \O^\qT U_{L-1} \dots U_{\ell} \O^\qS U_{\ell-1} \dots U_0 \Pi_{\tin}\,,
    \]
    with $\Pi_\tout = \Id \otimes \ketbra 1_{\reg R_\tout}$ and $\Pi_\tin = \Id \otimes \ketbra 0_{\reg R_\taux}$. In other words, the circuit queries the oracle $\O^\qS$ in the first $\ell$ queries and then switches to the oracle $\O^\qT$. Then
    on the unique witness $\ket{\phi_\qS}$ of $\O^\qS$,
    \begin{align*}
        \| V^\qT (\ket{\phi_\qS} \ket 0) - V^\qS (\ket{\phi_\qS} \ket 0)\|_2 &= \|V[0] (\ket{\phi_\qS} \ket 0) - V[L] (\ket{\phi_\qS} \ket 0)\|_2\\
                                                   &\leq \sum_{\ell=0}^{L-1} \|V[\ell] (\ket{\phi_\qS} \ket 0) - V[\ell+1]
                                                   (\ket{\phi_\qS} \ket 0)\|_2\\
                                                   &\leq \sum_{\ell=0}^{L-1} \| \Pi_\tout U_L \O^\qT U_{L-1} \dots
                                                   U_{\ell+1} (\O^\qT - \O^\qS) U_{\ell} \dots U_1 \O^\qS U_0 \Pi_\tin
                                                   (\ket{\phi_\qS} \ket 0)\|_2\\
                                                   &\leq \sum_{\ell=0}^{L-1} \|(\O^\qT - \O^\qS) \ket{\phi(\ell)}\|_2\\
                                                   \label{eq:hybrid_bound}
                                                   &= \sum_{\ell=0}^{L-1} 2\|\Pi_\Delta
                                                   \ket{\phi(\ell)}\|_2. \numberthis
    \end{align*}
    In the final line, we use that $\O^\qT - \O^\qS = 2\Pi_\qT - 2\Pi_\qS =  2\Pi_\Delta$.
    Crucially, $\ket{\phi(\ell)}$ is only a function of the subspace $\qS$ and does not depend on $\qT$. The distribution over $\Delta$ is equivalent to $\haar_{\C^N \setminus \qS, k' - k}$. Thus, by \Cref{lem:overlap_concentration},
    \[
        \Pr_{\Delta \sim \haar_{N - k, k' - k}}[\tr[\Pi_\Delta \phi(\ell)] \geq 2 (N - k')^{-1/3}] \leq \exp(-(N-k)^{1/3}) \leq \exp(-N^{1/3})
    \]
    where the last inequality holds as $k \in o(N)$. Thus \Cref{eq:hybrid_bound} is at most $\O(L \cdot N^{-1/3})$ with probability at least $1 - L \exp(-N^{1/3})$, which proves the claim.
\end{proof}

Replacing the quantum oracle with classical ones, we get nearly the same proof, except we use Markov's inequality rather than \Cref{lem:overlap_concentration}. Additionally, we need to set $k' \leq N^{1/4}$ for technical reasons.\footnote{Later, we'll union bound over all $k'$ elements of $\T$.} Applying Markov's inequality over the random variable $r_\ell(T) = \|V^\T (\ket{\phi(\ell)} \ket 0) - V^\S (\ket{\phi(\ell)} \ket 0)\|$ where $\E_T r_\ell(T) \leq N^{-3/4}$, we obtain the following lemma.
\begin{lemma}[Clasical oracle indistinguishability]\label{lem:classical_ext_indistinguishable} Suppose $k' \leq N^{1/4}$ and $k \in o(k')$. Let $V$ be an $L$-query oracle verifier and $\ket{\phi}$ be a quantum state. Let $\S$ be a $k$-dimensional classical subspace and $\T$ be a random $k'$-dimensional classical subspace extending $\S$. Then with probability at least $1-\epsilon'$ over the choice of $\qT$, it holds that $\| (V^{\O^\S} - V^{\O^\T}) \ket{\phi}\ket 0 \| \in \O(L N^{-3/8})$, where $\epsilon' \in \O(L N^{-3/8})$.
\end{lemma}

%% file: figs/circuit.tex
\tikzset{cross/.style={cross out, draw=black, minimum size=2*(#1-\pgflinewidth), inner sep=0pt, outer sep=0pt},
%default radius will be 1pt. 
cross/.default={1pt}}

\begin{figure}[h]
    \centering
    \newcommand{\yd}{-0.5}
    \begin{tikzpicture}

        % Auxiliary register
        \draw ($ (0,0) + 0*(0,\yd) $) node (auxstart) {} -- node [midway] (aux1mid) {} ($ (4,0) + 0*(0,\yd) $);
        \node[label={left:$\ket 0$}] at (auxstart) {};
        \draw[dashed] ($ (4,0) + 0*(0,\yd) $) -- ($ (5,0) + 0*(0,\yd) $);
        \draw ($ (5,0) + 0*(0,\yd) $) -- ($ (11,0) + 0*(0,\yd) $);
        \draw ($ (0,0) + 2*(0,\yd) $) node (oracleend) {} -- ($ (4,0) + 2*(0,\yd) $);
        \node[label={left:$\ket 0$}] at (oracleend) {};
        \draw[dashed] ($ (4,0) + 2*(0,\yd) $) -- ($ (5,0) + 2*(0,\yd) $);
        \draw ($ (5,0) + 2*(0,\yd) $) -- ($ (11,0) + 2*(0,\yd) $);
        \begin{scope}[transform canvas={xshift=0.5cm}]
            \draw[dotted, line width=0.4mm] ($ (auxstart) + (0,\yd/2) $) -- ($ (oracleend) + (0,-\yd/2) $);
        \end{scope}

        \draw ($ (0,0) + 4*(0,\yd) $) node (auxend) {} -- node [midway] (aux2mid) {} ($ (4,0) + 4*(0,\yd) $);
        \node[label={left:$\ket 0$}] at (auxend) {};
        \draw[dashed] ($ (4,0) + 4*(0,\yd) $) -- ($ (5,0) + 4*(0,\yd) $);
        \draw ($ (5,0) + 4*(0,\yd) $) -- ($ (11,0) + 4*(0,\yd) $);
        \begin{scope}[transform canvas={xshift=0.5cm}]
            \draw[dotted, line width=0.4mm] ($ (oracleend) + (0,\yd/2) $) -- ($ (auxend) + (0,-\yd/2) $);
        \end{scope}
        \begin{scope}[transform canvas={xshift=-0.75cm}]
            \draw[decorate,decoration={brace,mirror,raise=5pt}] (auxstart.center) -- node[left=6pt] {$\reg R_\O$} (oracleend.center);
            \begin{scope}[transform canvas={xshift=-0.75cm}]
                \draw[decorate,decoration={brace,mirror,raise=5pt}] (auxstart.center) -- node[left=6pt] {$\reg R_\text{aux}$} (auxend.center);
            \end{scope}
        \end{scope}
        
        \draw ($ (0,0) + 5*(0,\yd) $) node (instart) {} -- ($ (4,0) + 5*(0,\yd) $);
        \draw[dashed] ($ (4,0) + 5*(0,\yd) $) -- ($ (5,0) + 5*(0,\yd) $);
        \draw ($ (5,0) + 5*(0,\yd) $) -- ($ (11,0) + 5*(0,\yd) $);
        \draw ($ (0,0) + 7*(0,\yd) $) node (inend) {} -- node [midway] (in2mid) {} ($ (4,0) + 7*(0,\yd) $);
        \draw[dashed] ($ (4,0) + 7*(0,\yd) $) -- ($ (5,0) + 7*(0,\yd) $);
        \draw ($ (5,0) + 7*(0,\yd) $) -- ($ (11,0) + 7*(0,\yd) $);
        \begin{scope}[transform canvas={xshift=0.5cm}]
            \draw[dotted, line width=0.4mm] ($ (instart) + (0,\yd/2) $) -- ($ (inend) + (0,-\yd/2) $);
        \end{scope}
        % Register label
        \draw[decorate,decoration={brace,mirror,raise=5pt}] (instart.center) -- node[left=6pt] {$\reg R_\text{in}$} (inend.center);

        % First gate
        \draw[fill=white] (1,0.25) rectangle node[midway] {$U_0$} (2, -3.75);
        
        % First oracle call
        \draw[fill=white] (2.5,0.25) rectangle node[midway] {$\O$} (3.5, -1.25) node[fitting node] (orc1) {};
        \draw (orc1.south) -- (3, 52 |- auxend);
        \draw[fill=black] (3, 52 |- auxend) circle (3pt) node[fitting node] (o3) {};

        % Second to last gate
        \draw[fill=white] (5.5,0.25) rectangle node[midway] {$U_{L-1}$} (6.5, -3.75);
        
        % Last oracle call
        \draw[fill=white] (7,0.25) rectangle node[midway] {$\O$} (8, -1.25) node[fitting node] (orc1) {};
        \draw (orc1.south) -- (7.5, 52 |- auxend);
        \draw[fill=black] (7.5, 52 |- auxend) circle (3pt) node[fitting node] (o3) {};
        \draw (o3) node[cross=2.5pt,rotate=45] {};
        
        % last gate
        \draw[fill=white] (8.5,0.25) rectangle node[midway] {$U_{L}$} (9.5, -3.75);

        % Measurement
        \draw[fill=white] (10,0.25) rectangle node[midway] {$M$} (11, -0.25) node[fitting node] (measure) {};

        \node[label={right:$\reg R_\text{out}$}] at (measure.east){};
        
    \end{tikzpicture}
    \caption{Circuit diagram for the verifier $V$. $\reg R_\O$ is the oracle register, contained in the auxiliary register $\reg R_\text{aux}$, initialized to $\ket 0_{\reg R_\text{aux}}$. $\reg R_\text{in}$ is the input register. The final measurement $M$ is performed on $\reg R_\text{out}$, which is the first qubit of the auxiliary register.}
    \label{fig:unique-verifier}
\end{figure}

%% file: plausibility.tex
\subsection{Background}
\label{sec:plausible}
We first give a more detailed exposition on the significance of $\UQMA$ before stating our results. In a celebrated paper~\cite{Valiant_VV1986_NPEasyDetecting}, Valiant and Vazirani answered the analogous question in the classical setting by giving a randomized reduction from $\NP$ to $\UNP$ ($\NP \subseteq \RP^{\UNP}$). This was a surprising result to researchers in the 1980's~\cite{arora2009computational}.
Naturally, this result has led to the $\UQMA$ vs $\QMA$ question. However, previous
works~\cite{Aharonov_ABBS2022_PursuitUniquenessExtending, Jain_JKK+2011_PowerUniqueQuantum} did not succeed in giving a quantum analogue of the Valiant-Vazirani theorem for interesting reasons that seem to be inherently quantum.

Aharonov et al.~\cite{Aharonov_ABBS2022_PursuitUniquenessExtending} were the first to attempt to quantize Valiant-Vazirani protocol. They pointed out that ``the main difficulty in the $\QMA$ case is that we do not know in which basis to operate.'' In particular, the key step (see~\cite[Theorem 17.18]{arora2009computational}) is introducing an extra test that singles out an accepting witness of the $\NP$ circuit. This extra test is a pairwise independent hash function that, with high probability over the hash function family's randomness, uniquely maps some accepting bit string to the all-zeroes string. The work of~\cite{Aharonov_ABBS2022_PursuitUniquenessExtending} attempted to mimic this random hash function by a random binary POVM measurement (``random'' here means some 2-design, to allow for efficient implementation). However, they pointed out that the method fails even when the goal is to single out a single accepting witness out of just two valid witnesses! The reason is due to the concentration of measure in a high-dimensional Hilbert space: the overlaps of two orthogonal vectors with a random projector are both exponentially small in the number of qubits. 
Fundamentally, this issue arises because unlike in the classical setting, there is a freedom of basis for any high-dimensional subspace, and we do not know in which basis the accepting witnesses is specified.

Jain et al.~\cite{Jain_JKK+2011_PowerUniqueQuantum} took another approach that shows how to reduce $\cpoly(n)$ witnesses to one, yielding $\mathsf{FewQMA} \subseteq \Ptime^{\UQMA}$. The idea is that, at least classically, we can achieve this by simply asking the prover to provide a list of all $\poly(n)$ accepting witnesses, sorted in some canonical (e.g. lexicographical) order. In the quantum case,~\cite{Jain_JKK+2011_PowerUniqueQuantum} notice that a ``canonical order'' can be imitated by asking the prover to provide a state in the $t$-fold alternating subspace\footnote{This is the counterpart of the symmetric subspace, which has countless applications in quantum computing \cite{Harrow_Har2013_ChurchSymmetricSubspace}.} for the accepting witnesses. When $t$ is taken to be the dimension of accepting witnesses, this subspace has dimension one. Since the alternating subspace can be efficiently tested, this approach evades the need for knowing a basis for the accepting witnesses. However, their strategy requires a witnesses which scales with the number of accepting witnesses and thus restricts their algorithm to the case when there are at most $\cpoly(n)$ orthogonal accepting witnesses. In fact, their strategy is a black-box transformation and thus cannot resolve the relationship between $\UQMA$ and $\QMA$ due to our oracle separation result in this section.

One can also think about $\UQMA$ from a more physical point of view by the relationship between the existence of a unique witness and inverse polynomial spectral gap. More precisely, as shown in~\cite[Lemma 49]{Aharonov_ABBS2022_PursuitUniquenessExtending}, the non-equivalence of $\UQMA$ and $\QMA$ (under quantum reductions) implies that the local Hamiltonian problem with an inverse polynomial spectral gap promise is \emph{strictly easier} than general local Hamiltonians (this holds even when we restrict the Hamiltonians to be 1D, due to~\cite{AGIK09}).

\begin{figure}[H]
    \centering
\begin{tikzpicture}[scale=0.5]
    \draw[thick, |->] (-4,0) -- (22,0) node[below=0.1] {\stackanchor{spectral}{gap}};

    \draw[thick] (0,-0.3) -- (0,0.3) node[below=0.3] {$2^{-\poly(n)}$};

    \fill [gray, fill=gray, opacity=0.8] (-4,0) rectangle (0, 2.5); % gray box

    \node[color=black] at (-2,1.25){$\mathsf{QMA}$};

    \node[color=black] at (2.25, 1){ \stackanchor{?}{$=$}};

    \node[color=black] at (4.5,-0.8){$\frac{1}{\poly(n)}$};
    
    \draw[thick] (10,-0.3) -- (10,0.3) node[below=0.3] {$\frac{1}{n^4}$};

    \fill [gray, fill=gray, right color=white, opacity=0.3] (4.5,0) rectangle (10, 2.5); % gray box

    \node[color=black] at (7.25,1.25){$\mathsf{UniqueQMA}$};

    % \node[color=black] at (12,1.25){$\neq$};

    \draw[thick] (14,-0.3) -- (14,0.3) node[below=0.3] {$\frac{1}{\sqrt{n}}$};

    \fill [gray, fill=gray, right color=white, opacity=0.2] (14,0) rectangle (18, 2.5); % gray box

    \node[color=black] at (16,1.25){\stackanchor{Subexp}{time}};

    \draw[thick] (18,-0.3) -- (18,0.3) node[below=0.3] {$\frac{1}{\sqrt{\log n}}$};

    \fill [gray, fill=gray, right color=white, opacity=0.05] (18,0) rectangle (21.5, 2.5); % gray box

    \node[color=black] at (19.75,1.25){$\BPP$};

\end{tikzpicture}
    \caption{The computational complexity of estimating the ground energy of 1D local Hamiltonians to $1/\poly(n)$ accuracy depends on the spectral gap $\gamma$. The $\QMA$-completeness when there is no spectral gap promise is shown in~\cite{AGIK09}. Ref.~\cite{Aharonov_ABBS2022_PursuitUniquenessExtending} showed that the problem is in $\UQMA$ when $\gamma = 1/\poly(n)$ and is $\UQMA$-hard for $\gamma = \Omega(1/n^4)$. The subexponential-time and $\BPP$ classification for sufficiently large spectral gap is due to~\cite{arad2017rigorous}, who give a classical randomized algorithm with runtime $n^{\Tilde{O}(\gamma^{-2})}$. }
    \label{fig:1Dspectralgap}
\end{figure}

\begin{comment}
First, classical Hams are always gapped, Cook-Levin transformation always gapped. Valiant-Vazirani is operating the gapped regime... On the other hand, this is not the case in the Feynman-Kitaev transformation, where we face several no-go results (cite Cubitt et al.) in hope for proving QMA hardness of inverse-poy gapped Hams.
\end{comment}

A recent paper~\cite{deshpande2022importance} gives weak evidence for this statement by studying the `precise' version of the local Hamiltonian problem, where the goal is to estimate the ground energy to exponential accuracy. Without any promise on the spectral gap, this problem is known to be $\PSPACE$-complete~\cite{fefferman2016complete} (which is a class believed to be vastly more difficult than $\QMA$). However,~\cite{deshpande2022importance} show that the complexity of the same problem drops to $\PP$ under the assumption of an inverse polynomial spectral gap. Since it is believed that $\PP \subsetneq \PSPACE$, this means that the inverse polynomial spectral gap assumption makes the problem strictly easier in this high-accuracy regime.
Although this doesn't necessarily indicate the same conclusion in the standard polynomial accuracy regime, it is argued in~\cite{deshpande2022importance} that a spectral gap does seem to indicate easiness as observed in practice. They conjecture that inverse polynomial gapped Hamiltonians admit low-energy states with polynomial circuit complexity, which would in turn imply containment in $\QCMA$.

An inverse polynomial spectral gap of the parent Hamiltonian often implies provably efficient algorithms~\cite{schwarz2012preparing, schwarz2017approximating} in the context of injective tensor networks. 
For general 1D Hamiltonians, a spectral gap of roughly at least $1/\sqrt{n}$ implies a sub-exponential time algorithm~\cite{arad2017rigorous}. See~\Cref{fig:1Dspectralgap} for an incomplete computational complexity phase diagram of 1D Hamiltonians.

%% file: separation.tex
\subsection{Results overview}\label{subsec:oracle-separation-result}

In this section, we state the oracle separations involving $\UQMA$ and give the proofs for classical oracles in \Cref{sec:all_sep_proofs} and for quantum oracles in \Cref{sec:relativized_separation}. As mentioned before, the oracle problem we consider is related to the (Quantum) Approximate Dimension problem of
\cite{she2023unitary}. We consider a slightly different oracle promise problem, $\QApxDim\textsf{Exact}$ with parameters $k_1$ and $k_2$ so that the YES case corresponds to subspace dimension exactly $k_2$, and the NO
case corresponds to dimension exactly $k_1$. Then, the specific problems we consider are,
\begin{align*}
    \textsf{EmptyNonEmpty} &\dfn \ApxDim\textsf{Exact}(0, 1, N) \cup \ApxDim\textsf{Exact}(0, k, N)\text{ and }\\
    \textsf{QEmptyNonEmpty} &\dfn \QApxDim\textsf{Exact}(0, 1, N) \cup \QApxDim\textsf{Exact}(0, k, N)\,.
\end{align*}
where the union is taking within the NO case and YES case oracles individually (so that the NO case is still the identity oracle, and YES case corresponds to subspaces of dimension $1$ and $k$). Since $\textsf{EmptyNonEmpty} \subseteq \ApxDim(0, 1, N)$ (and $\textsf{QEmptyNonEmpty} \subseteq \QApxDim(0, 1, N)$), we could've considered the harder problem instead (which would essentially recover \Cref{task:uqma_classical}), but we use $\textsf{EmptyNonEmpty}$ since this includes all subspace dimensions needed to make the proof go through.

Intuitively, this definition is motivated by the desire to take advantage of $\UQMA$'s structural condition on having a unique witness. Restricting to subspaces of a single dimension, say $k = 1$, there is a natural unique verification protocol. Instead, our proof will heavily exploit the interaction between subspaces of different sizes.  

\subsubsection{Classical oracles}
We first restrict ourselves to classical (computational basis) oracles (see, \Cref{rem:classical_oracles}). We show the following query lower bound.
\begin{restatable}[$\UQMA$ query lower bound, $\textsf{EmptyNonEmpty}$]{theorem}{querylowerbound}
\label{thm:query_lower_bound}
Let $V$ be a $L$-query, $m$-qubit verifier which uniquely decides $\textsf{EmptyNonEmpty}$ with error $\eps \leq 2^{-n}$ both on $\O_\text{No}$ and on at least $(1 - 2^{-n})$-fraction of YES instances. Then, $L \in \Omega(\sqrt k)$. In particular,
if we take $k \in 2^{\Omega(n)}$, then $V$ requires \emph{exponentially}-many queries in $n$.
\end{restatable}

Combining with standard diagonalization techniques (see \Cref{sec:oracle_separation_proof}), these query lower bounds yield the separation $\UQMA^\O \subsetneq \QMA^\O$.

\begin{remark}[Size of the Witness Register] The query lower bound holds independently of the witness size (in fact it
    holds even if the witness size is infinite). But with an arbitrary-size witness, couldn't the prover simply send an
    exponentially long classical description of the oracle? This doesn't work because it is not uniquely verifiable --
    the prover could send a different string describing a slightly different oracle. With only a polynomial number of
    queries, these two are indistinguishable. Applying the Valiant-Vazirani theorem to uniquely single out a bit-string will also not suffice because Valiant-Vazirani is a randomized reduction that succeeds only with probability $1/|\mathrm{witness}|$.
\end{remark}

We can immediately strengthen the separation in two directions. First, via a simpler version of the techniques from
\Cref{thm:relativized_separation} (essentially, by replacing the application of Levy's Lemma with Markov's inequality,
then union bounding over all polynomially-many circuits queried by a $\Ptime$ verifier), the lower bound from
\Cref{thm:query_lower_bound} can be shown to hold for $\P^\UQMA$ algorithms as well.\footnote{This query lower bound is
tight up to quadratic factors; the algorithm of \cite{Jain_JKK+2011_PowerUniqueQuantum} gives a $\P^{\UQMA^\O}$ verifier
solving the above task in $\O(k)$ queries.} Second, since the oracle is classical, $\textsf{EmptyNonEmpty}$ is actually
solveable in $\NP$. This yields the following improved result,

\begin{corollary}
\label{cor:oracle_separation}
There exists a classical oracle $\O$ such that $\UQMA^\O \subsetneq \QMA^{\O}$, and in fact, $\P^{\UQMA^\O} \not \subseteq \NP^\O$.
\end{corollary}

\begin{remark}\label{remark:derandomVV}
    This implies that $\NP \not\subseteq \P^{\UNP}$ relative to the same oracle, capturing the open question of \emph{derandomizing} Valiant-Vazirani~\cite{beigel1998np, klivans1999graph}, where, given a satisfiable SAT formula, the goal is to deterministically output a list of $\cpoly(n)$ SAT formulae, at least one of which is uniquely satisfiable \footnote{This is not to be confused with a much more stringent task called \emph{deterministic isolation}, which requires one to deterministically output a single circuit that is uniquely satisfiable. There is very strong evidence against its possibility~\cite{hemaspaandra1996computing, dell2013valiant}.}. 
The work of~\cite{klivans1999graph} showed that this is possible under a certain circuit complexity assumption. An influential early work~\cite{beigel1998np} also produced an oracle under which $\NP \not\subseteq \Ptime^{\UNP}$. However, our oracle separation proof is arguably more natural and direct than~\cite{beigel1998np}. There, the authors gave an intricate oracle relative to which $\P=\UNP$ and $\NP=\EXP$, and hence $\NP \not \subseteq \Ptime =\Ptime^{\UNP}$ due to the time hierarchy theorem and the self-lowness of $\Ptime$. In contrast, our classical oracle is simpler as it is essentially a random function and allows us to directly show the non-containment $\NP \not\subseteq \Ptime^{\UNP}$ without blowing up the power of $\NP$ or collapsing $\UNP$ to $\Ptime$.
\end{remark}

\subsubsection{Quantum oracles}
For our proof of \Cref{thm:relativized_separation}, we'll need the same lower bound against the quantum oracles of $\textsf{QEmptyNonEmpty}$, where we lift the restriction that the subspace is in the computational basis. The simple $\QMA$ lower bound for $\textsf{QEmptyNonEmpty}$ follows directly from \Cref{thm:query_lower_bound}. Suppose a verifier succeeded with probability $\geq 1 - 2^{-n}$ over quantum oracles. Treat the quantum oracle distribution as drawing a Haar-random unitary $U$, then conjugating a randomly chosen classical subspace. Thus, this implies that there exists some unitary $U^\star$ for which the verifier succeeds with probability $\geq 1 - 2^{-n}$, and pulling $U^\star$ into the verifier's unitaries yields an algorithm for $\textsf{EmptyNonEmpty}$.

\begin{corollary}[$\UQMA$ Query lower bound, $\textsf{QEmptyNonEmpty}$]
\label{cor:quantum_query_lower_bound}
Let $V$ be a $L$-query, $m$-qubit verifier which uniquely decides $\textsf{QEmptyNonEmpty}$ with error $\eps \leq 2^{-n}$ both on $\O_\text{No}$ and on at least $(1 - 2^{-n})$-fraction of YES instances. Then, $L \in \Omega(\sqrt k)$. 
\end{corollary}

Classically, \cite{Valiant_VV1986_NPEasyDetecting} shows that approximate counting is contained in $\BPP^\UNP$, and without a de-randomization assumption it is not known whether approximate counting can be placed in $\UNP$. Therefore, in analogy to the classical case, a more surprising lower bound would be against $\BQP^\UQMA$. Indeed, by moving to quantum oracles (and therefore considering the quantum approximate counting problem), we are able to obtain such a lower bound.

\begin{restatable}[Relativized Oracle Separation]{theorem}{relativizedsep}
    \label{thm:relativized_separation}
    Let $V^\qT$ be any $L$-query $\BQP$ verifier with gates implementing controlled calls to
    the subspace reflection oracle $\O^S$, as well as an oracle for $\textsc{\textup{UniqueQCircuitSat}}^{\O^S}$. Suppose $V^\qT$ solves $\textsf{\textup{QEmptyNonEmpty}}$, with error at most $\eps \in \exp(-n)$. Then, if $k \geq N^\alpha$ for some $\alpha \in \Omega(1)$, then $L \in \Omega(N/k)$.
\end{restatable}

Standard techniques yield the desired oracle separation.

\begin{corollary}\label{cor:UQMA-quantum-oracle-sep}
    There exists a quantum oracle $\O$ such that $\BQP^{\UQMA^\O} \not\subseteq  \QMA^\O$.
\end{corollary}

\subsection{\texorpdfstring{Classical oracle separations and $\UQMA$}{Classical oracle separations and UniqueQMA}}
\label{sec:all_sep_proofs}

\subsubsection{Query lower bound}
In this section, we prove a query lower bound for $\UQMA$ algorithms against the problem
\[
\textsf{EmptyNonEmpty} = \ApxDim\textsf{Exact}(0, 1, N) \cup \ApxDim\textsf{Exact}(0, k, N)\,,
\]
establishing~\Cref{thm:query_lower_bound}. All oracles in this section are taken to be \emph{classical}. Technically, we prove a lower bound against verifiers with slightly weaker requirements, as this will lead to a simpler statement for the polarization lemma (\ref{lem:polarization}). However, the lemma directly implies the desired result.

\begin{lemma}
    \label{lem:query_lower_bound_weaker}
    Let $V$ be a $L$-query, $m$-qubit verifier which uniquely decides $\ApxDim(0,k,N)$ and accepts \emph{possibly non-uniquely} for instances corresponding to the YES case of $\ApxDim(0,1,N)$, both with error at most $\eps \in \exp(-n)$ on $(1-2^{-n})$-fraction of instances. Then, $L \in \Omega(\sqrt k)$.
\end{lemma}

For comparison, recall,

\querylowerbound*

The only difference between the statement of \Cref{lem:query_lower_bound_weaker} and \Cref{thm:query_lower_bound} is that in the former the verifier is allowed to accept non-uniquely on subspaces of dimension $1$. Any verifier satisfying the conditions of \Cref{thm:query_lower_bound} is subject to the same query lower bound. It remains to prove \Cref{lem:query_lower_bound_weaker}.

\begin{proof}
We define a hybridized
verifier between the NO case oracle $\O^\emptyset$ and some YES case oracle $\O^\T$.
Eventually,
we will choose $\T$ to be a random $k$-dimensional extension of some 1-dimensional subspace $\S$. For technical reasons, we take $k \in o(\sqrt N)$.

% to come from extensions of a particular $\S \in \tilde{\mc T}$. 
Let $\ket{\phi_\T}$ be the witness for
$\T$. By hybridizing as in the proof of \Cref{lem:ext_indistinguishable}, we find that
\[
    \|V^\T (\ket{\phi_\T} \ket 0) - V^\emptyset (\ket{\phi_\T} \ket 0)\| \leq \sum_{\ell = 0}^{L-1} 2 \| \Pi_\T \ket{\phi(\ell)}\|_2\,,
\]
where $\ket{\phi(\ell)}= U_\ell ...U_1 \Pi_\mathrm{in} (\ket{\phi_\T} \ket 0)$.

By assumption that $V$ is a valid $\UQMA$ verifier on instances of dimension $k$ with completeness $1 - \eps$ and soundness $\eps$, the left-hand side above is at least $1 - 2\eps$ and thus there is a ``critical
index'' $\ell^\star(\T)$ for which 
\begin{equation}
    \label{eq:lower}
    \|\Pi_\T \ket{\phi(\ell^\star(\T))} \|_2 \geq \frac{1 - 2\eps}{2L} \,.
\end{equation}
We write the critical index as a function of $\T$ to make the dependence on $\T$ explicit. Let $\ket{u}$ be a vector in $\T$. We split this expression into
contributions from $\S = \text{span}(\ket u)$ and the remaining subspace $\Delta$ (so that $\T = \Delta \oplus S$). Recalling that $\S$ is a 1-dimensional subspace spanned by some $\ket u$,
\begin{equation}
    \label{eq:separated}
    \|\Pi_\T \ket{\phi(\ell^\star(\T))}\|_2^2 = \Tr[\Pi_\T {\phi(\ell^\star(\T))}] = \underbrace{\Tr[\ketbra u \cdot {\phi(\ell^\star(\T))}]}_{\textup{(A)}} +
    \underbrace{\Tr[\Pi_\Delta \cdot {\phi(\ell^\star(\T))}]}_{\textup{(B)}}\,.
\end{equation}
We will argue that both terms above are small for some choice of $\T$. This will imply $L$ has to be for~\Cref{eq:lower} to be consistent.

\paragraph{Bounding the First Term (A)} We establish the following claim.
\begin{claim}
    Suppose $L \in o(\sqrt k)$. There exists a 1-dimensional subspace $\S=\{\ket{u}\}$ such that
\begin{align}
\label{eq:first-term}
    \sum_{\ell=0}^{L} \Tr(\Pi_\S U_\ell...U_1 (\ketbra{\phi_\S}{\phi_\S} \otimes \ketbra 0)U_1^\dagger ...U_\ell^\dagger) < \frac 1 {\sqrt k}.
\end{align}
\end{claim}
\begin{proof}
We prove this claim by contradiction. Consider any 1-dimensional space $\S$. By assumption, the oracular verifier $V^\S$ accepts some state $\ket{\phi_\S}$ (possibly non-uniquely!) with probability $\geq 1 - \eps$. Then~\Cref{lem:classical_ext_indistinguishable} says that for $\gamma, \delta\in \O(L N^{-3/8})$,
\begin{equation}
\label{eq:concentration_first}
\Pr_{\T \sim \mc C_{\uparrow S, k}}[\text{acc}(V^\T) \geq 1 - \eps - \delta] \geq 1 - \gamma \,.
\end{equation}
By uniqueness of $\T$'s witness $\ket{\phi_\T}$, we can bound the overlap
\begin{align}
    \label{eq:overlap_bound}
    |\braket{\phi_\S}{\phi_\T}| \geq 1 - 2\eps - \delta\,.
\end{align}
Thus for a random classical subspace $\T$ and a basis, with probability $ \geq 1-k\gamma$, it holds for all $i \in [k]$ that $|\braket{\phi_{\S_i}}{\phi_\T}| \geq 1-2 \eps - \delta$.

Summing the contrapositive of~\Cref{eq:first-term} over $\S_i$ we obtain
\begin{align}
    k\cdot \frac 1 {\sqrt k} &\leq \sum_{i=1}^k \sum_{\ell=0}^{L} \Tr(\Pi_{\S_i} U_\ell...U_1 (\ketbra{\phi_{\S_i}}\otimes \ketbra 0) U_1^\dagger ...U_\ell^\dagger) \\
    &\leq \sum_{\ell=0}^{L} \Tr( \left(\sum_{i=1}^k \Pi_{\S_i}\right) U_\ell...U_1 (\ketbra{\phi_\T} \otimes \ketbra 0) U_1^\dagger ...U_\ell^\dagger) + k \cdot (4 \eps + 2\delta)\\
    &\leq L + kL \cdot (4 \eps + 2\delta),
\end{align}
where the last inequality uses $\Pi_{\S_i}$ are mutually orthogonal. By choice of parameters, $L \ll \sqrt k$ and $\eps, \delta \ll \frac 1 {\sqrt k L}$. This yields a contradiction and hence \Cref{eq:first-term} holds for some $\S_i$.
\end{proof}

The reason for taking the sum over $\ell$ is to remove dependence on $\T$ in the claim. 
Picking $\S$ according to this claim and drawing $\T \sim \mc Q_{\uparrow S, k}$ yields a bound on the first term of \Cref{eq:separated}. Again using \Cref{eq:overlap_bound}, $|\braket{\phi_{\S_i}}{\phi_\T}| \geq 1 - 2 \eps - \delta$ so,
\begin{align}
    \Tr[\Pi_\S \cdot {\phi(\ell^\star(\T))}] &\leq  \Tr[\Pi_\S U_{\ell^\star(\T)}...U_1 (\ketbra {\phi_\S} \otimes \ketbra 0) U_1^\dagger ...U_{\ell^\star(\T)}^\dagger]   + (4 \eps - 2 \delta)\\
    \label{eq:bound1}
    &\leq \frac 1 {\sqrt k} + 4 \eps - 2 \delta.\numberthis  
\end{align}

\paragraph{Bounding the Second Term (B)}
Take $\S \in \mathcal S$ as above, and let $\T$ be a random extension of $\S$. Let $\mc A(\T)$ be the event that $\T$ has the same witness
as some witness of $\S$. Taking $\gamma, \delta$ as above, by \Cref{lem:classical_ext_indistinguishable}, $\Pr_{\T}[\mc A(\T)] \geq 1 - \gamma$. Thus, we may condition on this event
with at most an additive $\gamma$ loss in our final bound and treat $\ket{\phi} = \ket{\phi(\ell)}$ as independent of
$\T$. Next, $\Delta$ is a random
dimension subspace of dimension at most $N^{2/3}$ from a $N - 1$-dimensional ambient space (the subspace
orthogonal to $\S$). Applying Markov's inequality on $r_\Delta = \tr[\Pi_\Delta \phi]$, we have that $\Pr_\Delta [r_\Delta \geq N^{-{3/8}}] \leq N^{-3/8}$. Thus,
\begin{equation}
    \label{eq:concentration_second}
    \Pr_\T[\Tr[\Pi_\Delta \phi(\ell(\T))] \geq 2^{-n/3}] \leq 1 - N^{-3/8} + \gamma \leq 1 - 2\gamma\,.
\end{equation}

\paragraph{Concluding} Combining \Cref{eq:bound1} and \Cref{eq:concentration_second}, we conclude that for some $\S \in \mathcal S$, with probability $1 - 2^{-\Omega(n)}$ over extensions $\T$ of $\S$,
\begin{equation}
    (\ref{eq:separated}) \leq \frac 1 {\sqrt k} + 4 \eps - 2 \delta + \O(LN^{-3/8}) \in 2^{-\Omega(n)}\,.
\end{equation}
But recalling \Cref{eq:lower}, (\ref{eq:separated}) is lower bounded by $\frac 1 {2L}(1- 2 \eps)$. Thus, we must have that $L \in \Omega(\sqrt k)$.
\end{proof}

\subsubsection{Diagonalization argument}
\label{sec:oracle_separation_proof}

We now prove \Cref{cor:oracle_separation} using the standard diagonalization argument.

\begin{proof}[Proof of \Cref{cor:oracle_separation}]
    Let $\mc O = \{\O_n\}_{n \geq 1}$ be a set of oracles, with $\O_n = \O^{\S_n}$, where $\S_n \subseteq \C^{2^n}$ and is
    possibly empty. Define the unary language $L^{\O} \subseteq \{0\}^n$ so that $n \in L^\O$ if and only if $\S_n \neq
    \emptyset$ (and thus $\O^{\S_n}$ is a positive instance of $\textsf{EmptyNonEmpty}$). As in \cite{Aaronson_AK2007_QuantumClassicalProofs}, we assume for simplicity that for a length-$n$ input, any verifier for $L^\O$ only queries $\O^n$, and not $\O^m$ for $m \neq n$.\footnote{However, the same proof would work in the more general setting.} For any $\mc O$, $L^{\mc O}$ is clearly contained in $\QMA^{\O}$ (and $\NP^{\O}$!). Given input $n$, and a witness state $\ket \psi$, the verifier constructs the state $\tfrac 1 {\sqrt 2} (\ket 0 \ket \psi + \ket 1 \ket \psi)$. Then, apply the oracle $\O_n$ on the second register, conditioned on the first. Finally, the verifier measures the first qubit in the Hadamard basis. If $\O^n = \O_{\S_n}$ for some non-empty subspace $\S_n$, the prover can provide a state $\ket \psi \in S_n$ such that the prover accepts (with probability $1$). Otherwise, no choice of $\ket \psi$ will cause the verifier to accept with probability more than $\tfrac 1 2$.

    Now we show that there exists an oracle set $\mc O$ such that $L^{\mc O} \not\in \UQMA^\O$. We do this by enumerating all the (countably infinite) $\UQMA$ machines $\mc M_1, \mc M_2, \dots$ and for each $\mc M_i$ picking some oracle $\O^n$ so that $\mc M_i$ fails to correctly decide the string $0^n$. Here, we say that $M_i$ ``correctly decides'' some $O^n$ if when $0^n \in L^\O$, there exists a unique witness such that $\mc M_i$ accepts with probability at least $\tfrac 2 3$, and if $O^n \not\in L^\O$, then $\mc M_i$ rejects with probability at least $\tfrac 2 3$ for any witness.
    
    Start with $\mc M_1$. By \Cref{thm:query_lower_bound} for any sufficiently large $n$ there is a choice of oracle
    $\O^n$ such that $\mc M_1$ fails to correctly decide the string $0^n$. Otherwise $\mc M_1$ would yield a
    $\cpoly(n)$-query algorithm for $\oprob$. Therefore, we pick such an $n$ and fix this choice of $\O^n$. Now for $\mc M_i$, assume we've fixed $\O^{i}$, for all $i \leq N$, for some value of $N$. We can pick some $n > N$ and corresponding oracle $\O^n$ such that $\mc M_i$ incorrectly decides $O^n$. Constructing $L^\O$ in this way ensures that $L^\O$ is not decidable by any $\UQMA$ machine, and thus $L^\O \not\in \UQMA^{\O}$.
\end{proof}

%% file: vv_separation.tex
\subsection{\texorpdfstring{Quantum oracle separations and $\BQP^{\UQMA}$}{Quantum oracle separations and BPP\^UQMA}}
\label{sec:relativized_separation}
We now extend the previous oracle separation to show that relative to a \emph{random} subspace reflection oracle $\O^\qS$,
$\QMA^{\O^\qS} \neq \BQP^{\UQMA^{\O^\qS}}$.\footnote{Note that as in, e.g., \cite{irani2021quantum}, when we consider the
    class $\textsf{C}_1^{\textsf{C}_2^{\O}}$, we assume that the $\textsf{C}_1$ machine has access to not only the
oracle for $\textsf{C}_2^{\O}$, but for $\O$ itself as well. This only makes our lower bound stronger. Additionally, for simplicity we assume the oracle calls are to $\O$, and not its controlled version, but dealing with this case is easily done.} In this section, we consider quantum oracles, indicated by the calligraphic $\qS$, and study the task
\[
    \textsf{QEmptyNonEmpty} \dfn \QApxDim\textsf{Exact}(0, 1, N) \cup \QApxDim\textsf{Exact}(0, k, N)\,.
\]

Let $V^{\qS}$ be a $\BQP^{\UQMA^{\O^\qS}}$ verifier. In other words, $V^\qS$ is a $\BQP$ machine with the extra abilities of querying the oracle $\O^\qS$, and an oracle for solving $\UQMA^{\O^\qS}$ problems, both at unit cost. We assume the $\UQMA^{\O^\qS}$ solver oracle is realized by an oracle $\A$ for the $\UQMA^{\O^\qS}$-complete problem $\textsc{UniqueQCircuitSat}^{\O^\qS}$. As in \Cref{def:oracle_extension}, we take $\mc A$ to be a oracle extension of $\textsc{UniqueQCircuitSat}$, as this is a promise problem.

In the context of our work, this means that if our goal is to prove that $V^\UQMA$ cannot decide some language, it suffices to show that for an adversarially chosen extension $\A$, $V^{\mc A}$ fails to decide the language. We take a fairly ``reasonable'' (from our point of view) choice of extension that is similar to the choice in~\cite[Section 3]{irani2021quantum}.

\begin{definition}[Decision Oracle for $\textsc{UniqueQCircuitSat}^{\O^\qS}$]\label{def:uqma-solver}
    Let $\A^{\O^\qS}$ be the oracle extension for the promise problem $\textsc{UniqueQCircuitSat}^{\O^\qS}$ with completeness $1-\eps$ and soundness $\eps$. Then, $\mc A^{\O^\qS}$ is defined as follows:
    \begin{itemize}
        \item For oracular $\UQMA^{\O^\qS}$ circuits $\texttt C$ outside of the promise, outputs $\mc A^{\O^\qS}(\texttt C) = \mc A^{\emptyset} (\texttt C)$. In essence, we replace the oracle $\O^\qS$ by the identity (the empty subspace reflection) on invalid instances, where the behavior of $\mc A^{\emptyset} (\texttt C)$ is defined below,
        \item  $\mathcal{A}^{\emptyset}(\texttt{C})=1$ if the maximum acceptance probability of $\texttt{C}^{\emptyset}$ is $\geq1/2$; otherwise $\mathcal{A}^{\emptyset}(\texttt{C})=0$.
    \end{itemize}
    
\end{definition}

Thus, we realize calls to a $\UQMA^{\O^\qS}$ oracle via $\mc A^{\O^\qS}$. We will oftentimes simply write $\mc A^{\qS}$ or $\mc A^{\emptyset}$, or omit the superscripts $\qS, \O^\qS$ when it is clear from context what subspace reflection we are working with.

To generalize the query lower bound from the previous section, we'll argue the following:

\begin{enumerate}
    \item $V^\qS$ only calls $\A^\qS$ on at most $C(n) = 2^{n^c}$ distinct oracular quantum circuits $\texttt C_1, \dots, \texttt
        C_{C(n)}$ for some constant $c$.
    \item For each YES-case subspace $\qS \in \Omega_\text{Yes}$, the behavior of our $V^\qS$ must be distinct from $V^\emptyset$ and thus we must have that $\A^\qS(\texttt C_i) \neq \A^{\emptyset}(\texttt C_i)$ for some $i \in
        [C(n)]$.
    \item \label{list:proof_item3} On the other hand, if we can show \emph{every} $\texttt C_i$ fails to correctly decide
    $\oprob$ on all but $\frac 1 {F(n)}$ fraction of subspaces, with $F(n) > C(n)$, this implies
    (via a union bound) that there must be a subspace $\qS$ such that $\A^\qS(\texttt{C}_i) =
    \A^\emptyset(\texttt C_i)$ for all $i \in [C(n)]$.
    \item Therefore on this $\qS$, $V^\qS$ learns nothing by querying the oracle $\A^\qS$. Finally, to show that direct calls
        to $\O^\qS$ do not help, we can hybridize between $\emptyset$ and a Haar-random $\qS$.
\end{enumerate}

The primary difficulty in executing the above plan is showing Item 3. Indeed, the oracle separation from
\Cref{cor:oracle_separation} only rules out circuits which solve $\QApxDim(0,k,N)$ on nearly all instances;
on the otherhand, a priori $V^\qS$ could query $\A^\qS$ on circuits which fulfill the $\UQMA$ promise with probability only
half over random $\qS$. However, we will circumvent this issue by showing that, for any fixed oracular circuit $\texttt C$, its behavior on a random $\qS$
\emph{polarizes}. Recall that $\text{acc}(\texttt C^\qS)$ is the maximum acceptance probability of $\texttt C^\qS$ over
witnesses. A straightforward application of Levy's Lemma shows that if $\texttt C$'s
\emph{expected} acceptance probability on a random $\qS$ is $p$, then
\[
    \Pr_\qS [| \text{acc}(\texttt C^\qS) - p | \geq 2^{-\sqrt n}] \leq 2^{-\exp(n)}\,.
\]
This will allow us to argue if on larger than $2^{-\sqrt n}$-fraction of instances a circuit $\texttt C^\qS$ accepts a witness with high probability, concentration extends this behavior to nearly all subspaces $\qS$.

There is one subtlety, which is that for $\UQMA$, $\A^\qS(\texttt C)$ not only depends on the \emph{maximum} acceptance
probability (i.e. the operator norm of the associated POVM $\Pi_{\texttt C}$), but the uniqueness as well. We show that the same
concentration above also holds for the Ky Fan 2-norm, defined as the sum of the top two singular values
$\Pi_\texttt{C}$.

\subsubsection{Circuit polarization}
First, we bound the Lipschitz constant of the Ky Fan 2-norm. Treating $V$ as the accept POVM for a quantum verifier.
\begin{restatable}{lemma}{kyfanlipschitz}
    \label{lem:kyfanlipschitz}
    For a random $k$-dimensional subspace $\qS$ and verifier $V^\qS$ with oracle access to $\O^\qS$, define $F(V^\qS) = \|V^\qS\|_{\uparrow
    2}$. Then, $F$ is $\O(L)$-Lipschitz.
\end{restatable}

The proof is contained in \Cref{sec:lipschitz}. One can also show that the trace norm also has Lipschitz constant
$\O(L)$.

\begin{lemma}[Polarization of $\UQMA$ oracle circuits]
    \label{lem:polarization}
    Let $V$ be an polynomial-time oracular verifier and $k = N^\alpha$, with $\alpha \leq {2/3}$. Then for any function $f(n) \in \exp(-N^{1/4})$:
    $$\Pr_{\qT \sim \haar_{N,k}}[\mc A^\qT(V) \neq \mc A^{\emptyset}(V)] \leq f(n)\,.$$
\end{lemma}
\begin{proof}
    We prove by contradiction. Consider an oracle circuit $V$ such that,
    \begin{equation}
    \label{eq:accept_prob}\Pr_{\qT \sim \haar_{N,k}}[\mc A^\qT(V) \neq \mc A^{\emptyset}(V)] \geq {f(n)} \qquad \E_{S \sim \haar_{N,1}} \|V^S\|_{\text{op}} \geq 1 - 3 \eps
    \end{equation}
    where $\mc A$ is the $\textsc{UniqueQCircutSat}$ oracle from~\Cref{def:uqma-solver}.

    \paragraph{Case 1: $\|V^{\emptyset}\| < 1/2$} In this case $\mathcal{A}^{\emptyset}(V) = 0$. Consider the expected operator norm of $V^\qT $ over random choice of $\qT \sim \haar_{N,k}$. First, suppose $\E_\qT \| V^\qT \|_\text{op} < 1 - 2 \eps$, then by
    Levy's Lemma via \Cref{lem:levyslemma}, $\|V^\qT\|_\text{op} < 1 - \eps$ with probability at least $1 - \gamma$ where $\gamma = \exp(-N^{1/3}) \ll 1/f(n)$. For any such $\qT$, $\ttC^\qT$ is either a NO case or falls outside of the $\UQMA$ promise. Thus, by \Cref{def:uqma-solver}, $\mc A^\qT(V) = \mc A^\emptyset(V)$ with probability at least $1- \gamma$,
    contradicting \Cref{eq:accept_prob}. Thus, it holds that $\E_\qT \| V^\qT \|_\text{op} \geq 1 - 2 \eps$.
    
    Next, we consider the expected Ky-Fan 2 norm $\E_\qT\|V^\qT\|_{\uparrow 2}$. We have $\E_\qT\|V^\qT\|_{\uparrow 2} \geq \E_\qT\|V^\qT\|_\mathrm{op} \geq 1 - 2 \eps$. Suppose $\E_\qT\|V^\qT\|_{\uparrow 2}
    > 1 + 2 \eps$. Then, with probability $1- \gamma$, $\lambda_1(V^\qT) + \lambda_2(V^\qT) > 1 + \eps$. Using the trivial upper bound
    $\lambda_1(V^\qT) \leq 1$, 
    \[ \lambda_2(V^\qT) > 1 + \eps - \lambda_1(V^\qT) \geq \eps\,,\]
    which again puts $\ttC^\qT$ outside of the $\UQMA$ promise range and thus $\mc A^\qT(V) = \mc A^\emptyset(V)$ with probability at least $1-\gamma$, contradicting \Cref{eq:accept_prob}. Thus, we conclude that
    \[ 1 - 2 \eps \leq \E_{\qT \sim \haar_{N,k}}\|V^\qT\|_{\uparrow 2} \leq 1 + 2\eps\,. \]
    Putting this all together, we have shown that with probability at least $1 - \gamma$, $\lambda_1(V^\qT) \geq 1 - 3\eps$; and with probability $1-\gamma$, we get that $\lambda_1(V^\qT) + \lambda_2(V^\qT) \leq 1+3\eps$. Therefore, assuming~\Cref{eq:accept_prob}, then with probability at least $1 - 2\gamma$ over $ T \sim \haar_{N,k}$ it holds that
    \[ \lambda_1(V^\qT) \geq 1 - 3 \eps \quad \text{and} \quad \lambda_2(V^\qT) \leq 6 \eps\,.\]
    \begin{comment}
    Notice we did not use the dimension of $\qT$ in the above argument (except that it is at non-zero). Thus, repeating the above argument with $\qS$ in place of $\qT$ yields the same bound: with probability $1 - 2\gamma$ over $ S \sim \haar_{N,1}$ it holds that
    \[ \lambda_1(V^\qS) \geq 1 - 3 \eps \quad \text{and} \quad \lambda_2(V^\qS) \leq 6 \eps\,.\]
    \end{comment}
    But together with the assumption on $\|V^S\|_\text{op}$, this means that the circuit $V$ solves the $\QApxDim(0, k, N)$ problem with error $\leq 6 \eps$ on $1 -
    2 \gamma$ fraction of instances; this contradicts the lower bound established in \Cref{thm:query_lower_bound}. Thus, we conclude that one of the inequalities in \Cref{eq:accept_prob} does not hold. 
    
    It's clear that the converse of the first item of \Cref{eq:accept_prob} yields the desired conclusion. For the second item, if the average is below $1 - 3 \eps$, then by Levy's Lemma it concentrates below $1 - 2\eps$. \Cref{lem:ext_indistinguishable} implies that the operator norm over $\qT$ also concentrates below $1 - \eps$; but then $\qT$ is an invalid instance and $\mc A^\qT(V) = \mc A^\emptyset(V)$ as desired.

    \paragraph{Case 2:  $\|V^{\emptyset}\| \geq 1/2$} In this case $\mathcal{A}^{\emptyset}(V) = 1$. Again consider the expected operator norm of $V^\qT $ over random choice of $\qT \sim \haar_{N,k}$. First, if $2 \eps < \E_\qT \|V^\qT\|_\text{op} < 1 - 2\eps$ then at least $1- \gamma$ fraction of $\qT$ instances are invalid; for these $\mc A^\qT(V) = \mc A^\emptyset(V)$, yielding a contradiction. If $\E_\qT \|V^\qT\|_\text{op} \geq 1 - 2 \eps$, then for $1-\gamma$ fraction of $\qT$'s, $\lambda_1(V^\qT) \geq 1 - 3 \eps$. Consider any $\qT$ such that this condition holds. If either $1 - 3 \eps \leq \lambda_1 < 1 - \eps$ or $\lambda_2(V^\qT) > \eps$ then this is an invalid instance and $\mc A^\qT(V) = \mc A^\emptyset(V)$. Otherwise, it is a valid instance and again $\mc A^\qT(V) = \mc A^\emptyset(V)$.

    \vspace{1em}
    \noindent Thus, $\E_\qS \|V^\qS\|_\text{op} \leq 2 \eps$. Then we immediately have that with probability $1 - \gamma$, $\lambda_1(V^\qS) \leq 3 \eps$. But applying \Cref{lem:ext_indistinguishable} with $k = 0$ and $k' = 1$ contradicts this statement.
\end{proof}

\subsubsection{Putting it all together}
Finally, we prove the main theorem which shows that an oracular $\BQP^{\UQMA}$ verifier cannot solve the oracle problem $\textsf{QEmptyNonEmpty}$
\relativizedsep*

\begin{proof}
Consider any $\BQP^\UQMA$ oracular verifier $V$ and suppose $k \in \omega(\cpoly(n))$. We'll show that for $\qT \sim \haar_{N,k}$, $V$ fails to distinguish $\qT$ from $\emptyset$. When applied to a state $\ket \psi$ encoding a superposition of circuits $\{\texttt C_i\}_{i \in p(n)}$, $\A^\qT$ acts as,
\begin{equation}
    \label{eq:circuit_oracle_def}
    \A^\qT \ket \psi = \A^\qT \sum_{i \in [p(n)]} \alpha_{i} \ket{\texttt C_i} \ket b \ket{\phi_i} = \sum_{i \in [p(n)]}
    \alpha_i \ket{\texttt C_i} \ket{b \oplus \A^\qT(\texttt C_i)} \ket{\phi_i}\,.
\end{equation}
We then write verifier $V^\qT$ as,
\begin{align*}
    V^\qT &= U_{2L} (\Pi_\tcon \otimes \A^\qT) U_{2L-1} (\Pi_\tcon \otimes \O^\qT) U_{2(L-1)} \dots U_2 (\Pi_\tcon \otimes
    \A^\qT) U_1 (\Pi_\tcon \otimes \O^\qT) U_0 \ket 0\,,
\end{align*}
with interwoven calls to $\A^\qT$ and the oracle $\O^\qT$. We can bound $\|V^\qT \ket 0 - V^\emptyset \ket 0\|_2$ via hybridization. Define,
    \begin{align*}
        V[\ell] = \begin{cases}
            U_{2L} (\Pi_\tcon \otimes \A^\qT) \dots U_{t} (\Pi_\tcon \otimes \A_\emptyset) U_{t-1} (\Pi_\tcon \otimes
            \O_\emptyset) U_{2(T-1)} \dots U_0 \ket 0 & \text{if $t$ even and}\\
            U_{2L} (\Pi_\tcon \otimes \A^\qT) \dots U_{t+1} (\Pi_\tcon \otimes \A^\qT) U_{t} (\Pi_\tcon \otimes
            \O_\emptyset) U_{2(L-1)} \dots U_0 \ket 0& \text{if $t$ odd.}
        \end{cases}
    \end{align*}
    Thus,
    \begin{align*}
        \|V^\qT \ket 0 - V^\emptyset \ket 0\|_2 = \|V[2L] - V[0]\|_2 \leq \sum_{\ell=0}^{2L-1} \|V[\ell+1] - V[\ell]\|_2
    \end{align*}
    Expanding the RHS and using unitary invariance of the $\ell_2$-norm yields that,
    \begin{align*}
        \|V[\ell+1] - V[\ell]\|_2 = \begin{cases}
            \|(\A^\qT - \A^\emptyset) \ket{\phi_\ell}\|_2 & \text{if $\ell$ is even and}\\
            \|(\O^\qT - \O^\emptyset) \ket{\phi_\ell}\|_2 = 2\|\Pi_\qT \ket{\phi_\ell}\| & \text{if $t$ is odd.}
        \end{cases}
    \end{align*}
    Thus, we have that
    \[
        \|V^\qT \ket 0 - V^\emptyset \ket 0\|_2 \leq \sum_{\ell=0}^{L-1} \|(A^\qT - \A^\emptyset) \ket{\phi_{2\ell}}\|_2 + 2
        \|\Pi_\qT \ket{\phi_{2\ell+1}}\|_2 \dfn \Delta(\qT)\,,
    \]
    and we want to lower bound the probability over $\qT$ that this is at most $\frac 1 {\exp(n)}$. We bound the second term by \Cref{lem:overlap_concentration}; each term is at most $\O(k/N)$ with probability $\geq 1 - \gamma$, where $\gamma \in \exp(-N^{1/3})$.
    
    For the even case, consider the set $\mc C$ of circuits with non-zero amplitude for $\ket{\phi_0}, \dots,
    \ket{\phi_{2L-1}}$. Since the circuits correspond to polynomial-time quantum verifiers, each makes at most $\cpoly(n)$-queries to $\O^\qT$; thus we may apply \Cref{lem:polarization} and conclude that for each $\ttC \in \mc C$ that
    \[
        \Pr_{\qT \sim \haar_{N,k}}[\mc A^\emptyset(\ttC) \neq \mc A^\qT(\ttC)] \leq f(n)\,.
    \] 
    Moreover, the size of $\mc C$ can be bounded as $C(n) = |\mc C| \leq 2L \cdot 2^n$. Let $\mc
    E(\qT)$ be the event that $\A^\emptyset(\texttt C) = \A^\qT(\ttC)$ for every $\texttt C \in \mc C$. We have,
    \begin{align*}
        \label{eq:all_good_bound}
        \Pr_\qT[\mc E(\qT)] = \Pr_\qT[\forall \texttt C \in C,\ \A^\emptyset(\texttt C) = \A^\qT(\ttC)] &= 1 - \Pr_\qT[\bigvee_{\texttt C \in C}
        \A^\emptyset(\ttC) \neq \A^\qT(\texttt C)]\\
                                                                               &\geq 1 -\sum_{\texttt C \in \mc C}
                                                                               \Pr_\qT[\A^\emptyset(\ttC) \neq \A^\qT(\texttt C)]\\
                                                                               &\geq 1 - 
                                                                               {C(n)}{f(n)}\\
                                                                               &\geq 1 - \exp(-N^{1/4}) \,,
    \end{align*}
    where the second to last inequality holds by \Cref{lem:polarization}. If $\mc E(\qT)$ holds then
    $(\mc A^\qT - \mc A^\emptyset) \ket{\phi_\ell} = 0$ for \emph{any} $t$.

    Combining these bounds we have that with probability at least $1 - L(n) \cdot \gamma - \exp(-N^{1/4}) \geq \exp(-N^{1/4})$,
    \[
        \Delta(\qT) \leq L(n) \cdot \O(k / N)\,.
    \]
    But by the assumption on error, $\Delta(S) \geq 1 - 2 \eps$, and thus this requires $L(n) \in \Omega(N/k)$.
\end{proof}

%% file: apx_counting.tex
\section{\texorpdfstring{Quantum approximate counting versus $\QMA$}{Quantum approximate counting versus QMA}}
\label{sec:qxc}
In this section, we use the techniques from \Cref{thm:query_lower_bound} to obtain an oracle separation between $\QMA$
and quantum approximate counting. We recall our definition,

\qxcclass*

We formalize \Cref{task:qma_classical} as oracle promise problem $\ApxDim(k_1, k_2, N) = (\Omega_\text{Yes}, \Omega_\text{No})$ with $k_2 \leq N$ and $k_1
\leq k_2(1-1/\cpoly(n))$ where,
\begin{itemize}
    \item Each $\O \in \Omega_\text{Yes}$ corresponds to a reflection over a subspace $S \subseteq \C^{N}$ of dimension
        $k_1$, with $\O = \Id - 2 \Pi_S$.
    \item Each $\O \in \Omega_\text{No}$ corresponds to a reflection over a subspace $T \subseteq \C^N$ of dimension
        $k_2$, with $\O = \Id - 2 \Pi_T$.
\end{itemize}
This problem is almost by definition in $\QXC^\O$. Define the $\QMA^\O$ verifier $V^\O$ which applies the oracle $\O$ to the witness controlled on an auxiliary qubit initialized as $\ket +$, then measures the auxiliary register in the Hadamard basis. The
verifier accepts if it obtains the $\ket -$ outcome, and rejects otherwise. Now pick $b = \tfrac 1 3$ and $a = \tfrac 2 3$. For $\O \in \Omega_\text{YES}$, the verifier $V^\O$ accepts any $\ket \psi \in S$ with probability $1$, where $\dim(S) = k_1$. Thus, $D_a = D_b = k_1$. Similarly, for $\O \in \Omega_\text{NO}$, the verifier accepts with probability $1$ when
$\ket \psi \in T$, where $T$ is $k_2$-dimensional so $D_a = D_b = k_2$. Thus, we can solve the oracle promise problem
$\oprob$ by setting $\eps = \frac{k_2-k_1}{2k_2} \geq \frac 1 {\cpoly(n)}$ and querying a $\textsf{QXC}^\O$ algorithm on $V^\O$.

\subsection{\texorpdfstring{$\QMA$ versus $\QXC$}{QMA versus QXC}}

First we show that by using techniques developed in this work we can easily establish strong query lower bounds for any $\QMA$ verifier deciding $\ApxDim(k_1, k_2, N)$. As for $\UQMA$, when working with lower bounds against un-relativized classes, it suffices to consider classical (computational basis) oracles.

\begin{remark}
    A similar lower bound where dimension $k_1$ corresponds to the YES case and dimension $k_2$ corresponds to the NO case is obtained in \cite{she2023unitary, Aaronson_AKKT2020_QuantumLowerBounds} but through much more intricate methods utilizing a version of the polynomial method for Laurent polynomials. Our proof is much simpler and gets the same quantitative bound, which is shown to be tight in \cite[Theorem 5]{Aaronson_AKKT2020_QuantumLowerBounds}. However, their techniques are better in that their lower bound is symmetric (it also works when the YES case is the larger subspace, and the NO case is smaller). Intuitively, this is because their proof goes via the inclusion of $\QMA$ in $\textsf{SBQP}$, the exponentially precise analogue of $\BQP$, and thus is not sensitive to a witness.
\end{remark}

\begin{theorem}[$\QMA$ query lower bound for $\ApxDim$]
    \label{thm:qxc_query_lower}
    Let $k_1 \in o(N)$ and $k_2 = 2k_1$. Let $V^\O$ be an oracular $\QMA$ verifier deciding $\ApxDim(k_1, k_2, N)$ with error $\eps \leq \exp(-n)$. Then, $L \in \Omega(\sqrt{N/k_1})$.
\end{theorem}
\begin{proof}[Proof of \Cref{thm:qxc_query_lower}]
    As in \Cref{thm:query_lower_bound}, the idea is to consider a pair of oracles $\O^S$ and $\O_{T}$, which correspond
    to a ``Yes'' and ``No'' case respectively. We choose $T = S \oplus \Delta$, where $\Delta$ is a $\delta = k_2 -
    k_1$-dimensional subspace, orthogonal to $S$, where we've chosen $k_2 - k_1 = k_1$. Then, we will hybridize between the verifiers corresponding to each
    oracle. Let $V$ be a $L$-query verifier deciding $\ApxDim(k_1, k_2, N)$. $V$ can be represented
    by an alternating sequence of unitaries and calls to an oracle $\O$, with a measurement performed at the end. In
    particular,
    \[
        V^\O \dfn \Pi_\tout \O U_{L-1} \dots U_1 \O \Pi_\tin\,.
    \]
    We write the hybridized inner sequence of unitaries and oracles as,
    \[
        \mc U^{(S,T)}_t = \O^T U_{L-1} \dots \O^T U_{t} \O^S U_{t-1} \dots U_1 \O^S\,,
    \]
    where the first $t$ queries are to $S$ and the final $L-t$ are to $T$. Let $\ket{\phi_S} \dfn \ket{\psi_S}\ket 0$
    correspond to an optimal YES case witness for $\O^S$. Then, we have the following completeness and soundness
    guarantees,
    \begin{alignat}{2}
        &\text{(Completeness)} &\qquad &\|V^S (\ket{\phi_S} \ket 0)\|_2 = \|\Pi_\tout \mc U^{(S,T)}_L (\ket{\phi_S} \ket 0)\|_2 \geq 1 -
        \eps\\
        &\text{(Soundness)} &&\|V^T (\ket{\phi_S} \ket 0)\|_2 = \|\Pi_\tout \mc U^{(S,T)}_0 (\ket{\phi_S} \ket 0)\|_2 \leq \eps\,.
    \end{alignat}
    This implies that there exists a critical index $t$ such that $\|(U^{(S,T)}_{t+1} -
    U^{(S,T)}_\ell)(\ket{\phi_S} \ket 0)\|_2 \geq \frac{1-2\eps}{L+1}$. Equivalently,
    \begin{align*}
        \frac{1-2\eps}{L+1} &\leq \|(U^{(S,T)}_{\ell-1} - U^{(S,T)}_\ell)(\ket{\phi_S} \ket 0)\|_2\\
                            &= \|(\O^\Pi_{\Delta} - \Id) \O^S U_{\ell-1} \dots (\ket{\phi_S} \ket 0)\|_2\\
                            &= 2\|(\Pi \otimes \Pi_\Delta) U_{\ell-1} \dots (\ket{\phi_S} \ket 0)\|_2\,,
    \end{align*} 
    which implies that $\sket{\tilde \phi_\ell} \dfn U_{\ell-1} \dots (\ket{\phi_S} \ket 0)$ has overlap at least $\frac
    {1-2\eps}{2(L+1)}$ with $\Pi \otimes \Pi_{\Delta}$. Finally, we can pick $K = \frac{N-k_1}{\delta} = \frac{N-k_1}{k_1}$-many orthogonal $\delta$-dimensional subspaces $\Delta_1, \dots, \Delta_K$. Furthermore, there are at least $K/(L+1)$-many
    $\Delta_i$'s which share the same critical index $t$. Picking $k_1$ so that
    \begin{equation}
        \label{eq:qxc_lower}
        K > \frac{(L+1)^2}{1 - 2\eps}
    \end{equation}
    yields a contradiction as this implies $\sket{\tilde \phi_{\ell}}$ has norm greater than $1$. Therefore, no $\QMA^\O$ verifier exists unless $L \geq \sqrt{K(1-2\eps)} \in \Omega(\sqrt{N/k_1})$.
\end{proof}

Combining the above theorem, the $\QXC^\O$ algorithm for $\ApxDim(k_1, k_2, N)$, and a standard
diagonalization argument (as in the proof of \Cref{cor:oracle_separation}), we obtain the following corollary.

\begin{corollary}\label{cor:qmavsqxc}
    There exists a classical oracle $\O$ such that $\QXC^\O \not\subseteq \QMA^\O$.
\end{corollary}

Before moving to quantum oracles, we state one last strengthening of the above result. Over classical oracles, a variant
of \Cref{lem:polarization} holds, but we use Markov's inequality rather than Levy's Lemma and thus we take $f(n) \in
\exp(-n)$. Still, this is sufficient to union bound over all at most polynomial circuits queried by a $\Ptime^\QMA$
verifier. We obtain the following corollary.

\begin{corollary}\label{cor:pqmavsqxc}
    There exists a classical oracle $\O$ such that $\QXC^\O \not\subseteq \Ptime^{\QMA^\O}$, and in fact  $\mathsf{SBP} \not \subseteq \Ptime^{\QMA}$.
\end{corollary}

\subsection{\texorpdfstring{$\QMA^{\QMA}$ versus $\QXC$}{QMA to the QMA versus QXC}}

We can do better than the previous section, and by working with quantum oracles, we obtain a separation from between $\QMA^\QMA$ and $\QXC$. The key tool is an observation on the polarizing property of $\QMA$ circuits, as well as
a strengthening of a $\QMA$ query lower bound for $\QApxDim$ due to
\cite{she2023unitary}, based on the polynomial method of \cite{beals2001quantum}. One may wonder if these techniques immediately would yield lower bounds against $\QMA^\QMA$ (or even constant stacks of $\QMA$). However, the polynomial method by itself is not sufficient to establish a lower bound against verifiers calling oracles for $\QMA$ as there is no obvious degree bound on the steps of the algorithm querying the oracle. We circumvent this issue by a similar argument used in the proof of \Cref{thm:relativized_separation}: the acceptance probability of $\QMA$ circuits over random subspace oracles polarize and we can argue that either these oracle calls do not help \emph{or} we can apply the lower bound of \cite{she2023unitary}.

Before showing the polarization result, a couple technical refinements of \cite{she2023unitary} are necessary. First, as stated, their lower bound only holds for algorithms succeeding on all instances. We will need to extend this to work for algorithms succeeding with high probability. Second, we use our standard hybridization argument to complete the lower bound; this works by considering a fixed subspace $\qS$ and drawing $T \sim \mc Q_{\qS \uparrow, k}$. However $(\qS, T)$ where $T = S \oplus \Delta$ make up a very small fraction of possible pairs, and it
is plausible that a $\QMA$ verifier succeeds at distinguishing only these very small fraction of instances. Nonetheless, we show that a $\QMA$ algorithm succeeds at distinguishing extensions 
\emph{only if} it succeeds on random pairs $(\qS, T)$.

\begin{theorem}[Probabilistic variant of Theorem 4.2 of \cite{she2023unitary}]
    \label{thm:apx_dim_lower}
    Let $V$ be $m$-qubit $\QMA$ verifier taking in a $w \leq m$ qubit witness and making $L$ queries to an $n$-qubit
    oracle encoding an unknown subspace of $\C^N$. Suppose that $V$ solves the oracle promise problem in the sense that
    \[
        \Pr_{\substack{\qS \sim \mc H_{N, k_1} \\ T \sim \mc H_{N, k_2}}}[| \textup{acc}(V^\qS) - \textup{acc}(V^\qT) | \geq \eps ] \geq 1 - \gamma\,,
    \]
    where $\gamma \in \O(\exp(-N))$, $\eps$ is a constant smaller than $\tfrac 1
    2$, and $k_1 \in \O(k_2)$. Then, $L \in \Omega(\min\{k_2, \sqrt{N/k_2}\})$.
\end{theorem}

\begin{proof}
    We begin by removing the dependence of $V$ from the witness. By first applying the in-place amplification of
    \cite{Marriott_MW2005_QuantumArthurMerlinGames} to suitably amplify the error $\eps$ then replacing the witness with
    the maximally mixed state, we obtain a $L'$-query verification algorithm $\tilde V$ which takes no witness and,
    \begin{itemize}
        \item with probability $1-\gamma$ over $\qS$, $\text{acc}(V^\qS) \geq 2^{-w}$ and,
        \item with probability $1-\gamma$ over $\qT$, $\text{acc}(V^\qT) \leq 2^{-10w}$.
    \end{itemize}
    Moreover, $L' \in \O(w L)$. Next, obtain a symmetrization of our query algorithm as a degree $2L'$ polynomial in the \emph{eigenvalues} of the subspace reflection oracle. In our setting, these eigenvalues correspond to strings $(\lambda_1, \dots, \lambda_N) \in \{-1, 1\}^N$ where $\text{wt}_{=-1}(\vec \lambda) \dfn |\{i \st \lambda_i = -1\}|$ corresponds to the subspace dimension. To obtain this symmetrization, we only require a slight modification of the proof of Theorem 1.3 of \cite{she2023unitary} given in Section 3.1 therein. We sketch this modification briefly. First, given the verifier $\tilde V$, it's acceptance probability can be expressed as a polynomial $p(U, U^\star)$.\footnote{The input $U^\star$ is to allow access to the \emph{inverse} of the oracle; this is not strictly necessary for us, but either way it only increases the degree by a constant multiplicative factor, and thus we retain it for full generality.}
    The symmetrized polynomial $q$ is obtained by taking the Haar average of $U$, $U^\star$ so that
    \begin{equation}
        \label{eq:symmetrization}
        q(U, U^\star) = \E_{g \sim U(d)} p(g U g^\star, g U^\star g^\star)\,.
    \end{equation}
    Finally, the authors observe that the right-hand side simplifies to a linear combination of terms of the form $\Tr[U^r]$ or $\Tr[(U^\star)^s]$ for some integers $r, s$. These expressions clearly only depend on the eigenvalues of $U$ and $U^\star$.

    In our approximate setting, we do not have the exact equality in \Cref{eq:symmetrization}. However, these expressions still hold with an additive error $\gamma$, accounting for the fraction of $\qS, T$ for which the verifier may behave poorly. Thus, we obtain a polynomial $p$ for which
    \begin{equation}
        p(\lambda_1, \dots, \lambda_N, \lambda_1^\star, \dots, \lambda_N^\star) \text{ is } \begin{cases}
            \geq 2^{-w} - \gamma & \text{ when }\text{wt}_{=-1}(\vec \lambda) \geq k_2\text{ and,}\\
            \leq 2^{-10w} + \gamma & \text{ when }\text{wt}_{=-1}(\vec \lambda) \leq k_1\,.
        \end{cases}
    \end{equation}
    At this point, taking the polynomial $p$ in  \cite[Theorem 3]{Aaronson_AKKT2020_QuantumLowerBounds} to be as defined above yields the desired lower bound.
\end{proof}

Next, we extend the result to a fixed $\qS^\star$ and random extensions.
\begin{lemma}[Query lower bound with fixed $\qS^\star$ and extensions $\qT$]
    \label{lem:lower_bound_fixed_s}
    Fix a verifier $V$. Suppose for some constant $c$ with probability at least $p \in \Omega(\exp(-n^c))$ over $S \sim
    \haar_{N,k_1}$,
    \begin{equation}
        \label{eq:cond_fixed_s}
        \Pr_{\qT \sim \mc Q_{\qS \uparrow,k_2}}[| \textup{acc}(V^\qS) - \textup{acc}(V^\qT) | \geq \eps] \geq 1 - \gamma\,,
    \end{equation}
    where $\gamma \in \O(\exp(-N^{1/3}))$. Then $L \in \Omega(\sqrt{N/k_1})$.
\end{lemma}

\begin{proof}
    Let $\mu = \E_S [\textup{acc}(V^\qS)]$. As usual, we apply \Cref{lem:levyslemma} to bound $\textup{acc}(V^\qS)$ within $\eta$ of $\mu$ with probability at least $1 - \gamma'$, with $\eta = N^{-1/3}$ and $\gamma' \in \exp(-N^{1/3})$. Then with probability at least $p - \gamma'$ over $\qS$, we have that both $|\textup{acc}(V^\qS) - \mu| \leq \eta$ and \Cref{eq:cond_fixed_s} holds so
    that
    \[
        \Pr_{\qT \sim \mc Q_{\qS \uparrow, k_2}}[| \text{acc}(V^\qT) - \mu | \geq \eps + \eta] \geq 1 - \gamma\,.
    \]
    But since this holds for $1 - \gamma$ fraction of $T$ extending $p - \gamma'$-fraction of $S$, the above statement
    also holds for $p - \gamma - \gamma'$ fraction of $T$, drawn from $\haar_{N, k_2}$, \emph{independent of $\qS$}.
    By choice of $p$, $p - \gamma - \gamma' \gg \gamma$, and thus $\E_{\qT \sim \haar_{N,k_2}}
    [\textup{acc}(V^\qT)] = \mu'$ with $|\mu' - \mu| \geq \eps - 2 \eta$. Once again using concentration of the operator
    norm, this implies that
    \[
        \Pr_{\qT \sim \haar_{N, k_2}} [ |\textup{acc}(V^\qT) - \mu'| \geq \eta] \leq 1 - \gamma\,.
    \]
    Combining these inequalities, we have that with probability $1 - \gamma$ over $S$ and $1- \gamma$ over $T$
    independently, $\textup{acc}(V^\qT) \geq \mu' - \eta$ and $\textup{acc}(V^\qS) \leq \mu + \eta$ so that 
    \[
        |\textup{acc}(V^\qS) - \textup{acc}(V^\qT)| \geq \eps - 4 \eta\,.
    \]
    Since $\eps - 2 \eta \in \Omega(1)$, $V$ satisfies the same query lower bound as in \Cref{thm:apx_dim_lower}.
\end{proof}

The important take-away from this lemma is that for verifiers making only polynomially-many queries to the oracle, there
exists very few $\qS$'s (smaller the inverse exponential) for which \Cref{eq:cond_fixed_s} holds.
\begin{comment}
\begin{lemma}[Query lower bound with fixed $\qS^\star$ and extensions $\qT$]
    \label{lem:lower_bound_fixed_s}
    Let $V$ be a $L$-query $\QMA$ verifier such that for a fixed subspace $\qS^\star$ with dimension $k_1$ such that
    \[
        \Pr_{\qT \sim \mc Q_{\qS \uparrow}(k_2)}[| \text{acc}(V^\qS) - \text{acc}(V^\qT) | \geq \eps] \geq 1 - \gamma\,,
    \]
    where $\gamma \in \O(\exp(-N^{1/3}))$. Then $L \in \Omega(\sqrt{N/k_1})$.
\end{lemma}

\begin{proof}
    Consider any $L$-query $\QMA$ verifier $V$, taking in a witness on $w$ qubits. Let $\qS'$ be distributed according to
    $\mc Q_{\qS^\star \uparrow}(k_1)$. From \Cref{lem:ext_indistinguishable} we have that with probability at least $1 -
    \exp(-N^{2/3})$ over $\qS'$, $\| V^{\qS^\star} - V^{\qS'}\| \leq \O(k_1/N)$. In particular, this implies that $|\text{acc}(V^{\qS^\star}) -
    \text{acc}(V^{\qS'})| \leq \O(k_1/N)$ and $V$ also distinguishes random extensions $\qS'$ from $\qT$ with error $\eps'
    \leq \eps + \O(k_1/N)$ with probability $1 - \gamma - \exp(-2^{n/3}) \leq 1 - 2\gamma$. For sufficiently large $n$, $\eps < \tfrac 1 2$ implies $\eps' < \tfrac 1 2$.

    This yields a distinguisher between $\qS'$ and $\qT$, both extending $\qS^{\star}$. But now, we ``hard-code'' the
    subspace $\qS^{\star}$ into the verifier; this yields a new verifier which solves the oracle promise problem $\QApxDim(k_1 - 1,
    k_2 - 1, N - 1)$. Applying \Cref{thm:apx_dim_lower} yields
    the desired lower bound.
\end{proof}
\end{comment}
Finally, we show our polarization result.

\begin{lemma}[Polarization of $\QMA$ oracle circuits]
    \label{lem:oracle_polarization}
    Let $\ttC$ be any $n' = \cpoly(n)$-qubit $\QMA$ circuit which makes $L \in o(\sqrt{N/k_1})$ queries to a $n$ qubit oracle $\O_\qS = \Id - 2
    \Pi_\qS$. Let $\qS$ be any dimension $k_1$ subspace such that \Cref{eq:cond_fixed_s} fails to hold. Let $\qT$ be a
    $k_2 < N^{2/3}$ dimensional extension of $\qS$. Then for some decision-extension $\mc A$ of an oracle for $\textsf{QCiruitSAT}(\eta, 1- \eta)$,
    \[
        \Pr_{\qT \sim \mc Q_{\qS\uparrow, k_2}} [ \mc A(\ttC^\qS) \neq \mc A(\ttC^\qT)] \leq 2 \gamma\,,
    \]
    where $\gamma \leq \exp(-N^{1/3})$.
\end{lemma}

\begin{proof}
    Let $\mu_1=\|\ttC^\qS\|_\text{op}$. Define $\mc A$ so that on valid $\QMA$ circuits, it outputs
    correctly (accepting if the operator norm is at least $1 - \eta$, rejecting if at most $\eta$). Otherwise, if the
    input circuit falls outside of the promise gap, it returns $1_{\mu_1 > 1/2}$. Notice that if $2 \eta < \mu_2 < 1 -
    2 \eta$, then with probability $1 - \gamma$, $\ttC^\qT$ falls outside of the promise gap and thus $\mc A(\ttC^\qT) = \mc
    A(\ttC^\qS)$ with probability at least $1 - 2\gamma$. We may assume that either $\mu_2 \leq 2 \eta$ or $\mu_2 \geq 1
    - 2 \eta$. We proceed by casing on $\mu_1$.

    \paragraph{Case $\mu_1 \leq \tfrac 1 2$.} In this case $\mc A(\ttC^\qS)$ rejects (either as it is invalid or falls in
    the NO case). Now consider $\ttC^\qT$. If $\mu_2 \leq 2\eta$ then with probability $1 - \gamma$, either the circuit
    $\ttC^\qT$ is invalid ($\|\ttC^\qT\| \in (1 - \eta, 3 \eta)$) or it corresponds to a reject instance. In both cases
    $\mc A(C^\qT) = \mc A(C^\qS)$. Otherwise, if $\mu_2 \geq 1 - 2\eta$ then $\ttC$ solves $\oprob$ with completeness and
    soundness parameters $(1 - 3 \eta, \tfrac 1 2)$ on all but a $2\gamma$-fraction of instances. Applying the lower
    bound of \cite{she2023unitary} via \Cref{lem:lower_bound_fixed_s} yields a contradiction.

    \paragraph{Case $\mu_1 > \tfrac 1 2$.} This case is symmetric. With probability $1 -\gamma$, $\mc A(\ttC^\qS)$
    accepts. If $\mu_2 \geq 1 - 2 \eta$ then with probability $1- \gamma$, $\ttC^\qT$ is either invalid or yields an
    accepting instance. In both cases, $\mc A (\ttC^\qT)$ accepts. Otherwise, if $\mu_2 \leq 2 \eta$, then $\ttC$ solves
    the ``small-to-big'' variant of $\oprob$ with completeness/soundness $(\frac 1 2, 3\eta)$. This again yields a
    contradiction.
\end{proof}

\begin{theorem}[Query lower bound for $\QMA^\QMA$]
    Let $V$ be any $m = \cpoly(n)$-qubit, $L$-query oracular $\QMA$ verifier which purports to solve the $n$-qubit oracle promise
    problem $\QApxDim(k_1,k_2,N)$ with error $\eps \leq \exp(-n)$ and $k_1 < k_2 \leq N^{2/3}$. Suppose $V$ has gates
    implementing controlled calls to a $n$ qubit subspace reflection oracle $\O^\qS$,
    as well as an $n_{o} \leq n'$ qubit oracle $\mc A$ for $\textsc{\textup{QCircuitSat}}^{\O_\qS}$. Then $L \in \Omega(N^{1/3})$.
\end{theorem}

\begin{proof}
    We assume each of the $L$ queries to $\mc A$ are interleaved with queries to $\O$, so that queries to $\mc A$ happen
    on even indices $\ell$, while queries to $\O$ are on odd indices. Since $V$ queries $\mc A$ on circuits encoded in
    at most $n_o$ qubits, this implies that there are at most $2^{n_o}$ possible queried circuits. Let this set of
    circuits be $\ttC$. Since each encoded circuit makes $\cpoly(n) \in o(\sqrt{N/k_2})$ queries to the oracle, we have
    that \Cref{eq:cond_fixed_s} of \Cref{lem:lower_bound_fixed_s} holds for strictly less than $1 / 2^{n_o}$ fraction of
    $\qS$'s. Thus, take some $\qS^\star$ for which \Cref{eq:cond_fixed_s} fails for \emph{all} queried circuits $\ttC$.
    Letting $\qT \sim \mc Q_{\qS^\star \uparrow, k}$, we follow the same hybridization as in the proof of
    \Cref{thm:relativized_separation} which requires us to bound the quantity,
    \begin{align*}
        \|V[\ell+1] - V[\ell]\|_2 = \begin{cases}
            \|(\A^\qT - \A^{\qS^\star}) \ket{\phi_\ell}\|_2 & \text{if $\ell$ is even and,}\\
            \|(\O^\qT - \O^{\qS^\star}) \ket{\phi_\ell}\|_2 = 2\|\Pi_\Delta \ket{\phi_\ell}\| & \text{if $\ell$ is odd,}
        \end{cases}
    \end{align*}
    where $\ket{\phi_\ell}$ only depends on ${\qS^\star}$ (in particular, $\ket{\phi_1}$ encodes the witness to the outer $\QMA$
    verifier). We bound each of these over random extensions $\qT \sim \mc Q_{{\qS^\star} \uparrow, k}$.

    For the first quantity, we recall from \Cref{eq:circuit_oracle_def} that $\mc A^{\qS^\star} \ket{\phi_\ell} = \mc A^{\qT}
    \ket{\phi_\ell}$ if $\mc A^{\qS^\star}(\ttC) = \mc A^\qT(\ttC)$ for each circuit $\ttC$ encoded by $\ket{\phi_\ell}$. For any
    fixed $\ttC$, \Cref{lem:oracle_polarization} yields that $\Pr_\qT[\mc A^{\qS^\star}(\ttC) = \mc A^\qT(\ttC)] \leq
    \exp(-2^{n_o})$. Union bounding over all queried circuits $\mc C$, we have that the probability that $\mc A^{\qS^\star}
    \ket{\phi_\ell} = \mc A^\qT \ket{\phi_\ell}$ \emph{for any $t$} is at most $|\ttC| \cdot \exp(-2^{n_o}) \leq 2^{n_o}
    \exp(-2^{n_o}) =: p_1$.

    For the second, \Cref{lem:overlap_concentration} yields
    \[
        \Pr_{\Delta}[\|\Pi_\Delta \ket{\phi_\ell}\| \geq 2N^{-1/3}] \leq \exp(-N^{1/3})\,.
    \]
    Thus, we may union bound over all $L$ indices corresponding to this case, yielding that $\| (\O^\qT - \O^{\qS^\star})
    \ket{\phi_\ell}\|_2 \geq L N^{-1/3}$ for any $\ell$ with probability at most $p_2 := \exp(-N^{1/3})$.

    To complete the hybridization, with probability at least $1 - p_1 - p_2 \geq 1 - \exp(-N^{1/3})$ over random $\qT$
    extending ${\qS^\star}$,
    \[
        \sum_{\ell=1}^L \|V[\ell+1] - V[\ell]\|_2 \leq \O(L N^{-1/3})\,.
    \]
    Therefore, to achieve the stated error $\eps$, this requires $L \in \Omega(N^{1/3})$.
\end{proof}

\begin{corollary}\label{cor:qxc-vsqma2}
    There exists a quantum oracle $\O$ such that $\QXC^\O \not\subseteq \QMA^{\QMA^\O}$.
\end{corollary}

%% file: eth.tex
The oracle separations in the previous section suggest that we need to use the structure of $\QMA$ problems if we want to have a chance to design a unique verifier. In this section, we describe a $\UQMA$ verification protocol for local Hamiltonians that are sufficiently chaotic in the sense of the eigenstate thermalization hypothesis (ETH) -- an influential hypothesis in quantum statistical mechanics. 
We start by describing a computational variant of ETH. We then describe the protocol and analyze its performance under this assumption. Finally we discuss implications of the $\UQMA$ vs $\QMA$ question to the thermalization properties of $\QMA$-hard Hamiltonians.

\subsection{Eigenstate thermalization hypothesis}\label{sec:ethoverview}

Thermodynamics expects the evolution of a generic many-body system to be ergodic, meaning that the system evolves to eventually explore all possible configurations and thermalize. However, the evolution of a closed quantum system is governed by the Schrodinger's equation which is reversible and also leaves eigenstates invariant. Why do individual energy eigenstates of a large, interacting quantum system look thermal when probed by ``simple'' observables (like local operators), even though the full wavefunction evolves unitarily without randomness? A widespread explanation to this seeming paradox is the eigenstate thermalization hypothesis (ETH)~\cite{deutsch1991quantum, srednicki1999approach}, which has produced many predictions in agreements with thermalization on many physical systems~\cite{d2016quantum}. ETH has many slight variants, but they are typically a set of assumptions about generic Hamiltonians that, informally, predicts that the eigenstates locally resemble a Gibbs state, and furthermore, local observables obey certain random matrix behaviors in the Hamiltonian energy eigenbasis inspired by quantum chaos and random matrix theory.

The first mathematically clean and consistent description of ETH was proposed by Srednicki~\cite{srednicki1999approach}, which is in turn inspired by Wigner's pioneering random matrix theory (RMT) framework of analyzing quantum systems using statistical methods. Srednicki's ETH proposes a ``chaotic'' behavior of local operators in the eigenbasis of a local Hamiltonian. It is hard to characterize precisely how a local operator $A$ may act on an energy eigenstate $\ket{\nu_i}$, so ETH leverages randomness to state that the matrix elements of a local operator $A$ ``looks like'' $\bra{\nu_j}A\ket{\nu_i} = c_i\delta_{i,j} + e^{-S/2}R_{ij}$, where $c_i$ is some quantity capturing $\bra{\nu_i}A\ket{\nu_i}$ and $\ket{\nu_i}, \ket{\nu_j}$ are eigenstates with nearby eigen-energies (a local operator cannot connect far away eigenstates - see \cite{Arad_AKL2016_ConnectingGlobalLocal}) and $S$ is the entropy within this window.
Crucially, $R_{ij}$ are chosen to be i.i.d random Gaussians. More formally, we consider the following formulation closely following~\cite[Section IV]{Chen_CB2023_FastThermalizationEigenstate}, which is arguably more mathematically concrete than~\cite{srednicki1999approach} to work with at a rigorous level.

Let us first introduce some notations. Fix an energy $e_0 >n^{0.99}$ and $\Delta = 1/\mathrm{poly}(n)$ throughout this section\footnote{ETH is usually defined for ``extensive" energies, those that scale linearly with the system size. Our threshold $n^{0.99}$ is just a way to capture this without introducing another parameter to set the extensive threshold.}. Denote by $\mc I_\Delta$ the energy window of width $\Delta/2$ around $e_0$,
    \(
        \mc I_\Delta = [e_0 - \tfrac \Delta 2, e_0 + \tfrac \Delta 2].
    \) Denote by $\Pi_\Delta$ the projector onto the subspace of eigenstates
   with energy $e \in \mc I_\Delta$. We will always think of the eigenstates as ordered by non-decreasing energies.

\begin{hypothesis}[Srednicki's ETH ansatzes]\label{ETH}
 Consider an $n$-qubit local Hamiltonian $H$ and a narrow energy window
 of width $
\Delta$ around an energy $e_0$. The corresponding ``ETH ansatz ensemble'', denoted $\mathcal{G}$, is a distribution over Hermitian operators $A \sim \mc G$ with parameters $\{\mu_\alpha\}_\alpha$ and $\{f_{\alpha, \beta}\}_{\alpha, \beta}$, whose entries within the said window have the form
\begin{equation}
\bra{\nu_\alpha}A\ket{\nu_\beta} = \mu_\alpha\delta_{\alpha,\beta}+ \frac{f_{\alpha,\beta}}{\sqrt{\Tr(\Pi_{\Delta})}}
g_{\alpha,\beta},
\end{equation}
whenever the energies of the eigenstates $\ket{\nu_\alpha}$, $\ket{\nu_\beta}$ are in the interval $\mathcal{I}_\Delta$. In the off-diagonal entries, $g_{\alpha, \beta} = g^{*}_{\beta, \alpha}$ are complex i.i.d. Gaussians
with mean $0$ and variance $1$. The real coefficients $f_{\alpha, \beta}$ represent flexibility in the variance of the Gaussians, subject to $f_{\alpha,\beta}=f_{\beta, \alpha}$ due to the Hermitian condition and $f \leq f_{\alpha,\beta} \leq 1$,
where $f$ is a constant. The diagonal entries $\mu_\alpha =\bra{\nu_\alpha}A\ket{\nu_\alpha}$ are non-random, but smoothly vary within the window $\mathcal{I}_\Delta$ (as a function of the eigenenergy). The normalization factor $\Tr(\Pi_{\Delta})$ is simply the number of eigenstates within this window, as $\Pi_{\Delta}$ denotes the projector onto them. All these prescriptions are subject to the (global) constraints $A=A^\dagger$ and $\|A\|=1$. \footnote{This can be guaranteed by the countering fluctuations in the entries outside of the window $\mc I_\Delta$.} 

 We say that $H$ and a collection of local \footnote{The locality of $A_i$ is typically allowed up to $o(n)$, but $A_i$ should be simple (like Pauli, or efficiently implementable). For this work, it suffices to think of $A_i$ as constant-local Pauli operators.} observables $A_1, \ldots, A_m$ satisfy ETH within the $\mc I_{\Delta}$ window around energy $e_0$ if these operators ``look like'' randomly drawn from the ETH ensemble.
 
\end{hypothesis}

\begin{remark}[On the magnitudes of $\Delta$ and $f$] See \cite[Section IV]{Chen_CB2023_FastThermalizationEigenstate} for a discussion on the value of $\Delta=1/\poly(n)$ for which ETH is expected to hold. Its value can depend on $e_0$, but we will only consider a fixed energy $e_0$ throughout and thus ignore such considerations to keep the model simple.
The assumption of $f$ being a system size-independent constant, although being quite standard in the ETH literature, could be relaxed. Our arguments will also work if $f$ is inverse polynomially small for a mild polynomial. For example,~\cite{d2016quantum} numerically suggested that $f = \Omega(1/n^{1/D})$ if the Hamiltonian is on a $D$-dimensional lattice. The condition $f_{\alpha,\beta}\leq 1$ can be relaxed to a larger upper bound, but we do not consider this to keep track of less variables.
\end{remark}
For technical reasons, we will introduce two smoothness properties, which are often not explicitly included in ETH statements, but are expected to hold.

\begin{property}[Non-clustering of eigenstates]
    \label{as:smooth} Consider an energy window $\mc I_\omega \subset \mc I_\Delta$, where $\omega = \Delta-q(n)$ and $q(n)=1/\poly(n)$ is sufficiently small. It holds that
    $
        \frac{\Tr(\Pi_{\omega})}{  \Tr(\Pi_\Delta)} \geq C':=1 - f^{25}.
    $
\end{property}
This non-clustering property is used in the proof of
\Cref{clm:Qdoesntmatter} when we analyze a quantum phase estimation subroutine, which has finite energy resolution. This property is justified since within the small window $\mc I_\Delta$ we expect the eigenstate density to be smooth. In fact,~\Cref{as:smooth} is somewhat implicit by ETH since ETH is supposed to reduce to RMT when the energy interval is small~\cite{d2016quantum}, and RMT ensemble has a smooth density of states.

Similarly, we expect the quantities $\mu_\alpha$ do not vary too much within the window $\mc I_\Delta$.\footnote{Indeed, ETH says that an eigenstate locally resembles the Gibbs state. The temperature variation along the energy window $\mc I_\Delta$ can only be $1/\poly(n)$. }
This brings us to our next assumption, which simplifies an analysis in the proof of~\Cref{thm:effectiveop}.

\begin{property}[Smoothness of diagonal entries]
    \label{prop:muvariation}
    It holds that $|\max_{\alpha\in \mc I_\Delta} \mu_\alpha - \min_{\alpha\in \mc I_\Delta}\mu_\alpha| \leq f^{7}$.
\end{property}

 We refer the interested readers to~\cite{deutsch2018eigenstate} to a recent short overview of ETH that focuses more on the physical motivations, and the much more comprehensive~\cite{d2016quantum} that further consists of mathematical foundations, random matrix theory background, applications of ETH, the various pieces of evidence, and important references in the vast ETH literature.
 
 A quick objection to the ETH ansatz is - where is the randomness coming from? And what does ``look like'' mean? Clearly the Hamiltonian is fixed and the statement is being made for all eigenstates in the energy window. The view is that we make a choice of $g_{\alpha \beta}$ by drawing from iid Gaussian and then fix them for the entire calculation (and any predictions are taken ``with high probability''). Given that $H$ is a complicated object, the belief is that we would be unable to tell a difference (with high probability over the ETH ensemble) if we did calculations using the precise entries of the $A$'s in the energy eigenbasis. Such behaviors have been experimentally and numerically verified in many physical quantum systems~\cite{rigol2008thermalization}, as well as generic, random Hamiltonians, such as the famous SYK model, where ETH is postulated to hold, see~\cite[Section 4.3]{d2016quantum}.

In this work, we take this computational indistinguishability viewpoint to design a unique verification protocol for ETH Hamiltonians, but we hope the formulation below could be useful for other applications of ETH in complexity theory.
 
\begin{hypothesis}[Computational ETH]\label{hypo:cETH} Setup:

\textbf{\textup{(The Hamiltonian)}} Consider an $n$-qubit local Hamiltonian $H$ and a fixed energy window $\mc I_\Delta $ around $e_0$. Suppose that $H$ satisfies \Cref{as:smooth} (non-clustering of eigenstates) in the window $\mc I_\Delta$.

\textbf{\textup{(The Experiment)}}  Consider a polynomial-time algorithm $\mathcal{N}$ that takes in an input state $\ket{\xi}$ promised to be inside $\mathcal{I}_{\Delta}$ and runs a polynomial-sized circuit (possibly with ancillas) that ends with a binary POVM $\{\mathsf{Acc}, \id - \mathsf{Acc}\}$. Furthermore, there is a set of $m=\poly(n)$ local Pauli operators $\texttt{\textup{Alg}}=\{A_1,\hdots,A_m\}$, such that $\mathcal{N}$ consists of gates that depend directly as a black-box on $\{A_1, \hdots, A_m\}$ and the individual local terms in $H$ (which for examples could be time evolutions of any of these operators, or measurements, etc.)

\textbf{\textup{(The ETH Ensemble)}}
 Let $\mc N'$ be an algorithm obtained by the following procedure: (1) iid sample a set of operators $\texttt{\textup{Alg}}' =\{A'_1,\hdots,A'_m\}$ from the operator ensemble defined in~\Cref{ETH}, (2) replace the components related to $\texttt{\textup{Alg}}$ in $\mc N$ by their $\texttt{\textup{Alg}}'$ versions \footnote{These $A'$ may require exponential circuit complexity, but we think of them as a black box, and these replacements only appears in our calculations.}. More precisely, within the window $\mc I_\Delta$, each $A'_i$ has the form

\begin{equation}
\label{eq:eth_assumption}
\bra{\nu_\alpha}A'_i\ket{\nu_\beta} = \mu^i_\alpha\delta_{\alpha,\beta}+ \frac{f_{\alpha,\beta}}{\sqrt{\Tr(\Pi_{\Delta})}}
g^i_{\alpha,\beta},
\end{equation}
where the coefficients and random variables are as described in~\Cref{ETH}. Note that the diagonal entries $\mu^i_\alpha$ are non-random and thus depend on $A_i$, so that distinct $A'_i$ is sampled from a distinct distribution. Furthermore, we assume~\Cref{prop:muvariation} holds for each $i$. The $g^i_{\alpha,\beta}$ are normal complex Gaussian variables sampled iid across different $i$.
 
 Then, we say $H$ and $\mc N$ satisfy the computational ETH within the energy window $\mc I_\Delta$ if it holds, with high probability over the randomness of the ETH ensemble, that
 \begin{align}
     \Tr( (\mathcal{N} (\xi) - \mathcal{N}'(\xi )  )\mathsf{Acc}    ) < \varepsilon_\mathrm{ETH} = \mathrm{negl}(n), \qquad \forall \text{ states } \xi \text{ of energy } \in \mc I_\Delta.
 \end{align}

\end{hypothesis}

The above computational version is an attempt to capture in the computer science language the common physicists' interpretations of ETH. Above, we allow the `experiment' to be beyond $\BQP$: it can be given a proof (untrusted, except for the fact that its energy is within $\mc I_\Delta$ which can be efficiently verified). This allows us to work in the context of $\QMA$. The $\varepsilon_\mathrm{ETH}$ is similar to the distinguishability advantage encountered in cryptography. Here we take it to be $\mathrm{negl}(n)$, but we actually only need $\varepsilon_\mathrm{ETH}$ to be a sufficiently small constant to guarantee the performance of our $\UQMA$ verifier in this section.

\subsection{Verifying a unique state in an energy window satisfying ETH}\label{sec:algo}

Given a local Hamiltonian $H$ on $n$ qubits, in this section, we give a protocol that uniquely accepts a specific state in a known energy window around $e_0$ satisfying certain properties that will be instantiated by ETH. 

The intuition behind the protocol comes from tests of entanglement via quantum expanders~\cite{Aharonov_AHL+2014_LocalTestsGlobal}.
More precisely, the EPR state $|\Omega\rangle=\frac{1}{\sqrt{D}}\sum_{i=1}^D |ii\rangle$ can be tested locally using unitaries of the form $U\otimes \bar{U}$.
We briefly explain the idea:
Clearly, $U\otimes \bar{U}|\Omega\rangle=|\Omega\rangle$ for any unitary $U\in SU(D)$ so any averaged operator over the actions $U_i\otimes \bar{U}_i$ for unitaries $U_i\in SU(D), i=1,\ldots, M$ satisfies  $\frac{1}{M}\sum_{i=1}^M U_i\otimes \bar{U}_i|\Omega\rangle=|\Omega\rangle$.
A quantum expander is roughly a collection of unitaries $U_i$, such that the second highest eigenvalue of  $\frac{1}{M}\sum_{i=1}^M U_i\otimes \bar{U}_i$ is bounded away from $1$.
Such quantum expander families exist with surprisingly small $M$ and constant gap between second highest eigenvalue and $1$ (see e.g.~\cite{gross2007quantum,hastings2007random}).
Another common approach, used e.g. for self-testing, is to pick the local Pauli operators, which results in an inverse polynomial gap between the second highest eigenvalue and $1$ (see e.g.~\cite{natarajan2016robust}).
The expectation value $\frac{1}{M}\sum_{i=1}^M\langle \psi|U_i\otimes \bar{U}_i|\psi\rangle $ can be efficiently estimated in many settings and suggests a simple test: If the above expectation is close to $1$, then $|\psi\rangle$ is close to $|\Omega\rangle$. 
Otherwise, $|\psi\rangle$ and $|\Omega\rangle$ are far from each other.
In particular, $|\Omega\rangle$ is uniquely accepted by this text.
A plausible strategy is therefore to use this principle to uniquely verify the maximally entangled state on two copies of the space spanned by the eigenstates in an energy window. 
It will turn out that we cannot actually guarantee that the uniquely accepted state is the maximally entangled state on these subspaces.
But for our purposes it will suffice that the Perron-Frobenius theorem guarantees the existence of a non-degenerate highest eigenvalue for an eigenstate with non-negative entries (see~\Cref{claim:eq_gap}).
Then, Merlin can simply send this state and Arthur can uniquely verify the ground state energy via a quantum expander.

Two obstacles prevent us from applying this naively: 1) A general unitary will move not respect the subspace spanned by eigenstates with energies in a window and 2) even if we can apply such unitaries we need to bound the gap of the averaging operator to guarantee uniqueness.
We solve 1) by considering unitaries of the form $e^{-\i \varepsilon A}$ for sufficiently small, but constant $\varepsilon>0$, which approximately preserve subspaces. 
We then use phase estimation to prepare measurements of the energy window and effectively apply the projector onto the respective subspaces (see~\Cref{clm:Qmass}).
In order to approach 2) we invoke the ETH: By choosing the $A_i$ to be local (think local Pauli matrices), their matrix entries in the energy window are indistinguishable from the Gaussians in Hypothesis~\ref{ETH}. 
We then use probabilistic techniques to show with high probability over the Gaussians $A_i'$ (see~\Cref{clm:tensorprodconcen}) that the operators $\mathbb{E}_i e^{-\i \varepsilon A_i}\otimes  e^{\i \varepsilon A_i}$ are sufficiently gapped (see~\Cref{claim:eq_gap}).
In particular, for the Gaussian $A'_i$ we obtain a gap and, consequently, unique verification.
The ETH then guarantees that any quantum algorithm using the correct $A_i$ will have the same behavior.

Finally, we assume there is a known energy $e_0$ and interval $\mc I_\Delta$ such that the above properties holds.

\begin{assumption}[Existence of ``good'' Energy]
\label{as:good_energy}
There exists a known energy $e_0 \in (0,1)$ and value $\Delta$ such that~\Cref{hypo:cETH} holds. Furthermore, $e_0$ can be represented by $O(\log n)$ bits.
\end{assumption}

\subsubsection{Phase estimation subroutine}
\label{sec:phase_estimation}
In this subsection we describe a way to approximately perform the projector-valued measurement $\{\Pi_\Delta, \id - \Pi_\Delta\}$ using quantum phase estimation (QPE). This will be used in the testing algorithm. We label the system register as $\reg R_\tsys$.
Introduce an ancilla register $\reg R_\taux$ of dimension $L$ for QPE, where $L \in \poly(n)$ is the resolution of the QPE algorithm and chosen to be large enough such that $e_0$ can be represented exactly (see \Cref{as:good_energy}). Consider the energy window $\mc I_\omega$, where $\omega = \Delta-1/\sqrt{L}$. By choice of $L$ the values $(e_0 \pm \omega/2)L$ are integers exactly represented by the $P$ register. Let $\Pi_{\regt \taux}$ be the projector onto the state
$\ket{\mu}_{\regt \taux}=\frac{1}{\sqrt{L}}\sum_{k=0}^{L-1}\ket{k}_{\regt \taux}$. 
\newcommand{\qpe}{\textup{QPE}}
Define the operator $U_{\qpe}$ which performs controlled
evolution by the Hamiltonian $H$, followed by the inverse quantum Fourier transform.
\begin{align*}
    U_{\qpe} &\dfn \Paren{\frac{1}{\sqrt L} \sum_{m,m'  =0}^{L-1} e^{2\pi \i \cdot \tfrac{- m m'}{L}} \Id \otimes
    \ketbratwo{m}{m'}_{\regt \taux}} \Paren{\sum_{k=0}^{L-1} e^{2\pi \i k H} \otimes \ketbra k_{\regt \taux}}\\
            &= \frac{1}{\sqrt{L}}\sum_{m,k}e^{2\pi\iota(H- m/L)k}\otimes\ketbratwo{m}{k}_{\regt \taux}.
\end{align*}
This allows us to define a projector $\Pi_{\regt \tsys  \regt \taux}$, which one should think of as a projector onto eigenstates in $\mc I_\omega$.
\begin{equation}
    \label{eq:sp_proj_def}
    \Pi_{\regt \tsys \regt \taux} = U^{-1}_{\qpe}\Paren{\id_S\otimes\sum_{m=(e_0-\omega/2)L}^{(e_0+\omega/2)L}\ketbra{m}}U_{\qpe} = \frac{1}{L} \sum_{k,k'=0}^{L-1}
    \sum_{m=(e_0-\omega/2)L}^{(e_0+\omega/2)L}e^{2\pi\iota(H- m/L)(k-k')}\otimes\ketbratwo{k'}{k}_{\regt \taux}.
\end{equation}

Note that 
\begin{eqnarray}
    \Pi_{\regt \taux}\Pi_{\regt \tsys \regt \taux}\Pi_{\regt \taux} &=& \sum_{m=(e_0-\omega/2)L}^{(e_0+\omega/2)L}\left(\frac{1}{L^2}\sum_{k,k'=0}^{L-1}e^{2\pi\iota(H- m/L)(k-k')}\right)\otimes
        \Pi_{\regt \taux}\nonumber\\
    &=& \sum_{m=(e_0-\omega/2)L}^{(e_0+\omega/2)L}\left(\text{sinc}_L(H-m/L)\right)^2\otimes \Pi_{\regt \taux}\nonumber\\
    &\triangleq & Q(H)\otimes \Pi_{\regt \taux},\label{eq:q_def}
\end{eqnarray}
with the function $\text{sinc}_L(x)$ defined as $\frac{\sin(\pi Lx)}{L\sin(\pi x)}$. This function peaks at $x=0$ and
then sharply goes down to a zero at $1/L$ and then goes up and down at the rate roughly $1/L$. Summing over $m$ would
ideally yield the projector $\Pi_\omega$. Rather than proving this, we show that $Q(H)$ behaves approximately like a
projector over a slightly larger space, given by the full window $\Pi_\Delta$. This is formalized in the following claim which states that $Q(H)$ is mostly supported on $\Pi_\Delta$.
\begin{restatable}{claim}{qmass}
\label{clm:Qmass} 
$\|Q(H) - \Pi_{\Delta}Q(H)\|\leq \frac{2}{\sqrt L}$.    
\end{restatable}

See \Cref{sec:q_proofs} for a proof of this claim. Now, using $Q(H)$ and the two Assumptions, we can define algorithms to test for the maximally entangled state
over the subspace $\Pi_\Delta$.

\subsubsection{The testing protocol}
\label{sec:mixing_from_eth}
We are now in a position to describe our protocol to test for energy subspaces. 
We first describe the algorithm formally:
We reuse some notations from the phase estimation subroutine. $\texttt{\textup{Alg}}$ is described in \Cref{alg:epr_test}, which acts on two copies of $S$, denoted as $\regt \tsys^{(1)}, \regt \tsys^{(2)}$. Note that the projectors $\Pi_{\regt \tsys \regt \taux}$ are implicitly parameterized by $e_0$.

\RestyleAlgo{boxruled}
\LinesNumbered
\begin{algorithm}[H]
    \setstretch{1.35}
    \caption{Energy subspace test $\texttt{\textup{Alg}}$\label{alg:epr_test}.\\
    \textbf{Input}: state $\ket{\psi}$ on the registers $\regt \tsys^{(1)},\regt \tsys^{(2)}$}
    
    Prepare and append $\ket{\mu}_{\regt \taux^{(1)}}\ket{\mu}_{\regt \taux^{(2)}} $ to $\ket{\psi}_{\regt \tsys^{(1)},\regt \tsys^{(2)}}$, where $\ket{\mu}=\frac{1}{\sqrt{L}}\sum_{k=0}^{L-1}\ket{k}$.
    
    Measure $\{\Pi_{\regt \tsys^{(1)}\regt \taux^{(1)}}, \id-\Pi_{\regt \tsys^{(1)}\regt \taux^{(1)}}\}$, $\{\Pi_{\regt \tsys^{(2)}\regt \taux^{(2)}}, \id-\Pi_{\regt \tsys^{(2)}\regt \taux^{(2)}}\}$  and reject if $\id-\Pi_{\regt \tsys^{(1)}\regt \taux^{(1)}}$ or $\id-\Pi_{\regt \tsys^{(2)}\regt \taux^{(2)}}$ is obtained. 
    
    Measure $\regt \taux^{(1)}$ in the basis $\{\Pi_{\regt \taux^{(1)}}, \id-\Pi_{\regt \taux^{(1)}}\}$, $\regt \taux^{(2)}$ in the basis $\{\Pi_{\regt \taux^{(2)}}, \id-\Pi_{\regt \taux^{(2)}}\}$  and reject if $\id-\Pi_{\regt \taux^{(1)}}$ or $\id-\Pi_{\regt \taux^{(2)}}$ is obtained.
    
    Prepare a register $\regt \tcon$ and define the state $\frac{1}{\sqrt{2\ell}}\sum_i(\ket{i,0}_{\regt \tcon}+\ket{i,1}_{\regt \tcon})$. Controlled on $\ketbra{i,0}_{\regt \tcon}$, apply $e^{\iota \eps A_i}$ on $\regt \tsys^{(1)}$ and $e^{-\iota \eps \bar{A}_i}$ on $\regt \tsys^{(2)}$. Controlled on $\ketbra{i,1}_T$, apply $e^{-\iota \eps A_i}$ on $\regt \tsys^{(1)}$ and $e^{\iota \eps \bar{A}_i}$ on $\regt \tsys^{(2)}$. 
    
    Repeat Steps 2,3. 
    
    Measure $\regt \tcon$ to see if the register is in $\frac{1}{\sqrt{2\ell}}\sum_i(\ket{i,0}_{\regt \tcon}+\ket{i,1}_{\regt \tcon})$. Reject if it is not. 
\end{algorithm}
Steps 1-3 use the phase estimation procedure explained in the previous section to effectively apply the approximate projector $Q(H)$ to the state $|\psi\rangle$ in registers $S_{1,2}$. 
Step 4 applies the elements of the quantum expander test conditioned on the register ${\regt \tcon}$ and Step 5 restricts to the low energy subspace again.
Finally, measuring the register ${\regt \tcon}$ yields a high acceptance probability if and only if the state $|\psi\rangle$ has high overlap with the space $\Pi_{\Delta}$ spanned by states in the low enery window and sufficiently invariant under $e^{-\i\varepsilon A_i}\otimes e^{\i\varepsilon A_i}$.
The resulting guarantees are summarized in the following theorem.
\begin{theorem}[Energy Subspace Test]
\label{thm:alg_correctness}
Fix a parameter $e_0$ and consider a Hamiltonian $H$. There exists an algorithm $\texttt{\textup{Alg}} = \texttt{\textup{Alg}}(H, e_0, L)$ (where $L=\poly(n)$ is a precision parameter) such that
\begin{enumerate}
    \item If $\lambda_0(H) > e_0 + \frac 3 L$ then $\texttt{\textup{Alg}}$ accepts with probability at most $\frac 1 3$.
    \label{item:no_low_energy}
    \item If $\lambda_0(H) \leq e_0$ and $e_0$ satisfies \Cref{as:good_energy} w.r.t. $H$, then there exists a unique state $\ket \Psi$ such that $\texttt{\textup{Alg}}$ accepts on input $\ket \Psi$ with probability at least $\frac 2 3$, and on states orthogonal to $\ket \Psi$, $\texttt{\textup{Alg}}$ accepts with probability at most $\frac 1 3$.
    \label{item:unique_witness}
\end{enumerate}
\end{theorem}

Given a state $\ket{\psi}_{\regt \tsys^{(1)}\regt \tsys^{(2)}}$, the acceptance probability is given by
\begin{align*}
\|\mathbb{E}_i&(\Pi_{\regt \taux^{(1)}}\otimes \Pi_{\regt \taux^{(2)}})(\Pi_{\regt \tsys^{(1)}\regt \taux^{(1)}}\otimes \Pi_{\regt \tsys^{(2)}\regt \taux^{(2)}})(\frac{1}{2}e^{\iota \eps A_i}\otimes \Pi_{\regt \taux^{(1)}} \otimes e^{-\iota \eps \bar{A}_i}\otimes \Pi_{\regt \taux^{(2)}}+\frac{1}{2}e^{-\iota \eps A_i}\otimes \Pi_{\regt \taux^{(1)}} \otimes e^{\iota \eps \bar{A}_i}\otimes \Pi_{\regt \taux^{(2)}})\\
&(\Pi_{\regt \tsys^{(1)}\regt \taux^{(1)}}\otimes \Pi_{\regt \tsys^{(2)}\regt \taux^{(2)}})\ket{\psi}_{\regt \tsys^{(1)},\regt \tsys^{(2)}}\ket{\mu}_{\regt \taux^{(1)}}\ket{\mu}_{\regt \taux^{(2)}}\|^2\\
&=\|\mathbb{E}_i(Q(H) \otimes Q(H))(\frac{1}{2}e^{\iota \eps A_i}\otimes e^{-\iota \eps \bar{A}_i}+\frac{1}{2}e^{-\iota \eps A_i}\otimes e^{\iota \eps \bar{A}_i})(Q(H) \otimes Q(H))\ket{\psi}\|^2\\
&\labelrel{=}{eq:taylor} \|(Q(H)\otimes Q(H))\big((1-\eps^2)\id\otimes \id  +\eps^2 \mathbb{E}_i \big(A_i \otimes \bar{A}_i \big) (Q(H) \otimes Q(H))\ket{\psi}\|^2+ O(\eps^3)
\end{align*}
where \eqref{eq:taylor} follows from a Taylor expanding $e^{\iota\eps A_i}$ and $e^{-\iota \eps A_i}$ and using $A^2_i=\id$. Note that the first order terms in $\eps$ cancel due to $A_i$ and $-A_i$.

We can use Claim \ref{clm:Qmass} to
write the last equation as 
\begin{align*}\|(Q(H)\otimes Q(H))\big(&(1-\eps^2)\Pi_\Delta\otimes \Pi_\Delta +\eps^2 \mathbb{E}_i \Pi_\Delta A_i\Pi_\Delta \otimes \Pi_\Delta \bar{A}_i\Pi_\Delta \big)(Q(H) \otimes Q(H))\ket{\psi}\|^2+ O(\eps^3+\frac{1}{\sqrt{L}}).
\end{align*}

Thus, the success probability of the protocol is captured by the operator 
\begin{equation}
    \label{eq:o_def}
    O_{\mathrm{succ}}:=(Q(H)\otimes Q(H))\big((1-\eps^2)\Pi_\Delta\otimes \Pi_\Delta +\eps^2 \mathbb{E}_i \Pi_\Delta A_i\Pi_\Delta \otimes \Pi_\Delta \bar{A}_i\Pi_\Delta \big)(Q(H) \otimes Q(H)).
\end{equation}
We'll separately analyze $O_\text{succ}$ for the cases corresponding to \Cref{item:no_low_energy} and \Cref{item:unique_witness}. To prove the former case, we show the following lemma,
\begin{lemma}
    \label{lem:no_low_energy}
    Let $H$ be a Hamiltonian such that $\lambda_0(H) > e_0 + \frac 3 L$. Then $\|O_\text{succ}\| \leq \frac 1 3$.
\end{lemma}
\begin{proof}
    From \Cref{eq:o_def}, we see that $O_\text{succ} = (Q(H) \otimes Q(H)) \tilde O (Q(H) \otimes Q(H))$, where $\|\tilde O\| \leq 1$. Therefore, we can bound $\|O_\text{succ}\| \leq \|Q(H) \otimes Q(H)\| = \max_{\ket \psi} \|Q(H) \ket \psi \|_2^2 = \max_{\alpha} \|Q(H) \ket{v_\alpha}\|_2^2$ where in the last equality, we restrict to eigenstates of $H$ since $Q(H)$ is diagonal in $H$'s eigenbasis. And $\lambda_\alpha - e_0 \geq \frac 3 L$ by assumption, so \Cref{item:almost_zero} of \Cref{lem:q_properties} tells us that $\|Q(H) \ket{v_\alpha}\|_2^2 \leq \frac 1 3$, which proves the claim.
\end{proof}

For the remainder of this section, we focus on the case where $\lambda(H) \leq e_0$ and thus \Cref{as:good_energy} holds. To show \Cref{item:unique_witness}, it suffices to show that $O_\text{succ}$ has a unique largest eigenstate and a large enough spectral gap. We will first simplify $O_\text{succ}$ using the following proposition.

\begin{prop}
    \label{thm:effectiveop}
Let $m\geq n^4$ and let $\mc G$ be the ETH ensemble as defined in \Cref{ETH}. With probability $1-e^{-O(m^{1/3})}$ over $\mc G$, 
we have 
$$\|O_{\mathrm{succ}} - (Q(H)\otimes Q(H)) M (Q(H)\otimes Q(H))\|\leq 2\eps^2f^7+ O(\eps^2/m^{1/3})+ \varepsilon_\mathrm{ETH},$$
where $M:=(1-\eps^2+\eps^2\mathbb{E}_i (\mu^i)^2)\Pi_\Delta\otimes \Pi_\Delta + \eps^2\mathbb{E}_G B\otimes \bar{B}$ and $B$ is drawn from $\mc G$.     
\end{prop}

\begin{proof}
Since the protocol is a polynomial time quantum computation, we can invoke the computational ETH~\Cref{hypo:cETH} to follow the RMT prescription while only incurring a small error $\varepsilon_\mathrm{ETH}$ in the final result. From \Cref{as:good_energy}, the term $\Pi_\Delta A_i \Pi_\Delta$ can be decomposed into the diagonal terms given by $\mu^i_\alpha$ and the remaining terms which follow the RMT prescription. For the latter term, define 
$$B_i :=\frac{1}{\sqrt{\trdd}}\sum_{\alpha,\beta}f_{\alpha, \beta}g^i_{\alpha, \beta}\ketbratwo{v_{\alpha}}{v_{\beta}}$$
For the diagonal terms, we recall from \Cref{prop:muvariation} that $\mu^i:=\max_\alpha \mu^i_{\alpha}$ is close to all $\mu^i_\alpha$ by $f^7$ and thus we approximate this term as $\mu^i \Pi_\Delta$. Thus,
$$\|\Pi_\Delta A_i \Pi_\Delta - \mu^i\Pi_\Delta - B_i\| \leq f^7, $$ 

The same decomposition can be done for $\Pi_\Delta \overline{A}_i \Pi_\Delta$, with $\overline B_i$ in place of $B_i$ and we get
\begin{align}
\mathbb{E}_i \Pi_\Delta A_i \Pi_\Delta &\otimes \Pi_\Delta \overline{A}_i \Pi_\Delta\nonumber\\
&= \mathbb{E}_i (\mu^i)^2\Pi_\Delta\otimes \Pi_\Delta + \Pi_\Delta\otimes\mathbb{E}_i\mu^i \overline{B}_i+ \mathbb{E}_i\mu^i B_i\otimes \Pi_\Delta + \mathbb{E}_i B_i\otimes \overline{B}_i + E,    
\end{align}
with $\|E\|\leq 2f^7$.
We will now repeatedly use the following lemma from \Cref{sec:prob_tools}.
\begin{restatable*}{lemma}{iidgauss}
    \label{clm:iidgauss}
Let $P$ be a random Hermitian matrix of dimension $D \times D$ with each entry being an independent complex Gaussian with mean $0$ and (non-identical) variance $\leq 1$. Then with probability $1-e^{-D}$, we have $\|P\|\leq 10 \sqrt{D}$.
\end{restatable*}
Now, note that 
\begin{align*}
   \mathbb{E}_i\mu^iB_i &= \frac{1}{\sqrt{\trdd}}\sum_{\alpha,\beta}f_{\alpha, \beta}(\mathbb{E}_i\mu^ig^i_{\alpha, \beta})\ketbratwo{v_{\alpha}}{v_{\beta}}\\
   &= \frac{\sqrt{\sum_i (\mu^i)^2}}{m\sqrt{\trdd}}\sum_{\alpha,\beta}f_{\alpha, \beta}\frac{\sum_i\mu^ig^i_{\alpha, \beta}}{\sqrt{\sum_i (\mu^i)^2}}\ketbratwo{v_{\alpha}}{v_{\beta}}:= \frac{\sqrt{\sum_i (\mu^i)^2}}{m} B'.
\end{align*}
Note that $\sqrt{\sum_i (\mu^i)^2}\leq \sqrt{m}$. By \Cref{fact:Gaussian_comb} we see that $\frac{1}{\sqrt{\sum_i (\mu^i)^2}} \sum_i \mu^i g^i_{\alpha,\beta} \sim \mathcal{CN}(0,1)$ and thus $\|B'\| \leq 10$ by \Cref{clm:iidgauss}. This means that with probability $1 - e^{-\trdd}$, $\|\E_i \mu^i B_i \| \leq \frac 1 {\sqrt m}$. The same argument works for $\mathbb{E}_i\mu^i \bar{B}_i$. Thus, 
$$\mathbb{E}_i \Pi_\Delta A_i \Pi_\Delta \otimes \Pi_\Delta \bar{A}_i \Pi_\Delta = \mathbb{E}_i (\mu^i)^2\Pi_\Delta\otimes \Pi_\Delta + \mathbb{E}_i B_i\otimes \bar{B}_i + E',$$
with $\|E'\|\leq 2f^6+ O( 1/\sqrt{m})$.

Finally, we will show the following claim.
\begin{claim}
\label{clm:tensorprodconcen}
Let $m\geq n^4$ and let $\mc G$ be the distribution over Gaussians as defined in \Cref{ETH}. With probability $1-e^{-\O(m^{1/3})}$ (over the randomness of ETH), we have $$\expec{E}_iB_i\otimes \bar{B}_i = \expec{E}_{\mc G} B\otimes \bar{B} + \O(m^{-1/3}),$$ where $B$ is drawn from $\mc G$.
\end{claim}
\begin{proof}
From \Cref{clm:iidgauss}, we know that each $B_i$ has norm $\|B_i\|\leq 10$ with probability $1-e^{-\O(\trdd)}$. Each $B_i$ was sampled IID from $\mc G$. Let $\mc G'$ be the distribution obtained when we condition $\mc G$ to Gaussians which ensure $\|B\|\leq 10$. Note that $\|\mc G-\mc G'\|_1\leq e^{-\O(\trdd)}$. Let $\mathcal B(\mc G)=\expec{E}_{\mc G}B\otimes \bar{B}$ and $\mathcal B(\mc G')=\expec{E}_{\mc G'}B\otimes \bar{B}$. Note that $$\|\mathcal B(\mc G)-\mathcal B(\mc G')\|\leq 2e^{-\O(\trdd)}\max_B\|B\|^2 \leq e^{-\O(\trdd)},$$since for any $B$, the operator norm is at most $\textsf{poly}(\trdd)$ due to the Greshgorin disc theorem.
Thus, for $t:=m^{-1/3}$, 
\begin{align*}
\Pr_{\mc G}[\|\mathbb{E}_iB_i\otimes \bar{B}_i - \mathcal B(\mc G)\|\geq t] &\leq \Pr_{\mc G'}[\|\mathbb{E}_iB_i\otimes \bar{B}_i - \mathcal B(\mc G')\|\geq t - e^{-\O(\trdd)}]\\
&\leq \Pr_{\mc G'}[\|\mathbb{E}_iB_i\otimes \bar{B}_i - \mathcal B(\mc G')\|\geq t/2]
\end{align*}
where the subscript refers to the distribution according to which $B_1,\ldots B_m$ are sampled. We will show in \Cref{claim:eq_gap} that $\|\mathcal B(\mc G)\|\leq 1$, which means for draws from $\mc G'$, $\|B_i\otimes \bar{B}_i-\mathcal B(\mc G')\|\leq 12$. 
Then, we apply the Matrix Chernoff bound \cite[Theorem 6.6.1]{Tropp15} to the matrices $B\otimes \bar{B}-\mathcal B(\mc G')$, drawn iid from $\mc G'$. The bound says that 
 $$\Pr_{\mc G'}(\|\mathbb{E}_iB_i\otimes \bar{B}_i - \mathcal B(\mc G')\|\geq t/2) \leq 2\trdd e^{-t^2m^2/(8v + 16mt)},$$
 where $$v=\|\sum_i\expec{E}_{\mc G'}(B_i\otimes \bar{B_i}-\mathcal B(\mc G'))^2\|=m\|\expec{E}_{\mc G'}(B\otimes \bar{B})^2-\mathcal B(\mc G')^2\|\leq m\expec{E}_{\mc G'}\|(B\otimes \bar{B}-\mathcal B(\mc G'))^2\|\leq 144m.$$
 Thus
 $$\Pr_{\mc G'}(\|\mathbb{E}_iB_i\otimes \bar{B}_i - \mathcal B(\mc G')\|\geq t/2) \leq 2\trdd e^{-t^2m^2/(9\times 2^7m + 16mt)} \leq 2^n e^{-\O(m^{1/3})} = e^{-\O(m^{1/3})}.$$
 This completes the proof.
\end{proof}

As a result, we conclude that with probability $1- e^{-\O(m^{1/3})}$, 
$$\mathbb{E}_i \Pi_\Delta A_i \Pi_\Delta \otimes \Pi_\Delta \bar{A}_i \Pi_\Delta = \mathbb{E}_i (\mu^i)^2\Pi_\Delta\otimes \Pi_\Delta + \mathbb{E}_G B\otimes \bar{B} + E'',$$
with $\|E''\|\leq 2f^7+ \O(1/m^{1/3})$. Finally, we add an error term $\varepsilon_\mathrm{ETH}$ from~\Cref{hypo:cETH} (for simplicity we assume the quantifier `with high probability' in that hypothesis statement is $1-e^{-\poly(n)} > 1-e^{-\O(m^{1/3})}$). This completes the proof of  \Cref{thm:effectiveop}.
\end{proof}

Next claim helps us understand the spectrum of $\expec{E}_G B\otimes \bar{B}$.

\begin{claim}
    \label{claim:eq_gap}
    The matrix $\expec{E}_{\mc G} B\otimes \bar{B}$
    has a unique leading eigenvector $\ket \Psi$ with eigenvalue $1\geq \lambda \geq f^2$ and all other eigenvalues are below $(1 -f^4)\lambda $. Furthermore, $\ket{\Psi}$ has all positive entries with ratios bounded between $\frac{1}{f^2}$ and $f^2$. Thus, the matrix $M$ has unique max eigenvector $\ket{\Psi}$ with eigenvalue $\lambda_0:=1-\eps^2 + \eps^2\expec{E}_i(\mu^i)^2+\eps^2\lambda$ and spectral gap is $\Delta\geq \eps^2f^4\lambda$. 

\end{claim}

\begin{proof}
We have
\begin{align*}
\expec{E}_{\mc G} B\otimes \bar{B}&=\frac{1}{\trdd}\sum_{\alpha,\beta, \alpha',\beta'}f_{\alpha, \beta}f_{\alpha',\beta'} \left(\mathbb{E}_{\mc G}
             g_{\alpha, \beta}g^*_{\alpha',\beta'}\right)\ketbratwo{v_{\alpha}}{v_{\beta}}\otimes\ketbratwo{v_{\alpha'}}{v_{\beta'}}\\
&\labelrel{=}{eq:explain} 
\frac{1}{\trdd}\sum_{\alpha,\beta}f_{\alpha, \beta}^2 \ket{v_\alpha,v_\alpha}\bra{v_\beta,v_\beta}
\end{align*}
where the summations above and onward are over $\alpha, \beta, \alpha', \beta'$ such that
$\lambda_\alpha, \lambda_\beta, \lambda_{\alpha'}, \lambda_{\beta'} \in \mc I_\Delta$ (rather than in $\mc
I_{\Delta_\text{RMT}}$). For \eqref{eq:explain} we use that $\mathbb{E} |g_{\alpha\beta}|^2=1$, $\mathbb{E} g_{\alpha\beta}^2=0$.
    
Since the matrix has positive entries, by the Perron-Frobenius theorem, its unique leading eigenvector is $\ket{\Psi}=\sum_\alpha x_\alpha\ket{v_\alpha, v_\alpha}$. Let $\lambda$ be the corresponding eigenvalue. In this proof, we think of $\alpha's$ as integers between $1,\ldots, N=\Tr(\Pi_\Delta)$, ordered such that $x_1$ has the largest value and $x_N$ the smallest. The fact that $\ket \Psi$ is an eigenvector implies
\begin{equation}
\label{eq:eig_def}
\frac{1}{\trdd}\sum_{\beta}f^2_{\alpha,\beta}x_\beta = \lambda  x_\alpha, \qquad \forall \alpha \in \mc I_\Delta.
\end{equation}
Choosing $\alpha = N$ in the above and solving for $\lambda$ yields,
    \begin{align*}
        \lambda=\frac{1}{\trdd}\sum_{\beta}f^2_{N,\beta}\frac{x_\beta}{x_{N}} \geq \frac{1}{\trdd}\sum_{\beta\in \mc I_\Delta}f^2_{N,\beta} \cdot 1 \geq f^2,
    \end{align*}
    \begin{comment}
    \qn{Sum is only over $\Delta$ window but not $\Delta_{RMT}$, so we need constant $C$ from~\Cref{as:smooth}. But I forgot why did we not simply choose $\Delta=\Delta_{RMT}$?}
    \yw{This is because if we define $\Delta = \Delta_\text{RMT}$, then $\ketbra{\Psi}$ will have entries $\alpha, \beta$ with $|\lambda_\alpha - \lambda_\beta| = \Delta_\text{RMT}$. But this would mean that $\ketbra{v_\alpha}{v_\beta}$ falls outside of the RMT window.}
    \end{comment}
    where the last inequality follows from our assumption on $f_{\alpha,\beta}$ from \Cref{hypo:cETH}.
    \begin{comment}
    We would like our lower bound to be in terms of $d$, not $d_N$. Recall that
    in \Cref{clm:tilde_m}, $d_\alpha$ was defined as
    \[
        (1-\eps^2)^2 + 2(1-\eps^2)\eps^2 \mu_\alpha^4
    \]
    Thus, $d_N = d - 2(1-\eps^2)\eps^2(\mu_\text{max}^4 - \mu_N^4)$. By \Cref{prop:muvariation}, $\mu_{\max} \leq \mu_N +
    \tfrac{f^2}{64}$ and thus,
    \begin{align*}
        d_N = d - 2(1-\eps^2)\eps^2(\mu_\text{max}^4 - \mu_N^4) &\geq d- 2(1-\eps^2)\eps^2\left(\left(\mu_N^4 +
            \frac{f^2}{64}\right)-\mu_N^4\right)\\
                                                                &\geq d- \eps^2\left(\left(\mu_N+\frac{f^2}{64}\right)^4 -
                                                                \mu_N^4\right) \tag{Asuming $\eps \leq \tfrac 1 {\sqrt
                                                                2}$}\\
                                                                &\geq d- \eps^2 \Paren{\Paren{\frac{f^2}{64}}^4 + \mu_N^4
                                                                -\mu_N^4}\tag{ansiangle inequality}\\
                                                                &\geq d- \frac{\eps^2 f^8}{64^4}\numberthis
    \end{align*}
    Loosely lower bounding the above as $d_N \geq d - \eps^2 f^2$, we find that
    \begin{equation}
        \label{eq:lambda_diff}
        \lambda\geq d+\eps^2f^2/2
    \end{equation}
    \end{comment}
    We also
    have the upper bound $1$ on $\lambda$ using the Greshgorin disc theorem.
    
    Now, taking ratios of \Cref{eq:eig_def} for $\alpha = 1$ and $\alpha = N$ yields
    \begin{equation}
        \label{eq:x_bound}
    \frac{\sum_{\beta}f^2_{1,\beta}x_\beta}{\sum_{\beta}f^2_{N,\beta}x_\beta} = \frac{ x_{1}}{x_{N}} \implies
        \frac{1}{f^2} \geq  \frac{x_{1}}{ x_{N}}\geq 1.
    \end{equation}
    where we've used that $f \leq f_{\alpha,\beta} \leq 1$. This shows the desired properties of the leading eigenvector $\ket \Psi$, except the gap.
 
    Take an eigenstate orthogonal to $\ket \Psi$, denoting it $\ket \Phi = \sum_\alpha y_\alpha \ket{v_\alpha, \overline{v_\alpha}}$. Note that we can assume WLOG that $y_a$'s are real since $M_f$ is a matrix with real entries in the $\ket{v_\alpha}$ basis. Define the sets $S_+ = \{y_\alpha \st y_\alpha > 0\}$ and $S_- = \{y_\alpha \st y_\alpha \leq 0\}$, extending to
    sets of tuples $S_{(-,-)}, S_{(+,+)}, S_{(-,+)}, S_{(+,-)}$ in the natural way. The eigenvalue corresponding to $\ket \Phi$ can be written
    as,
    \begin{align}
        \trdd\bra{\Phi}\expec{E}_G B\otimes \bar{B} \ket{\Phi} &= \sum_{\alpha, \beta} f^2_{\alpha,\beta} y_\alpha y_\beta\nonumber\\
        &=\sum_{\alpha\neq\beta \in S_{(+,+)}} f^2_{\alpha,\beta} y_\alpha y_\beta + \sum_{\alpha\neq
        \beta \in S_{(-,-)}} f^2_{\alpha,\beta}y_\alpha y_\beta\nonumber\\
        &\qquad - \sum_{\alpha,\beta \in S_{(+,-)}}f^2_{\alpha,\beta}y_\alpha |y_\beta| -
        \sum_{\alpha,\beta \in S_{(-,+)}} f^2_{\alpha,\beta} |y_\alpha| y_\beta\nonumber\\
        \label{eq:first_norm}
        &=\Sigma_{(+,+)} + \Sigma_{(-,-)} - \Sigma_{(+,-)} - \Sigma_{(-,+)}
    \end{align}
    where we've defined $\Sigma_{(\cdot, \cdot)}$ to refer to the (positive) sums over $S_{(\cdot, \cdot)}$. Notice that $\bra{\Phi}\expec{E}_G B\otimes \bar{B} \ket{\Phi} $ includes the negative terms $-\Sigma_{(+,-)} - \Sigma_{(-,+)}$, whereas the leading eigenvector $\ket \Psi$ has non-negative entries (by Perron-Frobenius) and thus no such terms. This is the key aspect which will manifest the spectral gap.
    
    We claim that the ratios between the terms in \Cref{eq:first_norm} can be multiplicatively bounded. In particular,
    \begin{align}
       \frac{\Sigma_{(+,+)}}{\Sigma_{(-,+)}} = \frac{\sum_{\alpha\neq\beta \in S_{+,+}} f^2_{\alpha,\beta}
           y_\alpha y_\beta}{ \sum_{\alpha,\beta \in S_{-,+}} f^2_{\alpha,\beta} |y_\alpha|
           y_\beta}&\labelrel{\leq}{rel:m_bound}  \frac{ \sum_{\alpha\neq\beta \in S_{+,+}}y_\alpha y_\beta}{
           f^2\sum_{\alpha,\beta \in S_{-,+}} |y_\alpha| y_\beta}\nonumber\\
        &\leq \frac{ (\sum_{\alpha \in S_{+}}y_\alpha)^2}{ f^2(\sum_{\alpha \in
        S_{-}} |y_\alpha|)(\sum_{\alpha \in
        S_{+}} y_\alpha)}\nonumber\\
        \label{eq:ratio_bound_mid}
        &= \frac{ \sum_{\alpha \in S_{+}}y_\alpha}{ f^2(\sum_{\alpha \in
        S_{-}} |y_\alpha|)}
    \end{align}
    where \eqref{rel:m_bound} follows by $f^2 \leq f^2_{\alpha, \beta} \leq
    1$. By the assumption that $\braket{\Phi}{\Psi} = 0$, we have that $\sum_{\alpha \in S_+} y_\alpha
    x_\alpha = \sum_{\alpha \in S_-} |y_\alpha|x_\alpha$ and using the bounds $x_N \leq x_\alpha \leq x_1$ we obtain 
\begin{align}
    \label{eq:neg_pos_bounds}
    x_N \sum_{\alpha \in S_+} y_\alpha \leq \sum_{\alpha \in S_+} y_\alpha x_\alpha &\leq x_1 \sum_{\alpha \in S_-} |y_\alpha|.
\end{align}
Combining with \Cref{eq:ratio_bound_mid} yields 
 \begin{equation}
    \label{eq:ratio_bound}
     \frac{\Sigma_{(+,+)}}{\Sigma_{(-,+)}} = \frac{ \sum_{\alpha \in S_{+}}y_\alpha}{ f^2(\sum_{\alpha \in
     S_{-}} |y_\alpha|)} \leq \frac{x_1}{f^2x_N} \leq \frac{1}{f^4}.
\end{equation}
 Similarly, we can bound
$\frac{\Sigma_{(+,-)}}{\Sigma_{(-,+)}}$ as in \Cref{eq:ratio_bound_mid} but instead we get
\[
    \frac{\Sigma_{(+,-)}}{\Sigma_{(-,+)}} \leq \frac 1 {f^2} \frac{\sum_{\alpha \in S_+} y_{\alpha}}{\sum_{\alpha \in
    S_-} |y_\alpha|}\cdot \frac{\sum_{\alpha \in S_-} |y_{\alpha}|}{\sum_{\alpha \in S_+} y_\alpha} = \frac 1 {f^2} \leq
    \frac 1 {f^4}.
\]
The bounds on $\frac{\Sigma_{(-,-)}}{\Sigma_{(+,-)}}$ and $\frac{\Sigma_{(-,+)}}{\Sigma_{(+,-)}}$ follow identically.
This allows us to conclude that,
\begin{align}
    &2\Sigma_{(-,+)} \geq f^4 \left(\Sigma_{(+,+)} + \Sigma_{(+,-)}\right),\\
    &2\Sigma_{(-,+)} \geq f^4 \left(\Sigma_{(-,-)} + \Sigma_{(-,+)}\right),
\end{align}
and therefore
\begin{align}
    \Sigma_{(+,+)} + \Sigma_{(-,-)} - \Sigma_{(+,-)} - \Sigma_{(-,+)} &= \Sigma_{(+,+)} + \Sigma_{(-,-)} + \Sigma_{(+,-)}
        + \Sigma_{(-,+)}\nonumber\\
                                                                      &\Qquad- 2\left(\Sigma_{(+,-)} +
                                                                      \Sigma_{(-,+)}\right)\nonumber\\
                                                                      &\leq \Sigma_{(+,+)} + \Sigma_{(-,-)} +
                                                                      \Sigma_{(+,-)} + \Sigma_{(-,+)}\nonumber\\
                                                                      &\Qquad- f^4 \left(\Sigma_{(+,+)} +
                                                                      \Sigma_{(-,-)} + \Sigma_{(+,-)}
                                                                      +-\Sigma_{(-,+)}\right)\nonumber\\
                                                                      &= \Paren{1 - f^4 }\Paren{\Sigma_{(+,+)} +
                                                                      \Sigma_{(-,-)} + \Sigma_{(+,-)} + \Sigma_{(-,+)}}.
\end{align}
         
This leads to an upper bound on \Cref{eq:first_norm} as,
\begingroup
\allowdisplaybreaks
 \begin{align*}
        \bra{\Phi}\expec{E}_G B\otimes \bar{B} \ket{\Phi} &\leq \frac{1}{\trdd}\left(1-f^4\right)\left(\Sigma_{(+,+)} + \Sigma_{(-,-)} + \Sigma_{(+,-)} + \Sigma_{(-,+)}\right)\\
        &= \left(1-f^4\right)\left(\sum_{\alpha, \beta} (\expec{E}_G B\otimes \bar{B})_{\alpha,\beta} |y_\alpha| |y_\beta|\right)\\
        &\labelrel{\leq}{rel:lambda} (1-f^4)\lambda \numberthis
    \end{align*}  
\endgroup
where in \eqref{rel:lambda} we've used that the leading eigenvector has positive entries and thus upper bounds the quadratic form $\sum_{\alpha, \beta} (\expec{E}_G B\otimes \bar{B})_{\alpha,\beta} |y_\alpha| |y_\beta|$.
\end{proof}

Finally, we move to the spectral gap of $(Q \otimes Q) M (Q \otimes Q)$ (with $Q:=Q(H)$), which captures the success probability of the protocol, as shown in \Cref{thm:effectiveop}.  The following claim, whose proof is deferred to Appendix, shows that this conjugation does not change $\ket{\Psi}$ much. 
\begin{restatable}{claim}{qdoesntmatter}
\label{clm:Qdoesntmatter}
    We have $ \bra{\Psi}Q\otimes Q\ket{\Psi} \geq 1-\frac 1 {\sqrt L} - (1-C')\frac{6}{f^4} = 1-\frac 1 {\sqrt L} - 6f^{21}$  .
\end{restatable}
\begin{theorem}
    \label{thm:maintechnical}
Let $m=n^4$ and $\eps=\Theta(f^7)$. Assume that~\Cref{as:good_energy} holds and $\lambda_0(H)\leq e_0$.
With probability $1-e^{-O(m^{1/3})}$ over $G$, the following holds:  the protocol \Cref{alg:epr_test} accepts a unique vector with probability $p=1-O(f^{14})$ and accepts any orthogonal vector with probability $p-\Omega(f^{20}) + \varepsilon_\mathrm{ETH}+ 1/\poly(n)$.
\end{theorem}

\begin{proof}
Lets write $M=\lambda_0\ketbra{\Psi} + (\lambda_0-\Delta)N$, for some matrix $N$ orthogonal to $\ket{\Psi}$ with norm $\|N\|\leq 1$.
We have
\begin{equation}
\label{eq:raw_m}
    \bra{\Psi}(Q\otimes Q)M(Q\otimes Q)\ket{\Psi}=\lambda_0 |\bra{\Psi}Q\otimes Q \ket{\Psi}|^2 + (\lambda_0-\Delta)\bra{\Psi}Q\otimes Q N Q\otimes
    Q\ket{\Psi}.
\end{equation}
To bound the second term, note that
\begin{align*}
    |\bra{\Psi}(Q\otimes Q) N (Q\otimes Q)\ket{\Psi}| &= | \bra \Psi N \ket \Psi + \bra \Psi (\Id - Q \otimes Q) N \ket \Psi + \bra \Psi (Q \otimes Q) N (\Id - Q \otimes Q) \ket \Psi|\\
    &\leq \bra \Psi N \ket \Psi + |\bra \Psi (\Id - Q \otimes Q) N \ket \Psi| + |\bra \Psi (Q \otimes Q) N (\Id - Q \otimes Q) \ket \Psi|\\
    &\leq |\bra \Psi (\Id - Q \otimes Q) N \ket \Psi| + |\bra \Psi (Q \otimes Q) N (\Id - Q \otimes Q) \ket \Psi|\tag{Since $N \perp \ket \psi$}\\
    &\leq 2\|(\Id - Q \otimes Q) \ket \psi\|_2\\
    &\leq 2 \bra \Psi (\Id - Q \otimes Q)^2 \ket \Psi \\
    &\leq 2 \bra \Psi (\Id - Q \otimes Q) \ket \Psi\tag{Using that $Q \otimes Q \preceq \Id$}\\
    &\leq 2\Paren{\frac{1}{\sqrt L} + 6f^{21}}\tag{By \Cref{clm:Qdoesntmatter}}\\
    &\leq \frac 2 {\sqrt L} + 12f^{21}.
\end{align*}
Since $\lambda_0 - \Delta \leq 1$, combining with \Cref{eq:raw_m} yields
\begin{equation}
    \label{eq:psi_m_value}
    \bra{\Psi}(Q\otimes Q)M(Q\otimes Q)\ket{\Psi}\geq \lambda_0 -12f^{21}-\frac{2}{\sqrt L} = \tilde \lambda_0.
\end{equation}
Since $M\preceq \lambda_0 \id$ and $Q\preceq \id$, we also have that 
the largest eigenvalue of $(Q\otimes Q)M(Q\otimes Q)$ is at most $\lambda_0$. Thus, its largest eigenvalue is in the interval $[\tilde \lambda_0, \lambda_0]$. To show unique largest eigenvector and spectral gap,
consider any eigenvector vector $\ket{\phi}$ of $(Q\otimes Q)M(Q\otimes Q)$ such that $\bra{\phi}(Q\otimes Q)M(Q\otimes Q)\ket{\phi}\geq \lambda_0 - \delta$,
for $\delta = \frac{\Delta}{100}$ (where $\Delta$ is the spectral gap of $M$). For
such a $\ket \phi$, we have
\begingroup
\allowdisplaybreaks
\begin{alignat*}{2}
    &\bra{\phi}(Q\otimes Q)M(Q\otimes Q)\ket{\phi} \geq \lambda_0 - \delta&&\\[2ex]
    &\implies \lambda_0 |\bra{\phi}Q\otimes Q\ket{\Psi}|^2+ (\lambda_0-\Delta) \bra{\phi}Q\otimes Q N Q\otimes Q\ket{\phi} &&\geq
        \lambda_0 - \delta\\[2ex]
    &\implies \lambda_0|\bra{\phi}Q\otimes Q\ket{\Psi}|^2+ (\lambda_0-\Delta) \bra{\phi}Q\otimes Q (\id-\ketbra{\Psi}) Q\otimes
        Q\ket{\phi}&&\geq \lambda_0 - \delta\tag{$N \preceq \id - \ketbra \Psi$}\\[2ex]
    &\implies \Delta|\bra{\phi}Q\otimes Q\ket{\Psi}|^2+ (\lambda_0 - \Delta) \|Q\otimes
        Q\ket{\phi}\|^2 &&\geq \lambda_0 - \delta\tag{Regrouping}\\[2ex]
    &\implies \Delta|\bra{\phi}Q\otimes Q\ket{\Psi}|^2+ (\lambda_0 - \Delta)&&\geq \lambda_0 -
    \delta\tag{$\|Q\| \leq 1$}\\[2ex]
    &\implies |\bra{\phi}Q\otimes Q\ket{\Psi}|^2 \geq \frac{\Delta-\delta}{\Delta}= 0.99.&&
\end{alignat*}
\endgroup
Using \Cref{clm:Qdoesntmatter}, we find that $\ket{\phi}$ has 0.97 overlap with $\ket{\Psi}$. This forces that there
be a unique eigenvector above $\lambda_0 - \delta$ (two eigenvectors would not have 0.97 overlap with the same vector).
This eigenvector must have eigenvalue at least $\tilde \lambda_0$ by \Cref{eq:psi_m_value}. Thus the
spectral gap of $(Q\otimes Q)M(Q\otimes Q)$ is at least
\[
    \tilde \lambda_0 - (\lambda_0 - \delta) = \delta - 12f^{21}- \frac{2}{\sqrt L}  = \frac{\eps^2f^4\lambda}{100} - 12f^{21}- \frac{2}{\sqrt L} \geq \frac{\eps^2f^6}{100} - 12f^{21}- \frac{2}{\sqrt L} .
\] 

From \Cref{thm:effectiveop}, we thus find that a unique vector is accepted by the protocol with probability $$\frac{\eps^2
f^6}{200} - 2\eps^2f^7- 12f^{21}- O(\eps^3) - O(m^{-1/3} + \frac{1}{\sqrt L})$$ higher than other vectors orthogonal to this vector. Setting
$\eps= \Theta(f^{7})$, $m=n^4=poly(n)$, $L=poly(n)$, and using $C' \geq 1 - f^{25}$, this gives a success probability gap of $\Omega(f^{20})-\frac{1}{\poly(n)}$. Recalling from \Cref{claim:eq_gap} that $\lambda_0 = 1 - \eps^2 + \eps^2\E_i (\mu^i)^2 + \eps^2 \lambda \geq 1 - \eps^2$, we find that the unique vector is accepted with probability
$$\lambda_0 - 12 f^{21}- \frac{2}{\sqrt L}\geq 1-O(f^{14}).$$ Finally, we add the potential error $\varepsilon_\mathrm{ETH}$ from~\Cref{hypo:cETH}. This proves the theorem.
\end{proof}
Since the above procedure yields a promise gap of $\Omega(f^{20}) - \varepsilon_\mathrm{ETH}- \tfrac 1 {\cpoly(n)}$
(assuming $\varepsilon_\mathrm{ETH} \ll f^{20}$), we can amplify the gap to be any constant via the in-place
amplification technique of \cite{Marriott_MW2005_QuantumArthurMerlinGames}. Not only is the promise gap improved, but this amplification scheme is also spectral gap preserving as shown in~\cite{Jain_JKK+2011_PowerUniqueQuantum}.
\begin{corollary}
    \label{cor:amplified}
    For any choice of constants $1 > a > b > 0$, the in-place amplification of
    \cite{Marriott_MW2005_QuantumArthurMerlinGames} applied to \Cref{alg:epr_test} yields a protocol $\texttt{\textup{Alg}}$ which accepts a unique vector with probability $\geq a$ and accepts all orthogonal states with probability $\leq b$. 
\end{corollary}

Finally, \Cref{thm:alg_correctness} follows by \Cref{lem:no_low_energy} and \Cref{cor:amplified}.

\subsection{Connections between $\UQMA$ vs $\QMA$ and thermalization}\label{subsec:connections}

Finally, we discuss implications and potential interpretations of our main results. 
Our~\Cref{thm:alg_correctness} implies that local Hamiltonians satisfying ETH (as formalized in \Cref{as:good_energy}) at a slightly subextensive energy, e.g. $e_0= \|H\|^{0.999}$, is in $\UQMA$.
However, deciding whether the ground energy is $\leq a \|H\|^{1-\alpha}$ or $\geq b \|H\|^{1-\alpha}$ is $\QMA$-hard for constants $a < b$, and any constant $\alpha >0$~\cite{Chen20}.
Put differently, the non-equivalence of $\UQMA$ and $\QMA$ would imply that $\QMA$-hard local Hamiltonians must violate ETH at near-extensive energy regimes. 
It is worth noting that this violation is \emph{adversarially robust}, meaning that the perturbed Hamitlonian $H + \delta H$ continue to violate ETH for arbitrary perturbation $\delta V$ of near-extensive strength, e.g. $\|\delta H\| < \|H\|^{0.999}$.

The previous paragraph can be extended to extensive energy scales under the quantum PCP conjecture.

\begin{conjecture}[Quantum PCP conjecture, \cite{AAV13}] There are constants $c<d$ such that deciding if the ground energy of a general $n$-qubit local Hamiltonian $H$ is (yes case) $\leq c\|H\|$ or (no case) $\geq d\|H\|$ is $\QMA$-complete.
\end{conjecture}

Let $H$ be a ``yes'' case  local Hamiltonian witnessing the $\QMA$-hardness as in the quantum PCP conjecture. We can now argue that $H$ violates ETH and robustly so against adversarial extensive perturbations, assuming $\UQMA$ and $\QMA$ are not equivalent.
First, we have the following claim:

\begin{claim}[ETH Hamiltonians in $\UQMA$]
\label{thm:ETHinUQMA}
    Fix constants $a < b$. Let $H$ be a local Hamiltonian whose ground energy is either $\leq a \|H\|$ or $\geq b \|H\|$. Furthermore, in the ``yes'' case, it is promised that $H$ satisfies ETH (as formalized in \Cref{as:good_energy}) for an energy value $e_0 < a\|H\|$. Then there is a $\UQMA$ protocol that decides if the ground energy is $\leq a \|H\|$ or $\geq b \|H\|$. 
\end{claim}

\begin{proof}
    Take the verifier to be \Cref{alg:epr_test}. The fact that this algorithm corresponds to a $\UQMA$ verifier follows directly from \Cref{thm:alg_correctness}.
\end{proof}

Let $\delta=\frac{d-c}{4}$ and consider any local Hamiltonian $V$ with $\|V\|\leq \|H\|$. Suppose $\UQMA \neq \QMA$, we will show that $H+\delta V$ violates ETH  at all energy values $e_0 < c\|H\|$. Suppose for contradiction that $H+\delta V$ satisfies ETH (as formalized in \Cref{as:good_energy}). If the ground energy of $H$ is $\leq c\|H\|$, then the ground energy of $H+\delta V$ is $\leq (c+\delta)\|H\|$ and if the ground energy of $H$ is $\geq d\|H\|$, then the ground energy of $H+\delta V$ is $\geq (d-\delta)\|H\|$. Then we invoke Theorem \Cref{thm:ETHinUQMA} with constants $(c+\delta) < d-\delta$ and get a $\UQMA$ protocol to decide which is the case. But the question of determining whether the ground energy of $H+\delta V$ is $\leq (c+\delta)\|H\|$ or $\geq (d-\delta)\|H\|$ is $\QMA$-hard because this tells us if the ground energy of $H$ is $\leq c\|H\|$ or $\geq d\|H\|$. This violates the conjecture that $\UQMA \neq \QMA$.

The intuition that qPCP Hamiltonians violate an assumption like ETH is consistent with the widely held belief that any proof of the qPCP conjecture will require extensive use of quantum codes.
Such codes and their corresponding code Hamiltonians presumably have a lot of algebraic structure that might prevent thermalization.

\begin{remark}
    The results and discussions in this section can be adapted to Hamiltonians satisfying a `weak' ETH. Instead of assuming a known good energy window like in~\Cref{as:good_energy}, a weak ETH Hamiltonian is such that a randomly chosen energy window satisfies~\Cref{hypo:cETH} with high probability. In that case, applying~\Cref{alg:epr_test} to a randomly chosen energy window we find that weak ETH Hamiltonians are in $\BPP^{\UQMA}$. As a result, the non-containment $\QMA \not\subseteq \BQP^{\UQMA} $ would imply $\QMA$-hard Hamiltonians violate weak ETH.
\end{remark}

A similar observation can be made about many-body localization (MBL) of $\QMA$-hard Hamiltonians. 
Many-body localization~\cite{nandkishore2015many, abanin2019colloquium} is a phenomenon that is experimentally and numerically observed in quantum systems that violate ETH. 
Some common characterizations of MBL include~\cite{huang2019instability, huse2014phenomenology} slow growth of entanglement in Hamiltonian dynamics, area law for most eigenstates, and quasi-local integrals of motion. 
Defining and justifying these features is not the scope of this work, but from a complexity-theoretic viewpoint it is perhaps not surprising that a Hamiltonian with the above features are not expected to be $\QMA$-hard. 
As an example, the quasi-local integrals of motion property roughly states that there is a quasi-local unitary $U$ that diagonalizes $H$. A quasilocal unitary is a quantum circuit of depth $D=O(\log n)$ which can consist of non-local gates \footnote{The choice $D=O(\log n)$ is made since we are not restricting the geometry, In 1D systems, the depth would be $O(n)$ ~\cite{ehrenberg2022simulation}. Regardless, the argument actually stays intact for any $D=\poly(n)$.}. 
However, the constraint is that if any gate $g$ acts on $k$ sites, then $\|g-\id\|= e^{-\Omega(k)}$. Thus, we can replace the unitary $U$ with another unitary $U'$, in which any gate acting on more than $O(\log n)$ qubits is set to identity. The error incurred in this process is $\frac{1}{\poly(n)}$ and we can use the Solovay-Kitaev theorem to compile $U'$ into a polynomial sized circuit. This implies MBL Hamiltonians are in $\QCMA$: we expect a prover to provide us with the description of $U$ and the ground state of the diagonalized Hamiltonian as a bitstring, and then run the circuit to verify the energy. Hence, assuming $\QCMA \neq \QMA$~\cite{Aaronson_AK2007_QuantumClassicalProofs}, $\QMA$-hard Hamiltonians are not MBL. 
Similar to the ETH discussion, this fact can be made robust against adversarial extensive perturbations under the quantum PCP conjecture as follows. Suppose there is a $V$ such that $H+\delta V$ satisfies MBL. Then there is a polynomial-sized circuit $U$ such that $U^{\dagger}(H+\delta V)U$ diagonal and a classical string $\ket{x}$ that is the ground state in this basis. The prover sends the classical descriptions of $V, U, x$ and then the verifier runs a quantum computer to prepare the ground state and decide if the ground energy of $H+\delta V$ is $\leq (c+\delta)\|H\|$ or $\geq (d-\delta)\|H\|$. This in turn allows to decide if the ground energy of $H$ is $\leq c\|H\|$ or $\geq d\|H\|$, giving a $\QCMA$ protocol for the original $H$.
Therefore, under the assumption $\QCMA\neq \QMA$, MBL Hamiltonians cannot be expected to be $\QMA$-hard.

The connection to $\UQMA$ is more subtle. 
It is shown in~\cite{Aharonov_ABBS2022_PursuitUniquenessExtending} that $\UQCMA$ is contained in $ \QCMA$ under randomized reductions, thus $\QCMA \subseteq \BQP^{\UQMA}$ and the above discussion of MBL also holds under the assumption $ \QMA \not\subseteq \BQP^\UQMA$.
Then, $\QMA$-hard Hamiltonians are robustly neither ETH or MBL even up to adversarial perturbations of near-extensive strengths (extensive if quantum PCP conjecture is true as shown above).

Finally, we remark that the adversarially robust violation of ETH and MBL described above is impossible on lattice Hamiltonians. 
Observe that on a $D$-dimensional lattice of length $L$, there are local perturbations of total strength $O(\varepsilon^D L^D)$ that disentangle the Hamiltonian into non-interacting $D$-dimensional cubical patches of size $\frac{1}{\varepsilon}\times \hdots \times \frac{1}{\varepsilon}$. Such non-interacting Hamiltonian is trivially localized.

We find it interesting that a seemingly unrelated complexity theory assumption can lead to such a conclusion, especially that the question of ``whether there is only one phase transition, or could there possibly be some sort of intermediate phase that is neither fully localized nor fully thermal?''~\cite{nandkishore2015many} is still a central and active research question in quantum condensed matter physics.